\documentclass{article}
\usepackage{arxiv}

\usepackage[utf8]{inputenc} 
\usepackage[T1]{fontenc}    
\usepackage{hyperref}       
\usepackage{url}            
\usepackage{booktabs}       
\usepackage{amsfonts}       
\usepackage{nicefrac}       
\usepackage{microtype}      
\usepackage{graphicx}
\usepackage{enumerate}
\usepackage{animate}
\usepackage{natbib}
\usepackage[title]{appendix}
\usepackage{epstopdf}
\usepackage[flushleft]{threeparttable}
\usepackage{multirow}

\usepackage{amssymb}
\usepackage{booktabs}
\usepackage{mathtools}
\usepackage{makecell}
\usepackage[amsmath,thmmarks]{ntheorem}
\usepackage{multicol}
\usepackage{tikz} 
\usetikzlibrary{calc,backgrounds}
\usetikzlibrary{shapes.geometric, arrows, arrows.meta}
\usetikzlibrary{positioning}
\usepackage{bbm}
\usepackage{xcolor}
\usepackage[caption=false,subrefformat=parens,labelformat=parens]{subfig}
\usepackage{algorithmic}
\usepackage{algorithm}
\ifpdf
  \DeclareGraphicsExtensions{.eps,.pdf,.png,.jpg}
\else
  \DeclareGraphicsExtensions{.eps}
\fi

\DeclareMathOperator*{\argmax}{\arg\!\max}
\newcommand{\eqdef}{\overset{\mathrm{def}}{=\joinrel=}}

\newtheorem{theorem}{Theorem}[section]

\newtheorem{lemma}[theorem]{Lemma}
\newtheorem{proposition}[theorem]{Proposition}

\theoremstyle{plain}
\newtheorem{innercustomass}{Assumption}
\newenvironment{customass}[1]
  {\renewcommand\theinnercustomass{#1}\innercustomass}
  {\endinnercustomass}
  
\theoremstyle{nonumberplain}
\theorembodyfont{\normalfont}
\theoremsymbol{\ensuremath{\square}}
\newtheorem{proof}{Proof}
\makeatother

\title{Linked Gaussian Process Emulation for Systems of \\ Computer Models using Mat{\'e}rn Kernels and Adaptive Design}

\author{
  Deyu Ming\thanks{Corresponding author: \texttt{deyu.ming.16@ucl.ac.uk}.}\\
  Department of Statistical Science\\
  University College London\\
  London, UK \\
   \And
 Serge Guillas \\
  Department of Statistical Science\\
  University College London\\
  London, UK \\
}

\begin{document}
\maketitle

\begin{abstract}
The state-of-the-art linked Gaussian process offers a way to build analytical emulators for systems of computer models. We generalize the closed form expressions for the linked Gaussian process under the squared exponential kernel to a class of Mat{\'e}rn kernels, that are essential in advanced applications. An iterative procedure to construct linked Gaussian processes as surrogate models for any feed-forward systems of computer models is presented and illustrated on a feed-back coupled satellite system. We also introduce an adaptive design algorithm that could increase the approximation accuracy of linked Gaussian process surrogates with reduced computational costs on running expensive computer systems, by allocating runs and refining emulators of individual sub-models based on their heterogeneous functional complexity.
\end{abstract}


\section{Introduction}
\label{sec:intro}

Systems of computer models constitute the new frontier of many scientific and engineering simulations. These can be multi-physics systems of computer simulators such as coupled tsunami simulators with earthquake and landslide sources~\citep{salmanidou2017statistical,ulrich2019coupled}, coupled multi-physics model of the human heart~\citep{santiago2018fully}, and multi-disciplinary systems such as automotive and aerospace systems~\citep{fazeley2016multi,kodiyalam2004multidisciplinary,zhao2018multidisciplinary}. Other examples include climate models where climate variability arises from atmospheric, oceanic, land, and cryospheric processes and their coupled interactions~\citep{hawkins2016irreducible,kay2015community}, or highly multi-disciplinary future biodiversity models~\citep{thuiller2019uncertainty} using combinations of species distribution models, dispersal strategies, climate models, and representative concentration pathways. The number and complexity of computer models involved can hinder the analysis of such systems. For instance, the engineering design optimization of an aerospace system typically requires hundreds of thousands of system evaluations. When the system has feed-backs across computer models, the number of simulations becomes computationally prohibitive~\citep{chaudhuri2017multifidelity}. Therefore, building and using a surrogate model is crucial: the system outputs can be predicted at little computational cost, and subsequent sensitivity analysis, uncertainty propagation or inverse modeling can be conducted in a computationally efficient manner. 

Gaussian Stochastic process or Gaussian process (GaSP or GP) emulators have gained popularity as surrogate models of systems of computer models in fields including environmental science, biology and geophysics because of their attractive statistical properties. However, many studies~\citep{jandarov2014emulating,johnstone2016uncertainty,salmanidou2017statistical,simpson2001kriging,tagade2013gaussian} construct global GaSP emulators (named as composite emulators hereinafter) of such systems based on global inputs and outputs without consideration of system structures. One major drawback of such a structural ignorance is that designing experiments can be expensive because system structures may induce high non-linearity between global inputs and outputs~\citep{sanson2019systems}. Furthermore, runs of the whole system are required to produce new training points, even though the overall functional complexity global inputs and outputs originates from a few computer models. This pitfall is particularly undesirable because modern engineering and physical systems can include multiple computer models.

To overcome the disadvantages of the composite emulator, one could construct the surrogate for a system of computer models by integrating GaSP emulators of individual computer models. The idea of integrating GaSP emulators has been explored by~\cite{sanson2019systems} in a feed-forward system, but only using the Monte Carlo simulation to approximate the predictive mean and variance of the system output. The Monte Carlo method suffers from a low convergence rate and heavy computational cost, especially when the number of layers in a system is high~\citep{rainforth2018nesting} and the number of new input positions to be evaluated is large, making it prohibitive for complex systems. 

Recently, \cite{marque2019efficient} presents a nested emulator that works for systems of two computer models, while~\cite{kyzyurova2018coupling} derived a more flexible emulator, called linked GaSP, for two-layered feed-forward systems of computer models in analytical form (i.e., closed form expressions for mean and variance of the predicted output of the system at an unexplored input position). However, both of the work are carried out under the assumption that every computer model in the system is represented by a GaSP with a product of one-dimenional squared exponential kernels over different input dimensions. Indeed, the squared exponential kernel has been criticized for its over-smoothness~\citep{stein1999interpolation} and associated ill-conditioned problem~\citep{dalbey2013efficient,gu2018robust}. Thus, the generalization of the kernel assumption is necessary. In this study, we generalize the linked GaSP to a class of Mat{\'e}rn kernels for its wider applications in practice. We also demonstrate an iterative procedure, by which the linked GaSP can be constructed for any feed-forward computer systems.

Careful experimental design is important to construct efficient linked GaSP surrogate under limited computational resources. Poor designs can cause inaccurate linked GaSP with excessive designing cost, and numerical instabilities in training GaSP emulators of individual computer models. Particularly, the linked GaSP is more prone to the latter issue than the composite emulator because the design (e.g., the Latin hypercube design) of the global input can produce poor designs for GaSP emulators of internal computer models. Therefore, we discuss in the work several possible design strategies that can be used for linked GaSP emulation, and introduce an adaptive design algorithm that has the potential to effectively enhance the approximation accuracy of the linked GaSP with improved designs and reduced overall simulation cost.

The remainder of the manuscript is organized as follows. In Section~\ref{sec:review}, we review basics of the GaSP emulator and the linked GaSP. The extension of linked GaSP to Mat{\'e}rn kernels is then formulated with a synthetic experiment in Section~\ref{sec:extension}. An iterative procedure to produce linked GaSPs for any feed-forward computer systems is demonstrated with a feed-back coupled satellite model in Section~\ref{sec:feedforward}. In Section~\ref{sec:design}, we introduce an adaptive design strategy for the linked GaSP emulation and discuss its advantages and disadvantages in relating to other alternative designs. Limitations of the linked GaSP are discussed in Section~\ref{sec:discussion}. We conclude in Section~\ref{sec:conclusion}. Key closed form expressions for the linked GaSP under different kernels and associated proofs are contained in the appendices and supplementary materials, respectively.

\section{Review of GaSP Emulator and Linked GaSP}
\label{sec:review}
In this section, we first give a brief description of GaSP emulators for individual computer models in a computer system. Then the linked GaSP introduced in~\cite{kyzyurova2018coupling} is reviewed. Note that we present the linked GaSP using our own notations for the benefit of deriving kernel extensions in Section~\ref{sec:extension}.

\subsection{GaSP Emulators for Individual Computer Models}
\label{sec:gpmodel}
The GaSP emulator of a computer model considered in this work is itself a collection of GaSP emulators, approximating the functional dependence between the inputs of the computer model and its one-dimensional outputs. Each 1-D output emulator is constructed independently without the consideration of cross-output dependence, as in~\cite{gu2016parallel,kyzyurova2018coupling}.

Let $\mathbf{X}\in\mathbb{R}^p$ be a $p$-dimensional vector of inputs of a computer model and $Y(\mathbf{X})$ be the corresponding scalar-valued output. Then, given $m$ sets of inputs $\{\mathbf{X}_1,\dots,\mathbf{X}_m\}\,$, the GaSP model is defined by
\begin{equation*}
Y(\mathbf{X}_i)=t(\mathbf{X}_i,\,\mathbf{b})+\varepsilon_i,\quad i=1,\dots,m\,	
\end{equation*}
where $t(\mathbf{X}_i,\,\mathbf{b})=\mathbf{h}(\mathbf{X}_i)^\top\mathbf{b}$ is the trend function with $q$ basis functions $\mathbf{h}(\mathbf{X}_i)=[h_1(\mathbf{X}_i),\dots,h_q(\mathbf{X}_i)]^\top$ and $\mathbf{b}=[b_1,\dots,b_q]^\top\,$; $(\varepsilon_1,\dots,\varepsilon_m)^\top\sim\mathcal{N}(\mathbf{0},\,\sigma^2\mathbf{R})$ with $ij$-th element of the correlation matrix $\mathbf{R}$ given by $
R_{ij}=c(\mathbf{X}_i,\,\mathbf{X}_j)+\eta\mathbbm{1}_{\{\mathbf{X}_i=\mathbf{X}_j\}}$, where $c(\cdot,\cdot)$ is a given kernel function; $\eta$ is the nugget term; and $\mathbbm{1}_{\{\cdot\}}$ is the indicator function.

The specification of the kernel function $c(\cdot,\cdot)$ plays an important role in GaSP emulation as it characterizes the sample paths of a GaSP model~\citep{stein1999interpolation}. In this study we consider the kernel function with the following multiplicative form:
\begin{equation*}
c(\mathbf{X}_i,\,\mathbf{X}_j)=\prod_{k=1}^p c_k(X_{ik},\,X_{jk}),
\end{equation*}
where $c_k(\cdot,\cdot)$ is a one-dimensional kernel function for the $k$-th input dimension. Popular candidates for $c_k(\cdot,\cdot)$ are summarized in Table~\ref{tab:kernel}. In Section~\ref{sec:extension}, we will show that the linked GaSP is applicable to all these aforementioned choices. In the proofs of the supplement, we also consider the additive form of $c(\cdot,\cdot)$.

\begin{table}[htbp] 
\footnotesize{
\caption{Choices of $c_k(\cdot,\cdot)$. $\gamma_k>0$ is the range parameter for the $k$-th input dimension.}
\label{tab:kernel}	
\begin{center}
\begin{tabular}{llcll}
\toprule
\textbf{Exponential} & \multicolumn{4}{l}{$c_k(\cdot,\cdot)=\exp\left\{-\frac{|X_{ik}-X_{jk}|}{\gamma_k}\right\}$}\\
\addlinespace[0.3cm]
\textbf{\makecell[l]{Squared\\ Exponential}} & \multicolumn{4}{l}{$c_k(\cdot,\cdot)=\exp\left\{-\frac{(X_{ik}-X_{jk})^2}{\gamma^2_k}\right\}$}\\
\addlinespace[0.3cm]
\textbf{Mat{\'e}rn-1.5} & \multicolumn{4}{l}{$c_k(\cdot,\cdot)=\left(1+\frac{\sqrt{3}|X_{ik}-X_{jk}|}{\gamma_k}\right)\exp\left\{-\frac{\sqrt{3}|X_{ik}-X_{jk}|}{\gamma_k}\right\}$}\\
\addlinespace[0.3cm]
\textbf{Mat{\'e}rn-2.5} & \multicolumn{4}{l}{$c_k(\cdot,\cdot)=\left(1+\frac{\sqrt{5}|X_{ik}-X_{jk}|}{\gamma_k}+\frac{5(X_{ik}-X_{jk})^2}{3\gamma^2_k}\right)\exp\left\{-\frac{\sqrt{5}|X_{ik}-X_{jk}|}{\gamma_k}\right\}$}\\
\bottomrule
\end{tabular}
\end{center}
}
\end{table}

Assume that the GaSP model parameters $\sigma^2$, $\eta$ and $\boldsymbol{\gamma}=(\gamma_1,\dots,\gamma_p)^\top$ are known but $\mathbf{b}$ is a random vector that has a Gaussian distribution with mean $\mathbf{b}_0$ and variance $\tau^2\mathbf{V}_0$. Then, given $m$ inputs $\mathbf{x}^\mathcal{T}=(\mathbf{x}^\mathcal{T}_1,\dots,\mathbf{x}^\mathcal{T}_m)^\top$ and the corresponding outputs $\mathbf{y}^\mathcal{T}=(y_1^\mathcal{T},\dots,y_m^\mathcal{T})^\top$, the GaSP emulator of the computer model is defined by the predictive distribution of $Y(\mathbf{x}_0)$ (i.e., conditional distribution of $Y(\mathbf{x}_0)$ given $\mathbf{y}^\mathcal{T}$) at a new input position $\mathbf{x}_0$~\citep{santner2003design}, which is
\begin{equation}
\label{eq:1emu}
Y(\mathbf{x}_0)|\mathbf{y}^\mathcal{T}\sim\mathcal{N}(\mu_0(\mathbf{x}_0),\,\sigma^2_0(\mathbf{x}_0))
\end{equation}
with
\begin{align}
\label{eq:krigingmeannor}
\mu_0(\mathbf{x}_0)&=\mathbf{h}(\mathbf{x}_0)^{\top}\widehat{\mathbf{b}}+\mathbf{r}(\mathbf{x}_0)^\top\mathbf{R}^{-1}\left(\mathbf{y}^\mathcal{T}-\mathbf{H}(\mathbf{x}^\mathcal{T})\widehat{\mathbf{b}}\right)\\
\label{eq:krigingvarnor}
\sigma^2_0(\mathbf{x}_0)&={\sigma^2}\Big[1+\eta-\mathbf{r}(\mathbf{x}_0)^\top\mathbf{R}^{-1}\mathbf{r}(\mathbf{x}_0)+\left(\mathbf{h}(\mathbf{x}_0)-\mathbf{H}(\mathbf{x}^\mathcal{T})^\top\mathbf{R}^{-1}\mathbf{r}(\mathbf{x}_0)\right)^\top\\
    &\times\left(\mathbf{H}(\mathbf{x}^\mathcal{T})^\top\mathbf{R}^{-1}\mathbf{H}(\mathbf{x}^\mathcal{T})+\frac{\sigma^2}{\tau^2}\mathbf{V}_0^{-1}\right)^{-1}\left(\mathbf{h}(\mathbf{x}_0)-\mathbf{H}(\mathbf{x}^\mathcal{T})^\top\mathbf{R}^{-1}\mathbf{r}(\mathbf{x}_0)\right)\Big],\nonumber
\end{align}
where $\mathbf{r}(\mathbf{x}_0)=[c(\mathbf{x}_0,\mathbf{x}_1^\mathcal{T}),\dots,c(\mathbf{x}_0,\mathbf{x}_m^\mathcal{T})]^\top$, $\mathbf{H}(\mathbf{x}^\mathcal{T})=[\mathbf{h}(\mathbf{x}_1^\mathcal{T}),\dots,\mathbf{h}(\mathbf{x}_m^\mathcal{T})]^\top$ and 
\begin{equation*}
\widehat{\mathbf{b}}\eqdef\left(\mathbf{H}(\mathbf{x}^\mathcal{T})^\top\mathbf{R}^{-1}\mathbf{H}(\mathbf{x}^\mathcal{T})+\frac{\sigma^2}{\tau^2}\mathbf{V}_0^{-1}\right)^{-1}\left(\mathbf{H}(\mathbf{x}^\mathcal{T})^\top\mathbf{R}^{-1}\mathbf{y}^{\mathcal{T}}+\frac{\sigma^2}{\tau^2}\mathbf{V}_0^{-1}\mathbf{b}_0\right)	.
\end{equation*}
Let $\tau^2\rightarrow\infty$ (i.e., the Gaussian distribution of $\mathbf{b}$ gets more and more non-informative), then all terms associated with $\mathbf{b}_0$ and $\mathbf{V}_0$ in equation~\eqref{eq:krigingmeannor} and~\eqref{eq:krigingvarnor} become increasingly insignificant and thus we obtain the GaSP emulator defined by the predictive distribution of $Y(\mathbf{x}_0)$ with its mean and variance given by
\begin{align}
\label{eq:krigingmean}
\mu_0(\mathbf{x}_0)=&\mathbf{h}(\mathbf{x}_0)^{\top}\widehat{\mathbf{b}}+\mathbf{r}(\mathbf{x}_0)^\top\mathbf{R}^{-1}\left(\mathbf{y}^\mathcal{T}-\mathbf{H}(\mathbf{x}^\mathcal{T})\widehat{\mathbf{b}}\right)\\
\label{eq:krigingvar}
\sigma^2_0(\mathbf{x}_0)=&\sigma^2\Big[1+\eta-\mathbf{r}(\mathbf{x}_0)^\top\mathbf{R}^{-1}\mathbf{r}(\mathbf{x}_0)+\left(\mathbf{h}(\mathbf{x}_0)-\mathbf{H}(\mathbf{x}^\mathcal{T})^\top\mathbf{R}^{-1}\mathbf{r}(\mathbf{x}_0)\right)^\top\\
    &\quad\times\left(\mathbf{H}(\mathbf{x}^\mathcal{T})^\top\mathbf{R}^{-1}\mathbf{H}(\mathbf{x}^\mathcal{T})\right)^{-1}\left(\mathbf{h}(\mathbf{x}_0)-\mathbf{H}(\mathbf{x}^\mathcal{T})^\top\mathbf{R}^{-1}\mathbf{r}(\mathbf{x}_0)\right)\Big]\nonumber
\end{align}
with $\widehat{\mathbf{b}}\eqdef\left[\mathbf{H}(\mathbf{x}^\mathcal{T})^\top\mathbf{R}^{-1}\mathbf{H}(\mathbf{x}^\mathcal{T})\right]^{-1}\mathbf{H}(\mathbf{x}^\mathcal{T})^\top\mathbf{R}^{-1}\mathbf{y}^\mathcal{T}$, where $\mu_0(\mathbf{x}_0)$ and $\sigma^2_0(\mathbf{x}_0)$ match the best linear unbiased predictor (BLUP) of $Y(\mathbf{x}_0)$ and its mean squared error~\citep{stein1999interpolation}. In the remainder of the study we use the predictive distribution with mean and variance given in equation~\eqref{eq:krigingmean} and~\eqref{eq:krigingvar} as the GaSP emulator of a computer model. Note that the GaSP model parameters $\sigma^2$, $\eta$ and $\boldsymbol{\gamma}=(\gamma_1,\dots,\gamma_p)^\top$ in equation~\eqref{eq:krigingmean} and~\eqref{eq:krigingvar} are typically unknown and need to be estimated. One may estimate these parameters by solving the objective function
\begin{equation*}
    (\widehat{\eta},\,\widehat{\boldsymbol{\gamma}})=\underset{\eta,\,\boldsymbol{\gamma}}{\argmax}\,\mathcal{L}(\widehat{\sigma^2},\,\eta,\,\boldsymbol{\gamma}),
\end{equation*}
where 
\begin{multline*}
\mathcal{L}(\widehat{\sigma^2},\,\eta,\,\boldsymbol{\gamma})=\frac{|\mathbf{R}|^{-\frac{1}{2}}|\mathbf{H}(\mathbf{x}^\mathcal{T})^\top\mathbf{R}^{-1}\mathbf{H}(\mathbf{x}^\mathcal{T})|^{-\frac{1}{2}}}{(2\pi\widehat{\sigma^2})^{\frac{m-q}{2}}}\\
\times\exp\left\{-\frac{1}{{2\widehat{\sigma^2}}}\left(\mathbf{y}^\mathcal{T}-\mathbf{H}(\mathbf{x}^\mathcal{T})\widehat{\mathbf{b}}\right)^\top\mathbf{R}^{-1}\left(\mathbf{y}^\mathcal{T}-\mathbf{H}(\mathbf{x}^\mathcal{T})\widehat{\mathbf{b}}\right)\right\},  
\end{multline*}
is the marginal likelihood obtained by integrating out $\mathbf{b}$ from the full likelihood function $\mathcal{L}(\mathbf{b},\,\sigma^2,\,\eta,\,\boldsymbol{\gamma})$ and have $\sigma^2$ replaced by its maximum likelihood estimator
\begin{equation}
\label{eq:sigma2map}
\widehat{\sigma^2}=\frac{1}{m-q}\left(\mathbf{y}^\mathcal{T}-\mathbf{H}(\mathbf{x}^\mathcal{T})\widehat{\mathbf{b}}\right)^\top\mathbf{R}^{-1}\left(\mathbf{y}^\mathcal{T}-\mathbf{H}(\mathbf{x}^\mathcal{T})\widehat{\mathbf{b}}\right)
\end{equation}
with
$\widehat{\mathbf{b}}\eqdef\left[\mathbf{H}(\mathbf{x}^\mathcal{T})^\top\mathbf{R}^{-1}\mathbf{H}(\mathbf{x}^\mathcal{T})\right]^{-1}\mathbf{H}(\mathbf{x}^\mathcal{T})^\top\mathbf{R}^{-1}\mathbf{y}^\mathcal{T}$. Alternatively, the maximum a posterior (MAP) method is a more robust estimation technique~\citep{gu2018robust}. It maximizes the marginal posterior mode with respect to the objective function
\begin{equation}
\label{eq:map}
    (\widehat{\eta},\,\widehat{\boldsymbol{\gamma}})=\underset{\eta,\,\boldsymbol{\gamma}}{\argmax}\,\mathcal{L}(\widehat{\sigma^2},\,\eta,\,\boldsymbol{\gamma})\pi(\eta,\,\boldsymbol{\gamma}),
\end{equation}
where $\pi(\eta,\,\boldsymbol{\gamma})$ is the reference prior, see~\cite{gu2018robust} for different choices and parameterizations. 

After the estimates of $\sigma^2$, $\eta$ and $\boldsymbol{\gamma}$ are obtained, they are plugged into the predictive distribution mean~\eqref{eq:krigingmean} and variance~\eqref{eq:krigingvar}, forming the empirical GaSP emulator of a computer model. In the remainder of the study, all GaSP models of individual computer models are estimated using the MAP method via the~\texttt{R} package~\texttt{RobustGaSP}. Note that~\texttt{RobustGaSP} in fact estimates $\eta$ and $\boldsymbol{\gamma}$ with the marginal likelihood obtained by integrating out both $\mathbf{b}$ and $\sigma^2$. However, as demonstrated in~\cite{andrianakis2009parameter} the estimates of $\eta$ and $\boldsymbol{\gamma}$ are not influenced by the integration of $\sigma^2$. As a result, we can implement~\texttt{RobustGaSP} to obtain the estimates of $\eta$ and $\boldsymbol{\gamma}$ produced by the discussed MAP method and then have them plugged in equation~\eqref{eq:sigma2map} to obtain the estimate of $\sigma^2$.  

\subsection{Linked GaSP}
Consider a two-layered system of computer models, where the computer models in the first layer produce collectively $d$-dimensional output that feeds into a computer model in the second layer. Let $\mathbf{W}=[W_1(\mathbf{x}_1),\dots,W_d(\mathbf{x}_d)]^\top$ be the collection of the $d$-dimensional output produced by $d$ GaSP emulators $\widehat{f}_1,\dots,\widehat{f}_d$ of computer models in the first layer given the input positions $\mathbf{x}_1,\dots,\mathbf{x}_d$. Denote $\widehat{g}$ as the GaSP emulator of the computer model $g$ in the second layer, producing $Y(\mathbf{W},\mathbf{z})$ that approximates a scalar-valued output of $g$ at inputs $\mathbf{W}$ from $\widehat{f}_1,\dots,\widehat{f}_d$ and exogenous inputs $\mathbf{z}=(z_1,\dots,z_p)^\top$. Then the emulation of the two-layered system aims to link GaSP emulators connected as shown in Figure~\ref{fig:2systememu}.

\begin{figure}[htbp]
\centering
\scalebox{1}{
\begin{tikzpicture}[shorten >=1pt,->,draw=black!50, node distance=4cm]
    \tikzstyle{every pin edge}=[<-,shorten <=1pt]
    \tikzstyle{neuron}=[circle,fill=black!25,minimum size=20pt,inner sep=0pt]
    \tikzstyle{layer1}=[neuron, fill=green!50];
    \tikzstyle{layer2}=[neuron, fill=red!50];
    \tikzstyle{layer3}=[neuron, fill=blue!50];
    \tikzstyle{annot} = [text width=4em, text centered]

    \node[layer1, pin=left:$\mathbf{x}_1$] (I-0) at (0,-0.45) {$\widehat{f}_1$};
    \node[layer1, pin=left:$\mathbf{x}_2$] (I-1) at (0,-1.25) {$\widehat{f}_2$};
    \node[layer1, pin=left:$\mathbf{x}_d$] (I-2) at (0,-2.5) {$\widehat{f}_d$};

    \path[yshift=-1.375cm]
            node[layer2, pin={[pin edge={->}]right:$Y$},pin={[pin edge={<-}]below:$\mathbf{z}$}] (H-1) at (4cm,0 cm) {$\widehat{g}$};

            \path [draw] (I-0) -- (H-1) node[font=\small,pos=0.35,fill=white,align=left,sloped] {$W_1$};
            \path [draw] (I-1) -- (H-1) node[font=\small,pos=0.35,fill=white,align=left,sloped] {$W_2$};
            \path [draw] (I-2) -- (H-1) node[font=\small,pos=0.35,fill=white,align=left,sloped] {$W_d$};
            \path (I-1) -- (I-2) node [black, midway, sloped] {$\dots$};
            \path (I-1) -- (I-2) node [black, midway, sloped,transform canvas={xshift=-11.5mm}] {$\dots$};
            \path (I-1) -- (I-2) node [black, midway, sloped,transform canvas={xshift=14.5mm,yshift=2mm}] {$\dots$};
\end{tikzpicture}}
\caption{The connections of emulators to be linked for emulating a two-layered computer system. $\widehat{f}_1,\,\widehat{f}_2\dots,\widehat{f}_d$ are one-dimensional emulators approximating $d$ outputs from computer models in the first layer; $\widehat{g}$ is a one-dimensional GaSP emulator approximating a scalar-valued output of the computer model $g$ in the second layer of the system.}
\label{fig:2systememu}
\end{figure}
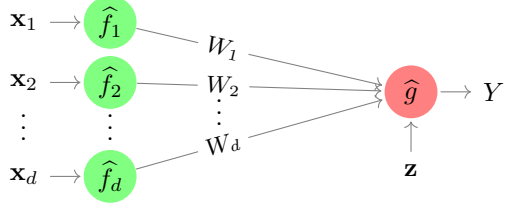

Perhaps the most straightforward way to build an emulator of the system is to obtain the predictive distribution of $Y(\mathbf{x}_1,\dots,\mathbf{x}_d,\,\mathbf{z})$, given the global inputs $\mathbf{x}_1,\dots,\mathbf{x}_d$ and $\mathbf{z}$. This predictive distribution, named as linked emulator by~\cite{kyzyurova2018coupling}, is naturally defined by the probability density function
\begin{equation}
\label{eq:intproblem}
p(y|\mathbf{x}_1,\dots,\mathbf{x}_d,\,\mathbf{z})=\int_{\mathbf{w}} p(y|\mathbf{w},\mathbf{z})\,p(\mathbf{w}|\mathbf{x}_1,\dots,\mathbf{x}_d)\,\mathrm{d}\mathbf{w},    
\end{equation}
where $\mathbf{w}=(w_1,\dots,w_d)^\top$. However, $p(y|\mathbf{x}_1,\dots,\mathbf{x}_d,\,\mathbf{z})$ is neither analytically tractable nor Gaussian in general. One might compute the integral in equation~\eqref{eq:intproblem} numerically or simply generate realizations of $Y(\mathbf{x}_1,\dots,\mathbf{x}_d,\,\mathbf{z})$ by sampling sequentially from Gaussian densities $p(y|\mathbf{w},\mathbf{z})$ and $p(\mathbf{w}|\mathbf{x}_1,\dots,\mathbf{x}_d)$,  and then use the resulting density or sampled realizations as the linked emulator. However, such approaches are computationally expensive and can soon become prohibitive for many uncertainty analysis as the dimensions of $\mathbf{x}_{i=1,\dots,d}$ and $\mathbf{w}$ increase. Fortunately, \cite{kyzyurova2018coupling} show that under some mild conditions, the mean and variance of the linked emulator can be calculated analytically, and its Gaussian approximation, called linked GaSP, is a Gaussian distribution with matching mean and variance. One of the key conditions that~\cite{kyzyurova2018coupling} make for the closed form mean and variance of the inked emulator is that the GaSP emulator $\widehat{g}$ is constructed under the squared exponential kernel. However, it is well known that the squared exponential kernel can have computational difficulties both in theory and practice~\citep{stein1999interpolation,dalbey2013efficient,gu2018robust}, limiting broader applications of the linked GaSP. In Section~\ref{sec:extension}, we relax this kernel limitation and show that there exists closed form expressions for the mean and variance of the linked emulator under a class of Mat\'{e}rn kernels.

\section{Generalization of Linked GaSP to Mat\'{e}rn kernels}
\label{sec:extension}
Assume that the GaSP emulator $\widehat{g}$ is built with $m$ training points $\mathbf{w}^{\mathcal{T}}=(\mathbf{w}_1^{\mathcal{T}},\dots,\mathbf{w}_m^{\mathcal{T}})^\top$, $\mathbf{z}^{\mathcal{T}}=(\mathbf{z}_1^{\mathcal{T}},\dots,\mathbf{z}_m^{\mathcal{T}})^\top$ and $\mathbf{y}^{\mathcal{T}}=(y_1^{\mathcal{T}},\dots,y_m^{\mathcal{T}})^\top$, where $\mathbf{w}_i^{\mathcal{T}}=(w_{i1}^{\mathcal{T}},\dots,w_{id}^{\mathcal{T}})^\top$ and $\mathbf{z}_i^{\mathcal{T}}=(z_{i1}^{\mathcal{T}},\dots,z_{ip}^{\mathcal{T}})^\top$ for all $i=1,\dots,m$. Then under the following two assumptions:
\begin{customass}{1}\label{ass:1}
The trend function $t(\mathbf{W},\,\mathbf{z},\,\boldsymbol{\theta},\,\boldsymbol{\beta})$ in the GaSP model for the computer model $g$ is specified by
$t(\mathbf{W},\,\mathbf{z},\,\boldsymbol{\theta},\,\boldsymbol{\beta})=\mathbf{W}^\top\boldsymbol{\theta}+\mathbf{h}(\mathbf{z})^\top\boldsymbol{\beta}$,
where 
\begin{itemize}
    \item $\boldsymbol{\theta}=(\theta_1,\dots,\theta_d)^\top$ and $\boldsymbol{\beta}=(\beta_1,\dots,\beta_{q})^\top$;
    \item $\mathbf{h}(\mathbf{z})=[h_1(\mathbf{z}),\dots,h_{q}(\mathbf{z})]^\top$ are basis functions of $\mathbf{z}$;
\end{itemize}
\end{customass}
\begin{customass}{2}\label{ass:2}
$W_k(\mathbf{x}_k)\stackrel{ind}{\sim}\mathcal{N}(\mu_k(\mathbf{x}_k),\,\sigma^2_k(\mathbf{x}_k))$ for $k=1,\dots,d$,
\end{customass}
we can derive in closed form the mean and variance of linked emulator subject to the choice of 1-D kernel functions used in GaSP emulator $\widehat{g}$.

\begin{theorem}
\label{thm:main}
Under Assumption~\ref{ass:1} and~\ref{ass:2}, the output $Y(\mathbf{x}_1,\dots,\mathbf{x}_d,\mathbf{z})$ of the linked emulator at the input positions $\mathbf{x}_1,\dots,\mathbf{x}_d$ and $\mathbf{z}$ has analytical mean $\mu_L$ and variance $\sigma^2_L$ given by
\begin{align}
\label{eq:intmu}
\mu_L=&\boldsymbol{\mu}^\top\widehat{\boldsymbol{\theta}}+\mathbf{h}(\mathbf{z})^\top\widehat{\boldsymbol{\beta}}+\mathbf{I}^\top\mathbf{A},\\  
\label{eq:intvar}
\sigma^2_L=&\mathbf{A}^\top\left(\mathbf{J}-\mathbf{I}\mathbf{I}^\top\right)\mathbf{A}+2\widehat{\boldsymbol{\theta}}^\top\left(\mathbf{B}-\boldsymbol{\mu}\mathbf{I}^\top\right)\mathbf{A}+\mathrm{tr}\left\{\widehat{\boldsymbol{\theta}}\widehat{\boldsymbol{\theta}}^\top\boldsymbol{\Omega}\right\}\\
&+\sigma^2\,\left(1+\eta+\mathrm{tr}\left\{\mathbf{Q}\mathbf{J}\right\}+\mathbf{G}^\top\mathbf{C}\mathbf{G}+\mathrm{tr}\left\{\mathbf{C}\mathbf{P}-2\mathbf{C}\widetilde{\mathbf{H}}^\top\mathbf{R}^{-1}\mathbf{K}\right\}\right),\nonumber
\end{align}
where
\begin{itemize}
\item $\boldsymbol{\mu}=[\mu_1(\mathbf{x}_1),\dots,\mu_d(\mathbf{x}_d)]^\top$ and $\left[\widehat{\boldsymbol{\theta}}^\top,\,\widehat{\boldsymbol{\beta}}^\top\right]^\top\eqdef\left(\widetilde{\mathbf{H}}^\top\mathbf{R}^{-1}\widetilde{\mathbf{H}}\right)^{-1}\widetilde{\mathbf{H}}^\top\mathbf{R}^{-1}\mathbf{y}^{\mathcal{T}}$;
\item $\boldsymbol{\Omega}=\mathrm{diag}(\sigma^2_1(\mathbf{x}_1),\dots,\sigma^2_d(\mathbf{x}_d))$ and $\mathbf{P}=\mathrm{blkdiag}(\boldsymbol{\Omega},\,\mathbf{0})$;
\item $\mathbf{A}=\mathbf{R}^{-1}\left(\mathbf{y}^{\mathcal{T}}-\mathbf{w}^{\mathcal{T}}\widehat{\boldsymbol{\theta}}-\mathbf{H}(\mathbf{z}^{\mathcal{T}})\widehat{\boldsymbol{\beta}}\right)$ with $\mathbf{H}(\mathbf{z}^{\mathcal{T}})=[\mathbf{h}(\mathbf{z}_1^{\mathcal{T}}),\dots,\mathbf{h}(\mathbf{z}_m^{\mathcal{T}})]^\top$; 
\item $\mathbf{Q}=\mathbf{R}^{-1}\widetilde{\mathbf{H}}\left(\widetilde{\mathbf{H}}^\top\mathbf{R}^{-1}\widetilde{\mathbf{H}}\right)^{-1}\widetilde{\mathbf{H}}^\top\mathbf{R}^{-1}-\mathbf{R}^{-1}$ with $\widetilde{\mathbf{H}}=\left[\mathbf{w}^{\mathcal{T}},\mathbf{H}(\mathbf{z}^{\mathcal{T}})\right]$;
\item $\mathbf{G}=[\boldsymbol{\mu}^\top,\,\mathbf{h}(\mathbf{z})^\top]^\top$, $\mathbf{C}=\left(\widetilde{\mathbf{H}}^\top\mathbf{R}^{-1}\widetilde{\mathbf{H}}\right)^{-1}$ and $\mathbf{K}=\left[\mathbf{B}^\top,\,\mathbf{I}\mathbf{h}(\mathbf{z})^\top\right]$;
\item $\mathbf{I}$ is a $m\times 1$ column vector with the $i$-th element given by 
\begin{equation*}
I_i=\prod_{k=1}^pc_k(z_k,z^{\mathcal{T}}_{ik})\prod_{k=1}^d\xi_{ik},
\end{equation*}
where $\xi_{ik}\eqdef\mathbb{E}\left[c_k(W_k(\mathbf{x}_k),\,w^{\mathcal{T}}_{ik})\right]$;
\item $\mathbf{J}$ is a $m\times m$ matrix with the $ij$-th element given by
\begin{equation*}
J_{ij}=\prod_{k=1}^p c_k(z_k,\,z^{\mathcal{T}}_{ik})\,c_k(z_k,\,z^{\mathcal{T}}_{jk})\prod_{k=1}^d \zeta_{ijk},
\end{equation*} 
where $\zeta_{ijk}\eqdef\mathbb{E}\left[c_k(W_k(\mathbf{x}_k),\,w^{\mathcal{T}}_{ik})\,c_k(W_k(\mathbf{x}_k),\,w^{\mathcal{T}}_{jk})\right]$;
\item $\mathbf{B}$ is a $d\times m$ matrix with the $lj$-th element given by
\begin{equation*}
B_{lj}=\psi_{jl}\prod^d_{\substack{k=1\\k\neq l}}\xi_{jk} \prod_{k=1}^pc_k(z_k,z^{\mathcal{T}}_{jk}),
\end{equation*}
where $\psi_{jl}\eqdef\mathbb{E}\left[W_l(\mathbf{x}_l)\,c_l(W_l(\mathbf{x}_l),\,w^{\mathcal{T}}_{jl})\right]$.
\end{itemize}
\end{theorem}
\begin{proof}
The proof is in Section~\ref{sec:thmproof} of supplementary materials.
\end{proof}

\begin{proposition}
\label{prop:kernel}
The three expectations $\xi_{ik}$, $\zeta_{ijk}$ and $\psi_{jl}$ defined in Section~\ref{thm:main} have closed form expressions for the squared exponential kernel and a class of Mat\'{e}rn kernels~\citep{williams2005gaussian} defined by
\begin{equation}
\label{eq:matern}
c_k(d_{ij,k})=\exp\left(-\frac{\sqrt{2p+1}\,d_{ij,k}}{\gamma_k}\right)\frac{p!}{(2p)!}\sum_{i=0}^p\frac{(p+i)!}{i!(p-i)!}\left(\frac{2d_{ij,k}\sqrt{2p+1}}{\gamma_k}\right)^{p-i},  
\end{equation}
where $d_{ij,k}=X_{ik}-X_{jk}$ and $p$ is a non-negative integer. 
\end{proposition}
\begin{proof}
Derivations for the squared exponential kernel, Mat\'{e}rn kernels~\eqref{eq:matern} with $p=0$ (exponential), $p=1$ (Mat\'{e}rn-1.5) and $p=2$ (Mat\'{e}rn-2.5) are detailed in Section~\ref{sec:proofprop} of supplementary materials. The corresponding closed form expressions are summarized in Appendice~\ref{app:expression}. The closed form expressions for Mat\'{e}rn kernels with $p\geq 3$ can be obtained straightforwardly by invoking Lemma~\ref{lemma:partial} of supplementary materials and using same arguments in proofs of Mat\'{e}rn-1.5 and Mat\'{e}rn-2.5. Note that we reproduce the result for the squared exponential kernel given in~\cite{kyzyurova2018coupling} using our own notations for completeness.
\end{proof}

\subsection{A Synthetic Experiment}
\label{sec:experiments}
Consider the computer system shown in Figure~\ref{fig:exp2}, which consists three computer models with the following analytical functional forms:
\begin{equation*}
f_1=30+5x_1\sin(5x_1),\quad f_2=4+\exp(-5x_2)\quad\mathrm{and}\quad f_3=(w_1w_2-100)/6
\end{equation*}
with $x_1\in[0,\,2]$ and $x_2\in[0,\,2]$.

\begin{figure}[htbp]
\centering
\scalebox{1}{
\begin{tikzpicture}[shorten >=1pt,->,draw=black!50, node distance=4.5cm]
    \tikzstyle{every pin edge}=[<-,shorten <=1pt]
    \tikzstyle{neuron}=[circle,fill=black!25,minimum size=17pt,inner sep=0pt]
    \tikzstyle{layer1}=[neuron, fill=green!50];
    \tikzstyle{layer2}=[neuron, fill=red!50];
    \tikzstyle{layer3}=[neuron, fill=blue!50];
    \tikzstyle{annot} = [text width=4em, text centered]

    \node[layer1, pin=left:${x}_1$] (I-1) at (0,-0.75) {$f_1$};
    \node[layer1, pin=left:${x}_2$] (I-2) at (0,-2) {$f_2$};

    \path[yshift=-1.375cm]
            node[layer2, pin={[pin edge={->}]right:$y$}] (H-1) at (4.5cm,0 cm) {$f_3$};
            \path [draw] (I-1) -- (H-1) node[font=\small,midway,fill=white,align=left,sloped] {$w_1$};
            \path [draw] (I-2) -- (H-1) node[font=\small,midway,fill=white,align=left,sloped] {$w_2$};

    \node[annot,above of=I-1, node distance=0.7cm] (hl) {Layer 1};
    \node[annot,right of=hl] {Layer 2};
\end{tikzpicture}}
\caption{The computer system in the synthetic experiment where $f_1$ and $f_2$ are two computer models with one-dimensional input and output, and $f_3$ is a computer model with two-dimensional input and one-dimensional output.}
\label{fig:exp2}
\end{figure}
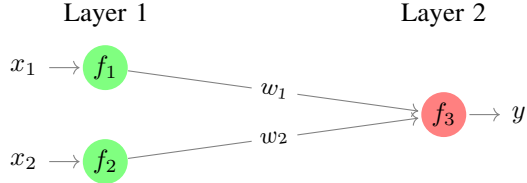

We generate ten training points from the maximin Latin hypercube and construct the composite emulator (Figure~\subref*{fig:exp2com_ang1}) and linked GaSP (Figure~\subref*{fig:exp2int_ang1}) of the system with Mat\'{e}rn-2.5 kernel. Figure~\subref*{fig:exp2int_ang1} indicates that the Mat\'{e}rn extension to the linked GaSP is valid because the constructed linked GaSP interpolates training points with sensible predictive mean and bounds.

\begin{figure}[htbp]
\centering 
\subfloat[Composite Emulator]{\label{fig:exp2com_ang1}\includegraphics[width=0.45\linewidth]{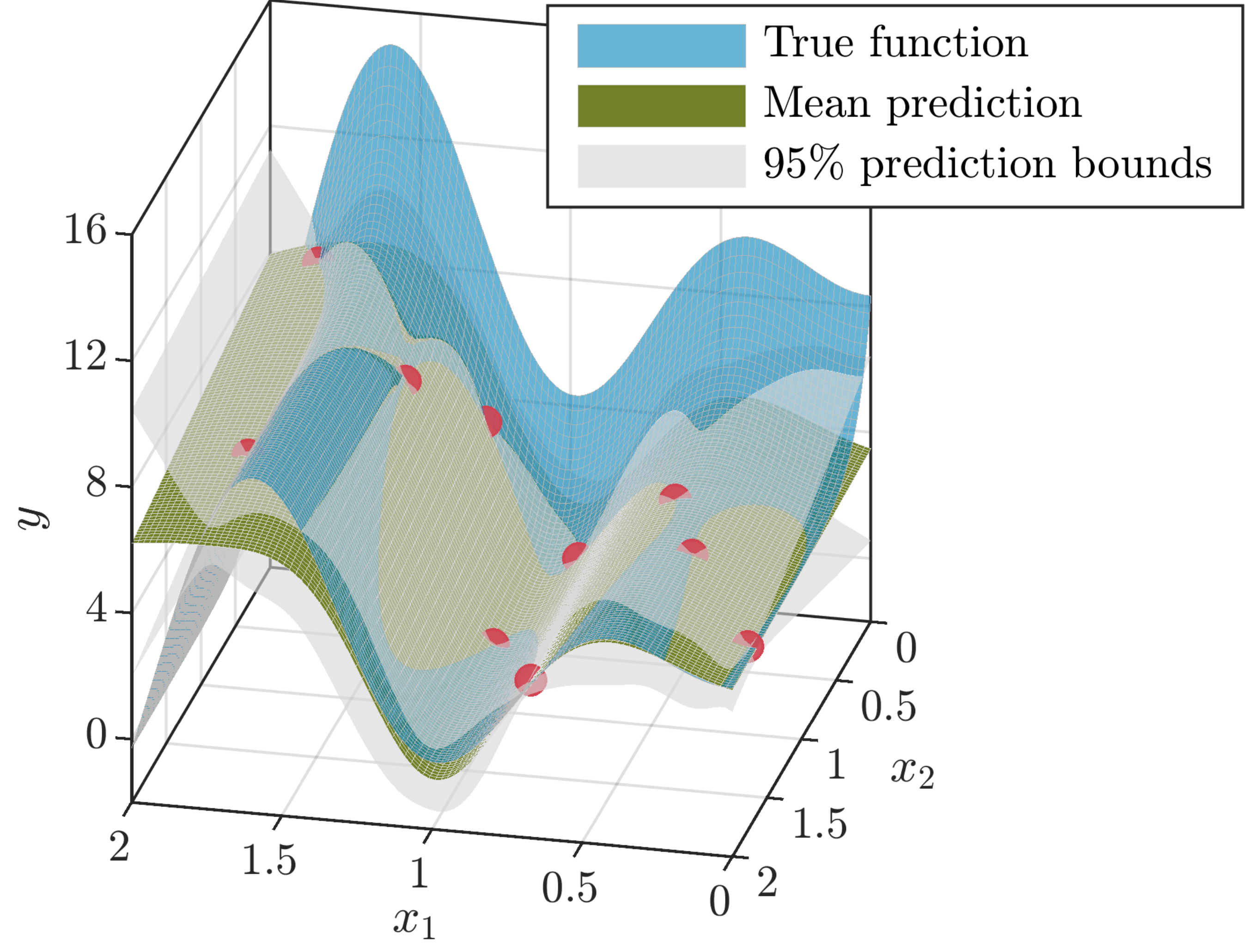}} 
\subfloat[Linked GaSP]{\label{fig:exp2int_ang1}\includegraphics[width=0.45\linewidth]{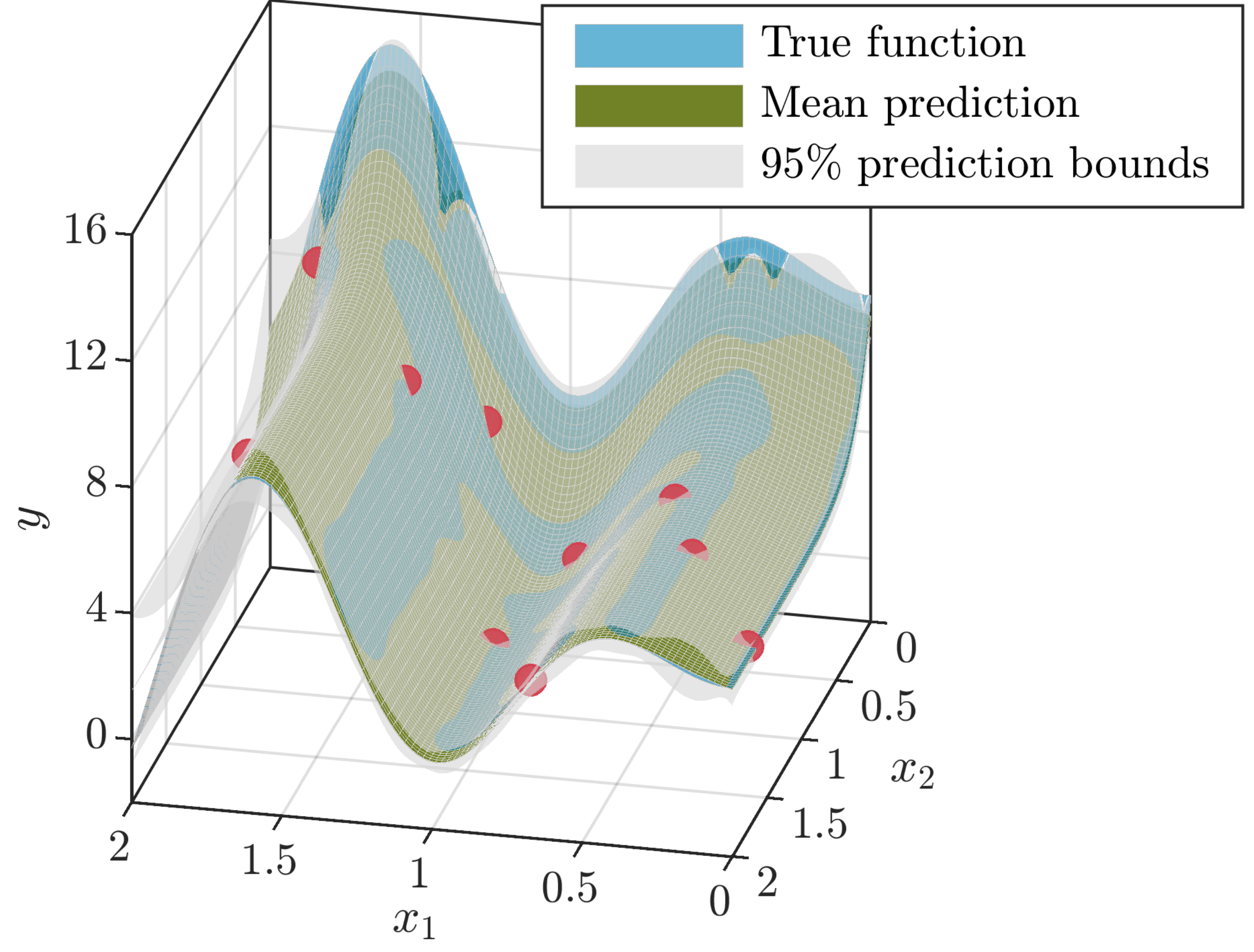}}
\caption{The composite emulator and linked GaSP of the system in Figure~\ref{fig:exp2}. The filled circles are training points used to construct the emulators.}
\label{fig:exp2emu}
\end{figure}

We further compare the linked GaSP with composite emulator with Mat\'{e}rn-2.5 kernel at different training sizes in Figure~\subref*{fig:exp2eva}. At each selected training set size, normalized root mean squared error of prediction (NRMSEP) of both composite emulator and linked GaSP are calculated, where
    \begin{equation}
    \label{eq:rmse1}
        \mathrm{NRMSEP}=\frac{\sqrt{\frac{1}{nT}\sum_{t=1}^T\sum_{i=1}^n(y(\mathbf{x}_i)-\mu^t_Y(\mathbf{x}_i))^2}}{\max\{y(\mathbf{x}_i)_{i=1,\dots,n}\}-\min\{y(\mathbf{x}_i)_{i=1,\dots,n}\}},
    \end{equation}
    in which $y(\mathbf{x}_i)$ denotes the true global output of the system evaluated at the testing input position $\mathbf{x}_i$ for $i=1,\dots,n$ with $n=2500$, which are equally spaced over the global input domain $[0,\,2]\times[0,\,2]$; $\mu^t_Y(\mathbf{x}_i)$ is the mean prediction of the respective emulator built with the $t$-th design of total $T=100$ designs sampled from the maximin Latin hypercube. Both Figure~\ref{fig:exp2emu} and~\subref*{fig:exp2eva} show that the linked GaSP outperforms (in terms of mean predictions, prediction bounds, NRMSEP and training cost) the composite emulator under the Mat\'{e}rn-2.5 kernel.

\begin{figure}[htbp]
  \centering
  \subfloat[Composite Emulator \emph{vs} Linked GaSP]{\label{fig:exp2eva}\includegraphics[width=0.45\linewidth]{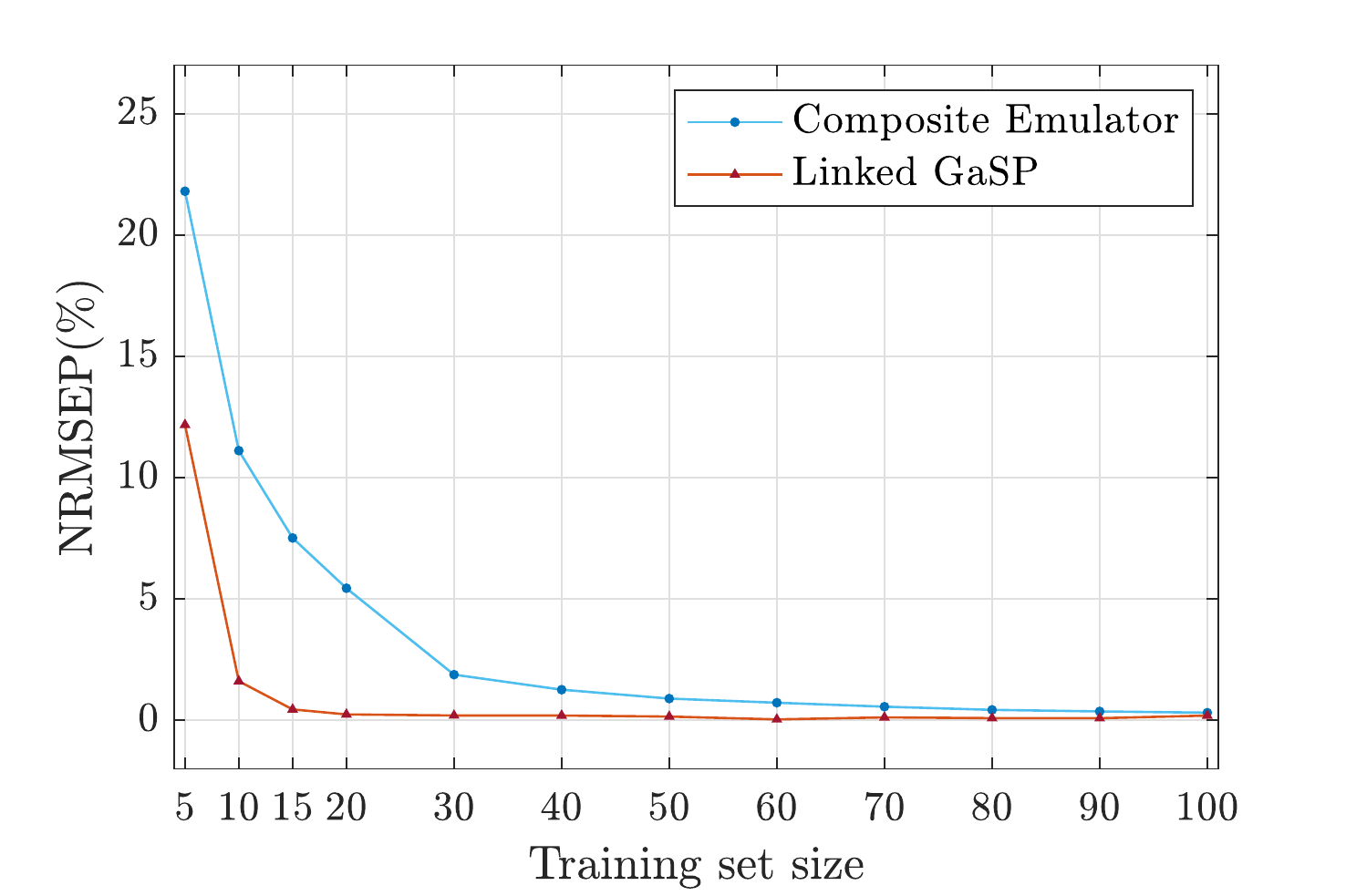}} 
\subfloat[Squared Exponential \emph{vs} Mat\'{e}rn-2.5]{\label{fig:exp2kernel}\includegraphics[width=0.45\linewidth]{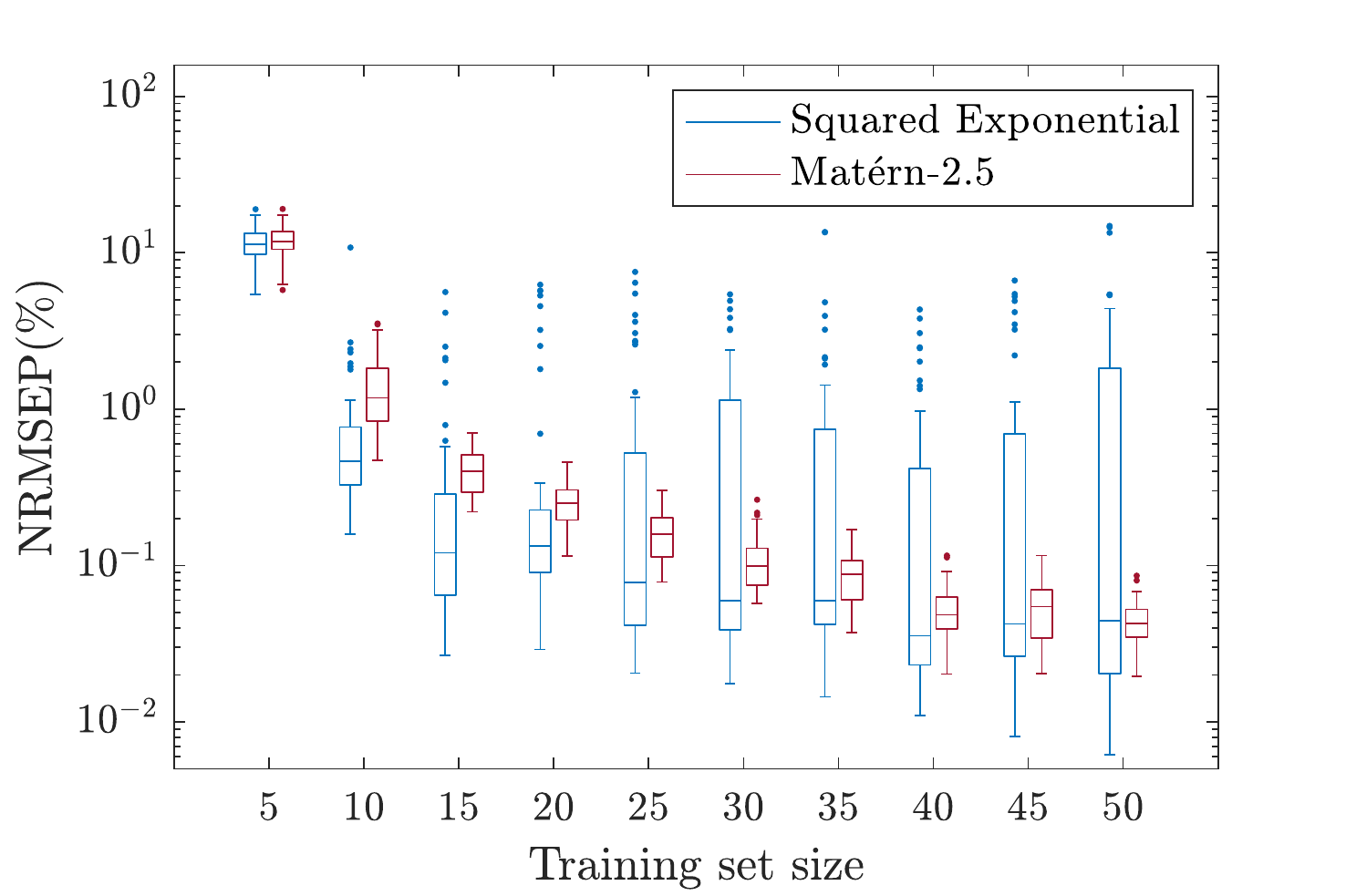}}
  \caption{Emulation results for the system in Figure~\ref{fig:exp2}. \emph{(a)} NRMSEP of composite emulator and linked GaSP with Mat\'{e}rn-2.5 kernel; \emph{(b)} NRMSEP of linked GaSPs with squared exponential and Mat\'{e}rn-2.5 kernels, both with a small nugget to handle ill-conditioned correlation matrices whenever necessary. NRMSEP in (b) is shown under the log-scale. }
  \label{fig:exp2result}
\end{figure}

In Figure~\subref*{fig:exp2kernel}, NRMSEP between linked GaSPs with squared exponential and Mat\'{e}rn-2.5 kernels are compared under ten different training set sizes. At each selected training set size, NRMSEPs are computed (without averaging over $T$ in equation~\eqref{eq:rmse1}) for $T=50$ random designs drawn from the maximin Latin hypercube. The NRMSEP of the linked GaSP with Mat\'{e}rn-2.5 kernel decays steadily as the training set size increases and its predictive performance is robust across different designs. On the contrary, NRMSEP of the linked GaSP with squared exponential kernel decreases with increasing oscillations over designs. Particularly, as the training set size increases beyond $15$, the linked GaSP with squared exponential kernel exhibits increasing chances of NRMSEPs over $1.0\%$ with extreme NRMSEPs reaching 5-10\% for some designs, whereas the linked GaSP with Mat\'{e}rn-2.5 kernel consistently provides NRMSEPs lower than 0.5-1.0\%. The large fluctuations of NRMSEPs displayed in the squared exponential case are due to the GaSP emulator $\widehat{f}_3$ that cannot capture adequately the true functional form of $f_3$ under some designs with the squared exponential kernel. It is also worth noting that in constructing GaSP emulators of individual computer models we experience ill-conditioned correlation matrices (which are subsequently addressed by enhancing their diagonal elements with a small nugget term) more frequently with the squared exponential kernel than the Mat\'{e}rn-2.5 kernel. These results stress the importance of Mat\'{e}rn extensions to the linked GaSP, in agreement with~\cite{gu2018robust,gramacy2020surrogates} that Mat\'{e}rn kernels are less vulnerable to ill-conditioning issues, provide reasonably adequate choices on the smoothness, and have both attractive theoretical properties and good practical performance. Furthermore, in practice, Mat\'{e}rn-1.5 and Mat\'{e}rn-2.5 are included in several computer emulation packages, such as~\texttt{DiceKriging} and~\texttt{RobustGaSP}, where Mat\'{e}rn-2.5 is the default kernel choice. In the remainder of the study, Mat\'{e}rn-2.5 is thus used for all GaSP emualtor constructions.

\section{Construction of Linked GaSP for Multi-Layered Computer Systems}
\label{sec:feedforward}

In this section, we demonstrate how to construct linked GaSP for a multi-layered system with feed-forward hierarchy, in which the  outputs of lower-layer computer models act as the inputs of higher-layer ones. 

It is a challenging analytical work to construct linked GaSP for a multi-layered feed-forward system in one-shot because there exists no closed form expressions for the mean and variance of the linked emulator, whose density function involves integration of GaSP emulators across a large number of layers. However, one could collapse a complex feed-forward system into a sequence of two-layered computer systems, and then successively construct linked GaSPs across two layers.

Consider a general feed-forward system of computer models, denoted by $e_{1\rightarrow L}$, with $L$ layers. The system can be decomposed into a sequence of $L-1$ sub-systems: $e_{1\rightarrow (i+1)}$ for $i=1,\dots,L-1$. Then, the linked GaSP of the whole system ($e_{1\rightarrow L}$) is built by the following steps:
\begin{enumerate}
    \item Construct the linked GaSP of $e_{1\rightarrow 2}$ by applying Theorem~\ref{thm:main} to GaSP emulators of computer models in the first and second layers of $e_{1\rightarrow L}$;
    \item For $i=2,\dots,L-1$, construct the linked GaSP of $e_{1\rightarrow i+1}$ by applying Theorem~\ref{thm:main} to the linked GaSP of $e_{1\rightarrow i}$ and GaSP emulators of computer models in the $(i+1)$-th layer of $e_{1\rightarrow L}$;
\end{enumerate}

For example, the system in Figure~\ref{fig:recsystem} can be decomposed into three recursive systems: $e_{1\rightarrow 2}$, $e_{1\rightarrow 3}$ and $e_{1\rightarrow 4}$, and the linked GaSP of the whole system $e_{1\rightarrow 4}$ takes three iterations to be produced. It is noted that the above iterative procedure works because Assumption~\ref{ass:2} only requires normality while has no constraints on specific forms of corresponding mean and variance.

\begin{figure}[htbp]
\centering
\scalebox{1}{
\begin{tikzpicture}[shorten >=1pt,->,draw=black!50, node distance=2.5cm]
    \tikzstyle{every pin edge}=[<-,shorten <=1pt]
    \tikzstyle{neuron}=[circle,fill=black!25,minimum size=17pt,inner sep=0pt]
    \tikzstyle{layer1}=[neuron, fill=green!50];
    \tikzstyle{layer2}=[neuron, fill=red!50];
    \tikzstyle{layer3}=[neuron, fill=blue!50];
    \tikzstyle{layer4}=[neuron, fill=purple!50];
    \tikzstyle{annot} = [text width=3.5em, text centered]
    
    \path[yshift=-1.625cm]
     node[layer1] (I-1) at (0,0) {${f}_1$};

    \node[layer2] (I-2) at (2.5,-0.75) {$f_2$};
    \node[layer2] (I-3) at (2.5,-2.5) {$f_3$};
    
    \node[layer3] (I-4) at (5,-0.75) {$f_4$};
    \node[layer3] (I-5) at (5,-2.5) {$f_5$};
    
    \path[yshift=-1.625cm]
     node[layer4] (I-6) at (7.5cm, 0cm) {${f}_6$};

    \path [draw] (I-1) -- (I-2);
    \path [draw] (I-1) -- (I-3);
    \path [draw] (I-2) -- (I-4);
    \path [draw] (I-2) -- (I-5);
    \path [draw] (I-3) -- (I-4);
    \path [draw] (I-3) -- (I-5);
    \path [draw] (I-4) -- (I-6);
    \path [draw] (I-5) -- (I-6);
    
    \draw[red!80,thick,dotted] ($(I-1.north east)+(-0.7,1.6)$) rectangle ($(I-3.south west)+(1,-0.2)$);
    \node[text width=4em, text centered,above of=I-2,node distance=0.7cm,text=red!80,xshift=0.2cm] {$\mathbf{e_{1\rightarrow 2}}$};
    \draw[blue!80,thick,dotted] ($(I-1.north east)+(-0.85,1.75)$) rectangle ($(I-5.south west)+(1,-0.35)$);
    \node[text width=4em, text centered,above of=I-4,node distance=0.85cm,text=blue!80,xshift=0.2cm] {$\mathbf{e_{1\rightarrow 3}}$};
    \draw[black!80,thick,dotted] ($(I-1.north east)+(-1,1.9)$) rectangle ($(I-6.south west)+(1,-1.4)$);
    \node[text width=4em, text centered,above of=I-6,node distance=1.85cm,text=black!80,xshift=0.2cm] {$\mathbf{e_{1\rightarrow 4}}$};
    
    \node[annot,left of=I-1, node distance=2cm] (input1) {Global Input 1};
    \node[annot,below of=I-1, node distance=2.5cm] (input2) {Global Input 2};
    \node[annot,below of=I-6, node distance=2.5cm] (input3) {Global Input 3};
    \node[annot,right of=I-6, node distance=2cm] (output) {Global Output};
    \path [draw] (input1) -- (I-1);
    \path [draw] (input2) -- (I-1);
    \path [draw] (input3) -- (I-6);
    \path [draw] (I-6) -- (output);
    \node[annot,above of=I-2, node distance=1.6cm] (hl1) {Layer 2};
    \node[annot,left of=hl1] {Layer 1};
    \node[annot,above of=I-4, node distance=1.6cm] (hl2) {Layer 3};
    \node[annot,right of=hl2] {Layer 4};
\end{tikzpicture}}
\caption{An illustration on the iterative procedure to construct linked GaSP for a 4-layered feed-forward computer system.}
\label{fig:recsystem}
\end{figure}
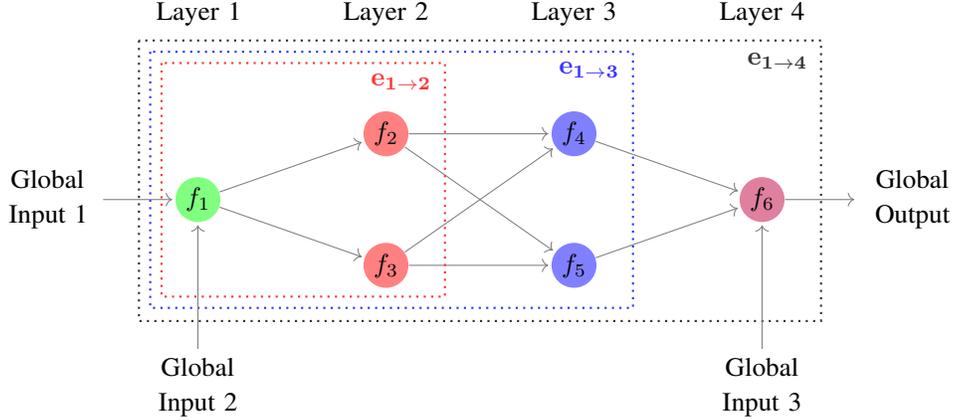

\subsection{Linked GaSP for a Feed-back Coupled Satellite Model} 
\label{sec:satellite}
In this section, we show the construction of the linked GaSP for a multi-layered fire-detection satellite model studied in~\cite{sankararaman2012likelihood}. This satellite is designed to conduct near-real-time detection, identification and monitoring of forest fires. The satellite system consists of three sub-models, namely the orbit analysis, the attitude control and power analysis. The satellite system is shown in Figure~\ref{fig:example}. It can be seen from Figure~\ref{fig:example} that there are nine global input variables $H,\,F_s,\,\theta,\,L_{sp},\,q,\,R_D,\,L_a,\,C_d,\,P_{other}$ and three global output variables of interest $\tau_{tot},\,P_{tot},\,A_{sa}$. The coupling variables are $\Delta t_{orbit}$, $\Delta t_{eclipse}$, $\nu$, $\theta_{slew}$, $P_{ACS}$, $I_{max}$ and $I_{min}$. Since $\Delta t_{orbit}$ is the input to both power analysis and attitude control, there are total eight coupling variables. Note that the system has feed-back coupling because the coupling variables $P_{ACS}$, $I_{max}$ and $I_{min}$ form an internal loop between power analysis and attitude control. Therefore, to implement the iterative procedure to build the linked GaSP of the system, we first convert the system to a feed-forward one by applying the decoupling algorithm proposed in~\cite{baptista2018optimal}. The decoupling algorithm identifies four weakly coupled variables $\Delta t_{orbit}$ (between orbit analysis and attitude control), $\theta_{slew}$, $I_{max}$ and $I_{min}$. Since the weakly coupled variables have insignificant impact on the accuracy of global outputs, they are neglected from the interaction terms between sub-models, producing a feed-forward system (see Figure~\ref{fig:example} without the dashed arrows). Table~\ref{tab:input} gives the domains of global inputs considered for the emulation.     

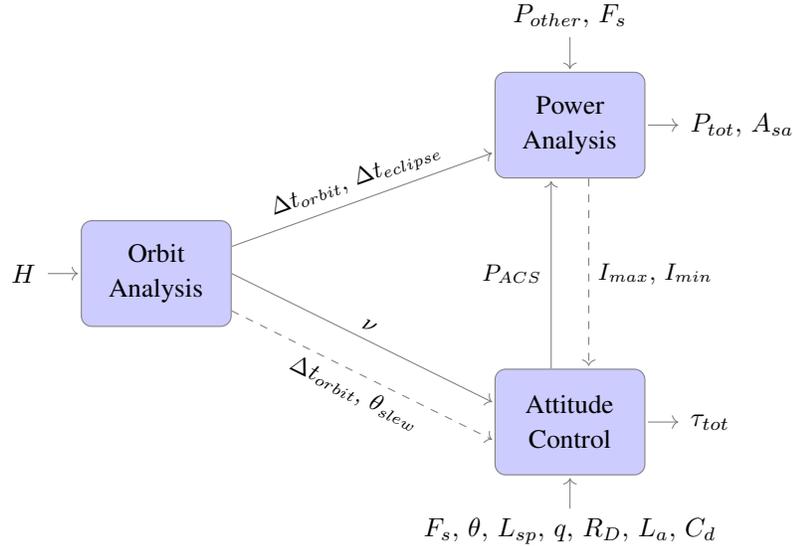
\begin{figure}[htbp]
\centering
\scalebox{1}{
\begin{tikzpicture}[shorten >=1pt,->,draw=black!50, node distance=5.5cm]
   \tikzstyle{every pin edge}=[<-,shorten <=1pt]
   \tikzstyle{block}=[rectangle, draw, fill=blue!20, text width=5em, text centered, rounded corners, minimum height=4em]
    \tikzstyle{annot} = [text width=4em, text centered]
    \tikzstyle{txtnode} = [text width=20em, text centered] 
    \node[block, pin=left:$H$] (I-1) at (0.5,-2.525) {Orbit Analysis};
    \node[block, pin={[pin edge={->}]right:$P_{tot},\,A_{sa}$}] (I-2) at (6,-0.55) {Power Analysis};
    \node[block, pin={[pin edge={->}]right:$\tau_{tot}$}] (I-3) at (6,-4.5) {Attitude Control};
    \node[txtnode,below of=I-3, node distance=1.45cm] (h) {$F_s,\,\theta,\,L_{sp},\,q,\,R_D,\,L_a,\,C_d$}; 
    \node[txtnode,above of=I-2, node distance=1.45cm] (h1) {$P_{other},\,F_s$};
    
           \path [draw] (h) -- (I-3);
           \path [draw] (h1) -- (I-2);
           \path [draw] (I-1) -- (I-2) node[above, font=\small,midway,align=left,sloped] {$\Delta t_{orbit},\,\Delta t_{eclipse}$};
           \path [draw,dashed] (I-1) -- ([yshift=-0.25cm]I-3.west) node[below,font=\small,midway,align=left,sloped] {$\Delta t_{orbit},\,\theta_{slew}$};
           \path [draw] ([yshift=0cm]I-1.east) -- ([yshift=.25cm]I-3.west) node[above,font=\small,midway,align=left,sloped] {$\nu$};
           \path [draw,dashed] ([xshift=.25cm]I-2.south) -- ([xshift=.25cm]I-3.north) node[right,font=\small,midway,align=left] {$I_{max},\,I_{min}$};
           \path [draw] ([xshift=-.25cm]I-3.north) -- ([xshift=-.25cm]I-2.south)  node[left,font=\small,midway,align=left] {$P_{ACS}$};
\end{tikzpicture}}
\caption{Fire-detection satellite model from~\cite{sankararaman2012likelihood}, where $H$ is altitude; $\Delta t_{orbit}$ is orbit period; $\Delta t_{eclipse}$ is eclipse period; $\nu$ is satellite velocity; $\theta_{slew}$ is maximum slewing angel; $P_{other}$ represents other sources of power; $P_{ACS}$ is power of attitude control system; $I_{max},\,I_{min}$ are maximum and minimum moment of inertia respectively; $F_s,\theta,\,L_{sp},\,q,\,R_D,\,L_a,\,C_d$ represent average solar flux, deviation of moment axis from vertical, moment arm for the solar radiation torque, reflectance factor, residual dipole, moment arm for aerodynamic torque, and drag coefficient respectively; $P_{tot}$ is total power; $A_{sa}$ is area of solar array; and $\tau_{tot}$ is total torque. The dashed arrows indicate the connections that can be decoupled between sub-models, according to the decoupling algorithm from~\cite{baptista2018optimal}.}
\label{fig:example}
\end{figure}

\begin{table}[htbp]  
\footnotesize{
\caption{Domains of the nine global input variables to be considered for the emulation.}
\label{tab:input}	
\begin{center}
\begin{tabular}{lcc}
\toprule
\textbf{Global input variable (unit)} & \textbf{Symbol} & \textbf{Domain}\\
\midrule
Altitude ($m$) & $H$ & $\left[1.50\times10^{17},\,2.10\times10^{17}\right]$\\
\addlinespace[0.1cm]
Other sources of power ($W$) & $P_{other}$ & $\left[8.50\times10^{2},\,1.15\times10^{3}\right]$\\
\addlinespace[0.1cm]
Average solar flux ($W/m^2$) & $F_s$ & $\left[1.34\times10^{3},\,1.46\times10^{3}\right]$\\
\addlinespace[0.1cm]
Deviation of moment axis from vertical ($^{\circ}$) & $\theta$ & $[12.00,\,18.00]$\\
\addlinespace[0.1cm]
Moment arm for the solar radiation torque ($m$) & $L_{sp}$ & $[0.80,\,3.20]$\\
\addlinespace[0.1cm]
Reflectance factor & $q$ & $[0,\,1]$\\
\addlinespace[0.1cm]
Residual dipole ($A\cdot m^2$)& $R_D$ & $[2.00,\,8.00]$\\
\addlinespace[0.1cm]
Moment arm for aerodynamic torque ($m$) & $L_a$ & $[0.80,\,3.20]$\\
\addlinespace[0.1cm]
Drag coefficient & $C_d$ & $[0.10,\,1,90]$ \\
\bottomrule
\end{tabular}
\end{center}
}
\end{table}

Maximin Latin hypercube sampling is then used to generate inputs positions for seven training sets, with sizes of $10$, $15$, $20$, $25$, $30$, $35$ and $40$ respectively. The corresponding output positions are consequently obtained by running the satellite model. For each of the seven training set and each of the three global output variables, we build the composite emulator and linked GaSP. Leave-one-out cross-validation is utilized for assessing the predictive performance of the two emulators. For example, in case of the composite emulation of the output variable $P_{tot}$ with training set size of $10$, we build ten composite emulators, each based on nine training points by dropping one training point out of the set. The dropped training point is then serves as the testing point to assess the associated composite emulator. The performance of the emulator (composite emulator or linked GaSP) of a global output variable given a certain training set is ultimately summarized by
\begin{equation*}
\mathrm{NRMSEP}=\frac{\sqrt{\frac{1}{n}\sum_{i=1}^n(f(\mathbf{x}_i)-\mu^{-i}(\mathbf{x}_i))^2}}{\max\{f(\mathbf{x}_i)_{i=1,\dots,n}\}-\min\{f(\mathbf{x}_i)_{i=1,\dots,n}\}},
\end{equation*}
where $\mathbf{x}_i$ is the $i$-th input position of a training set with size $n$; $f(\mathbf{x}_i)$ is the value of the output variable of interest produced by the satellite model at the input $\mathbf{x}_i$; the mean prediction $\mu^{-i}(\mathbf{x}_i)$ at input $\mathbf{x}_i$ is provided by the corresponding emulator constructed using all $n$ training points except for $\mathbf{x}_i$.   

The NRMSEP of the composite emulators and linked GaSPs of the three global output variables $\tau_{tot}$, $P_{tot}$ and $A_{sa}$ against seven different training sizes are presented on the top row of Figure~\ref{fig:sta_rmse}. It can be seen that for the output variable $\tau_{tot}$, the linked GaSP is only marginally better than the composite emulator. For the output variables $P_{tot}$ and $A_{sa}$, the linked GaSPs present better predictive performance than the composite ones when the training set size is small. The superiority of the linked GaSP soon vanishes when the training set size increases over $20$. To investigate the possible cause for this quick depreciation, we construct GaSP emulators for outputs produced by the three sub-models. The NRMSEP of these GaSP emulators across different training sizes are summarized on the bottom row of Figure~\ref{fig:sta_rmse}. We observe that the GaSP emulator of the attitude control with respect to $\tau_{tot}$ requires around $35$ training points to reach a low NRMSEP, while the GaSP emulator of the orbit analysis with respect to $\nu$ can reach such level with only $10$ training points. This indicates that the functional complexity between the global inputs and the output $\tau_{tot}$ is dominated by the sub-model attitude control, and thus the linked GaSP of $\tau_{tot}$ shows no obvious superiority over the corresponding composite emulator. Although the attitude control still dominates the functional complexity between the global inputs and $P_{tot}$ and $A_{sa}$ (see Figure~\subref*{fig:power}), $P_{tot}$ and $A_{sa}$ are produced not only by the orbit analysis and attitude control, but also by the power analysis. This extra sub-model increases the input dimension that the composite emulators need to explore, and thus cause the composite emulators slow to learn the functional dependence of $P_{tot}$ and $A_{sa}$ to the global inputs when training data size is small.

\begin{figure}[tbhp]
\centering 
\subfloat[$\tau_{tot}$]{\label{fig:tau_tot}\includegraphics[width=0.3\linewidth]{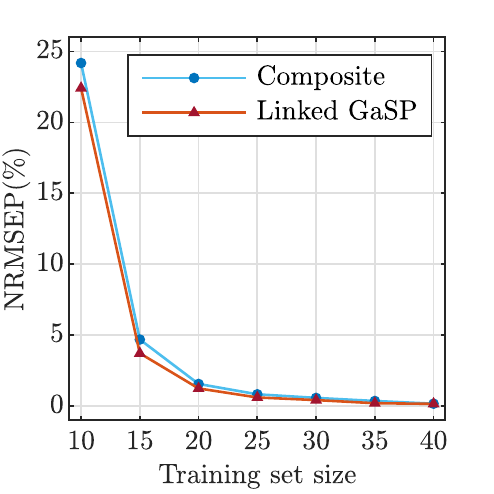}} 
\subfloat[$P_{tot}$]{\label{fig:p_tot}\includegraphics[width=0.3\linewidth]{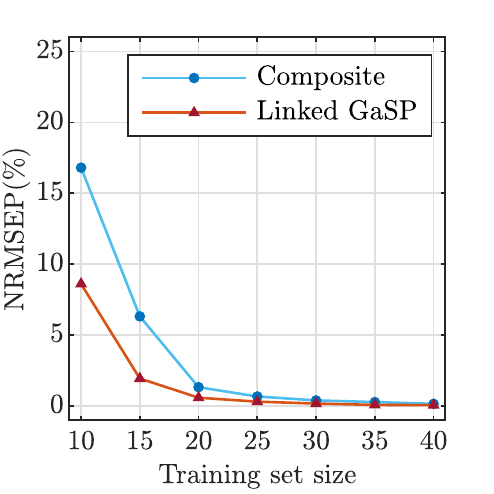}}
\subfloat[$A_{sa}$]{\label{fig:A_sa}\includegraphics[width=0.3\linewidth]{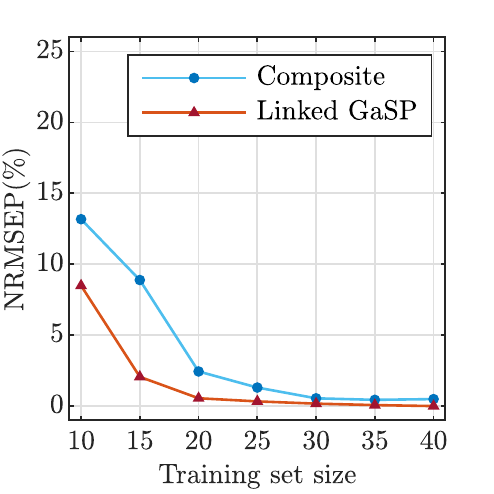}}\\[-1em]
\subfloat[Orbit Analysis]{\label{fig:orbit}\includegraphics[width=0.3\linewidth]{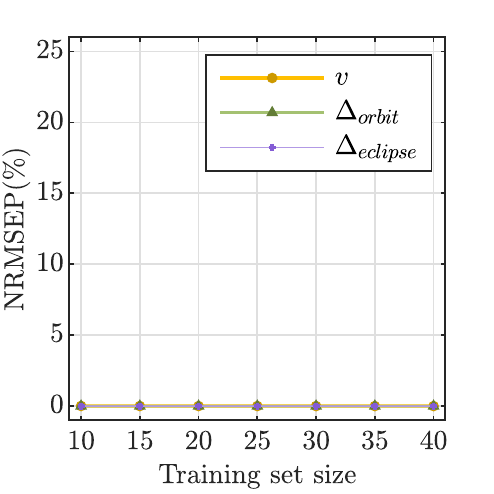}} 
\subfloat[Attitude Control]{\label{fig:ac}\includegraphics[width=0.3\linewidth]{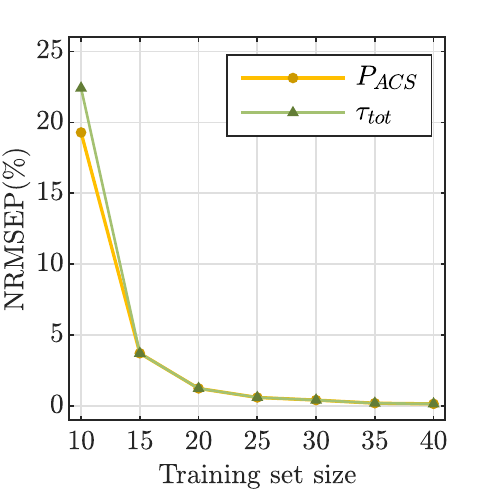}}
\subfloat[Power Analysis]{\label{fig:power}\includegraphics[width=0.3\linewidth]{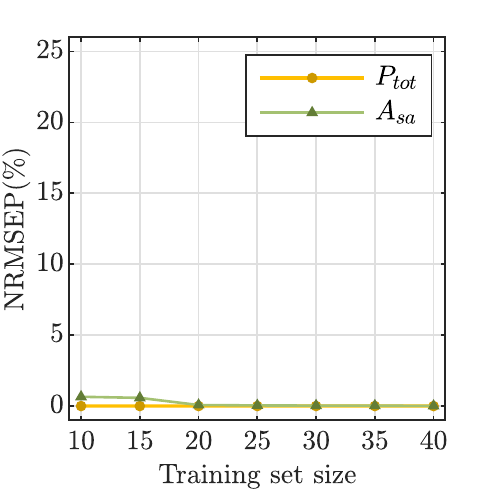}}
\caption{\emph{(Top)} NRMSEP of the composite emulators and linked GaSPs of the three global output variables $\tau_{tot}$, $P_{tot}$ and $A_{sa}$ against different training set sizes. \emph{(Bottom)} NRMSEP of the GaSP emulators of outputs produced by the three subsystems: orbit analysis, attitude control and power analysis.}
\label{fig:sta_rmse}
\end{figure}

\section{Experimental Designs for Linked GaSP}
\label{sec:design}
The linked GaSP is so far constructed using the Latin hypercube design (LHD)~\citep{santner2003design} in a sequential fashion. It means that a one-shot LHD is applied only to the global inputs (i.e., the inputs to the computer models in the first layer of the system) and designs for the inputs to the computer models in higher layers are automatically determined by the outputs from the lower-layer computer models. This design, called sequential LHD hereinafter, is a simple strategy and has the benefit that it only explores input spaces of individual computer models that have impact on the global outputs. However, the complexity of system structures and non-linearity of individual computer models can produce poor designs for sub-models in higher layers when the LHD of the global input is propagated through the system hierarchy. This issue can be seen from the sequential LHD (see Figure~\ref{fig:lhd}) that we used for the synthetic experiment in Section~\ref{sec:experiments}. Figure~\ref{fig:lhd} shows that although the LHD gives satisfactory input exploration for the global inputs $x_1$ and $x_2$, the design for the computer model $f_3$ is poor. This is because of the steep decrease of $f_2$ over $x_2\in[0,0.5]$, which concentrates most of the design points for $f_3$ on the border of its input $w_2$ while few of them locate over $w_2\in[4.1,5.0]$. Indeed, such an issue could be alleviated by increasing the size of the sequential LHD or implementing adaptive design strategies (e.g.,~\cite{beck2016sequential}) over the global inputs. However, these solutions can result in excessive design points that contain similar information about the underlying computer model. In addition, such sequential designs require full runs of entire systems, and thus can be computationally expensive and inefficient when the designs for some sub-models are already satisfactory and no further enhancements are needed. 

\cite{kyzyurova2018coupling} suggest an independent design strategy where the designs of sub-models are developed (by either one-shot LHD or adaptive designs) separately without considering their structural dependence. This design strategy is useful because the construction of the linked GaSP does not require realizations generated by running the whole system and thus different computer models can be ran in parallel rather than in sequence; one can even use existing realizations (with different sizes) from individual computer models to build the linked GaSP; the experimental design can be tailor-made for each computer model and thus one avoids issues related to the aforementioned sequential designs. 

While it is desirable to construct accurate GaSP emulators of individual computer models via the independent design and then integrate them to have a well-behaved linked GaSP, ignoring the structure dependence can cause \textit{unnecessary refinements} of GaSP emulators (and thus excessive experimental costs) over input spaces of computer models that are insignificant to the global output. Similarly, the ignorance of structural dependence may also cause GaSP emulators to be accurate only in part of input spaces that are significant to the global output. We illustrate such an issue in Section~\ref{sec:ind_design} of supplementary materials. In Section~\ref{sec:adaptive}, we introduce an adaptive design strategy for the linked GaSP that utilizes the analytical variance decomposition of linked emulators. As we will show, this design not only takes system structures into account but also shares some advantages of the independent design.

\begin{figure}[tbhp]
\centering 
\subfloat{\includegraphics[width=0.3\linewidth]{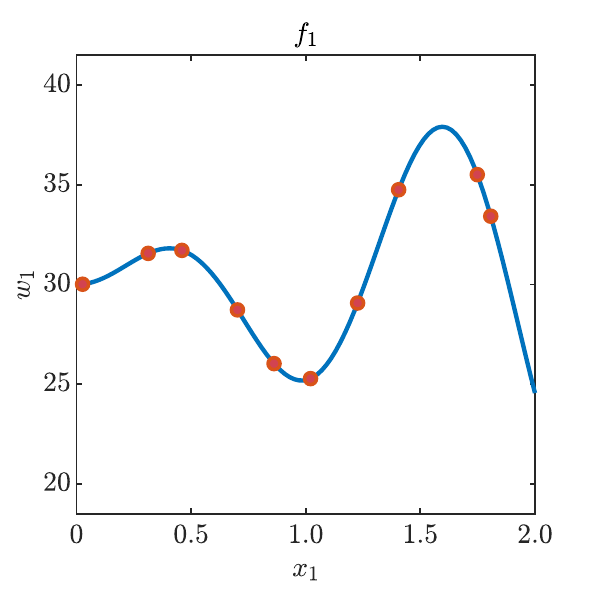}} 
\subfloat{\includegraphics[width=0.3\linewidth]{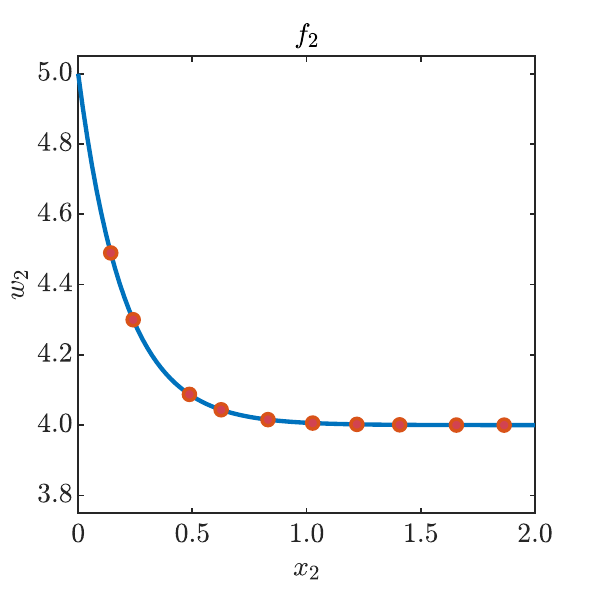}}
\subfloat{\includegraphics[width=0.3\linewidth]{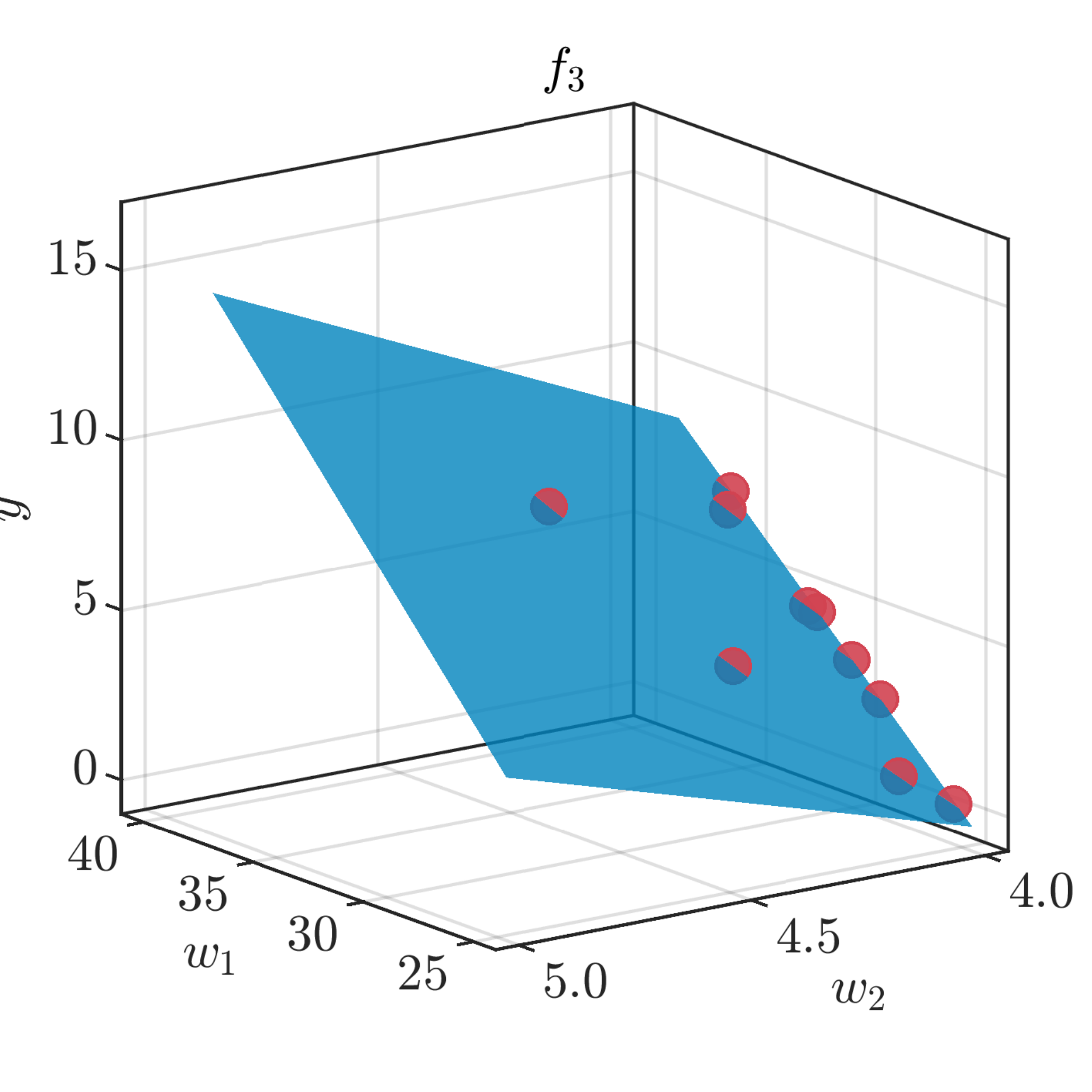}}
\caption{The sequential Latin hypercube design (LHD) used to build the linked GaSP for the synthetic experiment in Section~\ref{sec:experiments}. The solid lines and surface represent the true functional forms of each computer model; the filled circles are design points.}
\label{fig:lhd}
\end{figure}

\subsection{A Variance-Based Adaptive Design for Linked GaSP}
\label{sec:adaptive}
The adaptive design introduced in this section extends the simulation-based Single Model Selection training strategy given in~\cite{sanson2019systems}. At each iteration, the adaptive design conducts the follow three steps:
\begin{enumerate}
    \item Select one sub-model and determine the input position to run the model;
    \item Run the selected sub-model and refine its GaSP emulator given the new run;
    \item Construct the linked GaSP of the system.
\end{enumerate}

It can be seen that at each iteration the adaptive design only requires a single run of one sub-model. Therefore, one can save computational resources by avoiding runs of the whole system and only refining the GaSP emulator of one sub-model to improve the overall accuracy of the linked GaSP. We select the target sub-model at each iteration by searching for the sub-model whose GaSP emulator contributes the most to the variance of the linked GaSP. We demonstrate the approach on a two-layered system whose sub-models have their GaSP emulators connected as in Figure~\ref{fig:2systememu}. Note (see Section~cref{sec:thmproof} of supplementary materials) that the variance of linked emulator in equation~\eqref{eq:intvar} of Theorem~\ref{thm:main} can be written as
\begin{equation*}
\sigma^2_L=\mathrm{Var}\left(\mu_g(\mathbf{W},\mathbf{z})\right)+\mathbb{E}\left[\sigma^2_g(\mathbf{W},\mathbf{z})\right],
\end{equation*}
where
\begin{align*}
\mathrm{Var}\left(\mu_g(\mathbf{W},\mathbf{z})\right)&=\mathbf{A}^\top\left(\mathbf{J}-\mathbf{I}\mathbf{I}^\top\right)\mathbf{A}+2\widehat{\boldsymbol{\theta}}^\top\left(\mathbf{B}-\boldsymbol{\mu}\mathbf{I}^\top\right)\mathbf{A}+\mathrm{tr}\left\{\widehat{\boldsymbol{\theta}}\widehat{\boldsymbol{\theta}}^\top\boldsymbol{\Omega}\right\}\\
\mathbb{E}\left[\sigma^2_g(\mathbf{W},\mathbf{z})\right]&=\sigma^2\,\left(1+\eta+\mathrm{tr}\left\{\mathbf{Q}\mathbf{J}\right\}+\mathbf{G}^\top\mathbf{C}\mathbf{G}+\mathrm{tr}\left\{\mathbf{C}\mathbf{P}-2\mathbf{C}\widetilde{\mathbf{H}}^\top\mathbf{R}^{-1}\mathbf{K}\right\}\right)
\end{align*}
with $\mu_g(\mathbf{W},\mathbf{z})$ and $\sigma^2_g(\mathbf{W},\mathbf{z})$ being the mean and variance of $\widehat{g}$. 

Define 
\begin{equation*}
V_1=\mathrm{Var}\left(\mu_g(\mathbf{W},\mathbf{z})\right)\quad\mathrm{and}\quad V_2=\mathbb{E}\left[\sigma^2_g(\mathbf{W},\mathbf{z})\right],
\end{equation*}
then $V_1$ represents the overall contribution of GaSP emulators $\widehat{f}_1,\dots,\widehat{f}_d$ to $\sigma^2_L$, and $V_2$ represents the contribution of $\widehat{g}$ to $\sigma^2_L$. Analogously, the variance contribution of GaSP emulators $\widehat{f}_{k\in\mathbb{S}}$ for $\mathbb{S}\subseteq\{1,\dots,d\}$ can be defined by
\begin{equation*}
V_1(\mathbb{S})=\mathrm{Var}_{W_{k\in\mathbb{S}}}\left(\mathbb{E}_{W_{k\in\mathbb{S}^\mathsf{c}}}\left[\mu_g(\mathbf{W},\mathbf{z})\right]\right),    
\end{equation*}
where $\mathbb{S}^\mathsf{c}$ is the complement of $\mathbb{S}$. One can compute $V_1(\mathbb{S})$ analytically according to Proposition~\ref{prop:variance}.

\begin{proposition}
\label{prop:variance}
Under the same conditions of Theorem~\ref{thm:main}, $V_1(\mathbb{S})$ has the closed form expression given by
\begin{equation*}
V_1(\mathbb{S})=\mathbf{A}^\top\left(\widetilde{\mathbf{J}}-\mathbf{I}\mathbf{I}^\top\right)\mathbf{A}+2\widehat{\boldsymbol{\theta}}^\top\left(\widetilde{\mathbf{B}}-\boldsymbol{\mu}\mathbf{I}^\top\right)\mathbf{A}+\mathrm{tr}\left\{\widehat{\boldsymbol{\theta}}\widehat{\boldsymbol{\theta}}^\top\widetilde{\boldsymbol{\Omega}}\right\},
\end{equation*}
where
\begin{itemize}
\item $\widetilde{\boldsymbol{\Omega}}$	is a $d\times d$ diagonal matrix with $k$-th diagonal element given by
$\sigma_k^2(\mathbf{x}_k)\mathbbm{1}_{\{k\in\mathbb{S}\}}$;
\item $\widetilde{\mathbf{J}}$ is a $m\times m$ matrix with the $ij$-th element given by
\begin{equation*}
\widetilde{J}_{ij}=\prod_{k\in\mathbb{S}} \zeta_{ijk}\prod_{k\in\mathbb{S}^\mathsf{c}}\xi_{ik}\xi_{jk}\prod_{k=1}^p c_k(z_k,\,z^{\mathcal{T}}_{ik})\,c_k(z_k,\,z^{\mathcal{T}}_{jk});
\end{equation*} 
\item $\widetilde{\mathbf{B}}$ is a $d\times m$ matrix with the $lj$-th element given by
\begin{equation*}
\widetilde{B}_{lj}=\begin{dcases*}
\psi_{jl}\prod^d_{\substack{k=1\\k\neq l}}\xi_{jk} \prod_{k=1}^pc_k(z_k,z^{\mathcal{T}}_{jk}), & $l\in\mathbb{S}$,\\
\mu_l\prod^d_{\substack{k=1}}\xi_{jk} \prod_{k=1}^pc_k(z_k,z^{\mathcal{T}}_{jk}), & $l\in\mathbb{S}^\mathsf{c}$.
\end{dcases*}
\end{equation*}
\end{itemize}
\end{proposition}
\begin{proof}
The proof is in Section~\ref{sec:proofvariance} of supplementary materials.
\end{proof}

Thanks to the closed form expressions of $V_1$, $V_2$ and $V_1(\mathbb{S})$, the adaptive design can quickly locate the sub-model and determine the input position to run the model. To show the performance we implement the adaptive design on the synthetic example in Section~\ref{sec:experiments} via Algorithm~\ref{alg:design}, where the optimization problem in Line~\ref{alg:opt} is done by grid search due to the low global input dimension. The linked GaSP built by the adaptive design is summarized in Figure~\ref{fig:adp_exp}. It can be observed from Figure~\ref{fig:adp_exp} that the linked GaSP built via the adaptive design can achieve lower NRMSEP than that built via the sequential LHD, with a smaller number of computer model runs. This is because, in contrast to the poor design for $f_3$ created by the sequential LHD (see Figure~\ref{fig:lhd}), the adaptive design creates a satisfactory design by adding extra design points to the input space of $f_3$ that is not well-explored by the sequential LHD but still significant to the global output. It can also be seen that the adaptive design leads to more runs of $f_1$, whose functional form is more complex than other models and thus needs to generate more realizations to be emulated adequately. Thus the adaptive design is able to improve the emulation performance of the linked GaSP with reduced experimental costs by allocating runs to computer models according to their heterogeneous functional complexity. We also report in Figure~\ref{fig:adp_exp} the NRMSEP of the linked GaSP trained with the independent design, by which GaSP emulators of individual computer models are built separately with their own training points independently generated from the LHD. Although the linked GaSP with the independent design achieves a low NRMSEP, its accuracy is overestimated because we assume that the input domain of $f_3$ that is significant to the global output is perfectly known or can be determined in a cost efficient way, e.g., we were able to determine the important input domain of $f_3$ by evaluating $f_1$ and $f_2$ exhaustively over the entire domain of the global input thanks to the cheap cost of the synthetic models. However, in practice it is rarely possible to gain perfect knowledge about the important input domain of a computer model or feasible to evaluate models thoroughly without constraints. 

\begin{algorithm}[htbp]
\caption{Adaptive design for the synthetic system illustrated in Section~\ref{sec:experiments}}
\label{alg:design}
\begin{algorithmic}[1]
\STATE{Choose $K$ number of enrichment (i.e., iterations) to the initial design.}
\FOR{$k=1,\dots,K$}
\STATE\label{alg:opt}{Find $\widehat{\mathbf{x}}$ and $\widehat{l}$ such that
\begin{equation*}
(\widehat{\mathbf{x}},\,\widehat{l})=\argmax_{\mathbf{x},\,l\in\{1,\,2\}}V_l(\mathbf{x}),
\end{equation*}
where $\mathbf{x}=(x_1,\,x_2)$, and $V_1(\mathbf{x})$ and $V_2(\mathbf{x})$ respectively are contributions of $\widehat{e}_{1}$ (i.e., GaSP emulators $\widehat{f}_1$ and $\widehat{f}_2$ in the first layer) and $\widehat{f}_3$ to the variance of the linked GaSP;}
\IF{$\widehat{l}=1$}
\STATE{Compute $V_{1k}(\widehat{\mathbf{x}})$ for $k\in\{1,\,2\}$ according to Proposition~\ref{prop:variance}, where $V_{1k}(\widehat{\mathbf{x}})$ is the contribution of $\widehat{f}_k$ to the variance of linked GaSP;}
\IF{$V_{11}(\widehat{\mathbf{x}})>V_{12}(\widehat{\mathbf{x}})$}
\STATE{Enrich the training points for $\widehat{f}_1$ by evaluating $f_1$ at the input position $\widehat{x}_1$;}
\ELSE
\STATE{Enrich the training points for $\widehat{f}_2$ by evaluating $f_2$ at the input position $\widehat{x}_2$;}
\ENDIF
\ELSE
\STATE{Enrich the training points for $\widehat{f}_3$ by evaluating $f_3$ at the input position $(\mu_1(\widehat{x}_1),\,\mu_2(\widehat{x}_2))$, obtained by evaluating the predictive mean $\mu_1$ and $\mu_2$ of $\widehat{f}_1$ and $\widehat{f}_2$ at the input position $\widehat{x}_1$ and $\widehat{x}_2$, respectively;}
\ENDIF
\STATE{Update the GaSP emulator $\widehat{f}_1$, $\widehat{f}_2$ or $\widehat{f}_3$ and construct the linked GaSP.}
\ENDFOR
\end{algorithmic} 
\end{algorithm}

\begin{figure}[htbp]
\centering
\includegraphics[width=0.9\linewidth]{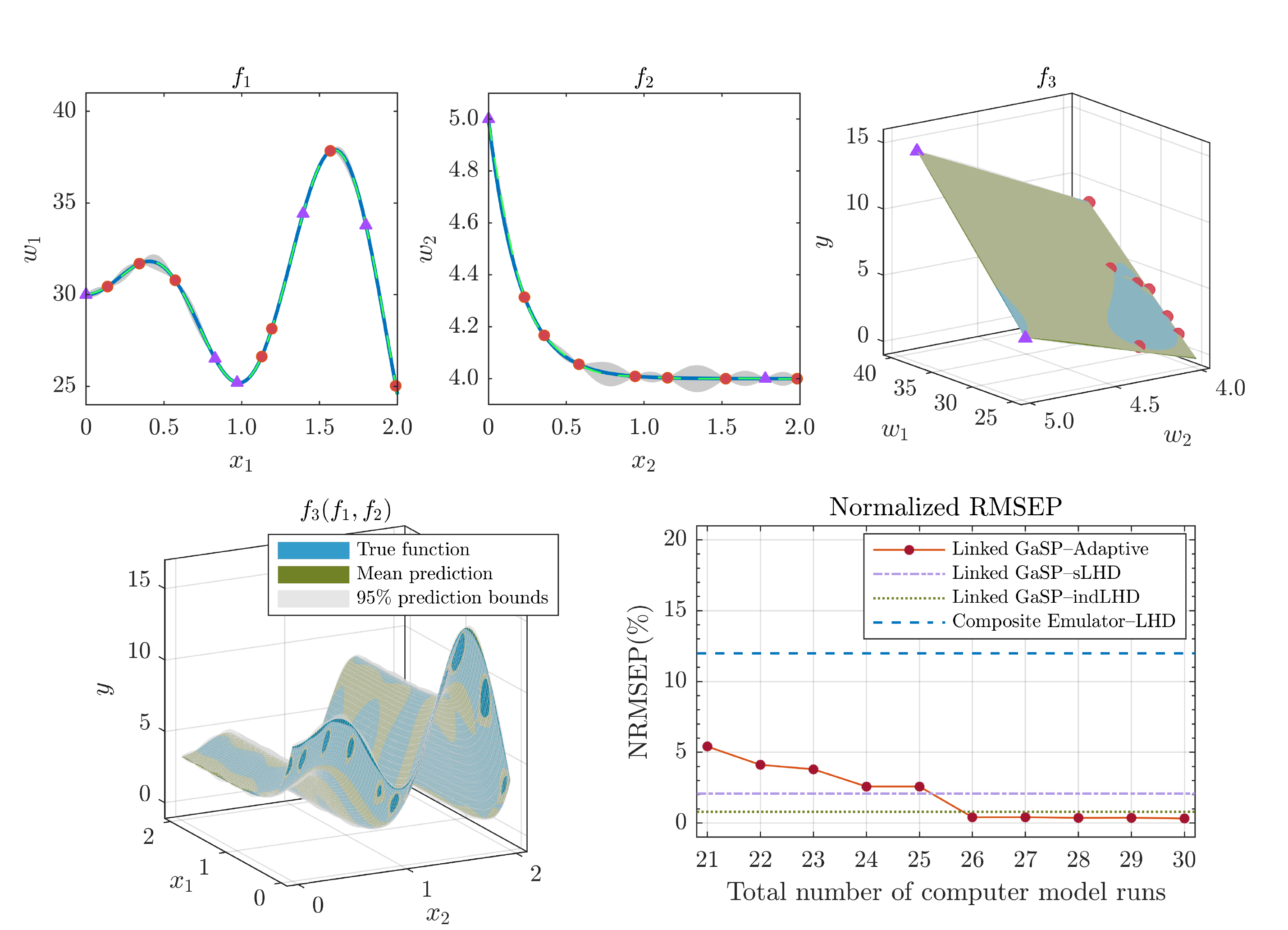}
\caption{The adaptive design for the synthetic experiment in Section~\ref{sec:experiments}. \emph{(Top-left)} GaSP emulator of $f_1$ ; \emph{(Top-middle)} GaSP emulator of $f_2$; \emph{(Top-right)} GaSP emulator of $f_3$; \emph{(Bottom-left)} linked GaSP of the system; \emph{(Bottom-right)} Comparison of NRMSEP between the linked GaSP with the adaptive design, the linked GaSP with the sequential LHD (sLHD), the linked GaSP with the independent LHD (indLHD), and the composite emulator with the LHD. The linked GaSP with the sLHD and the composite emulator are trained with $30$ computer runs (i.e., $10$ full runs of the entire system).The linked GaSP with the indLHD is trained with $10$ runs for each sub-model. The linked GaSP with the adaptive design is trained with $21$ initial computer model runs determined by the sLHD (i.e., $7$ runs of the whole system, corresponding to the filled circles in the top panels) and $9$ additional sub-model runs (corresponding to the filled triangles in the top panels) over $9$ iterations.}
\label{fig:adp_exp}
\end{figure}

Although the adaptive design is a desirable design strategy, it has its own limitations. Firstly, the adaptive design updates the GaSP emulator of one sub-model iteratively. Therefore, unlike the independent design, it does not allow sub-models of a system to run simultaneously during the experimental design. Beside, the adaptive design is still a sequential method because the input location at which the selected sub-model needs to run is determined by propagating the determined global input location through the GaSP emulators of those sub-models in lower layers. As a result, inaccurate GaSP emulators in lower layers may produce sub-optimal input positions to improve the GaSP emulators in higher layers. One thus need to implement the adaptive design with more iterations, and in turn spend more computational resources, to improve the linked GaSP sufficiently. Furthermore, the maximization problem involved in the adaptive design to search for the sub-model whose GaSP emulator needs to be updated is a challenging task especially when the global input dimension is high. Therefore, developing a fast and efficient searching algorithm is essential. Fortunately, the closed form expressions for the variance decomposition given in Proposition~\ref{prop:variance} render the exact evaluation of their derivatives respect to the input positions, thus many existing optimization algorithms (e.g., gradient ascent) could be applied. We leave this aspect as a future development without exploring further in this study.
 
\section{Discussion}
\label{sec:discussion}

The development of Theorem~\ref{thm:main} in Section~\ref{sec:extension} depends on Assumption~\ref{ass:2}, which asks for independence of input variables to the GaSP emulator of $g$ in the second layer. This independence assumption helps reduce analytical efforts in deriving the closed form mean and variance of the linked emulator. In addition, the consideration of dependence between input variables requires specification of their dependence structures, which can be a difficult task as careful dependence modeling, model training and predictions are needed. Nevertheless, ignoring the dependence structure between input variables feeding to the second layer can cause biased mean and variance of the linked emulator if the dependence is non-negligible. \cite{kyzyurova2018coupling} explore the impact of such dependence ignorance and conclude that in the case of Gaussian dependence under the squared exponential kernel, one could diagnose the significance of dependence by calculating the following ratios $r_k=\widehat{\gamma}^2_k/{\sigma^2_k}$ for all $k=1,\dots,d$, where $\widehat{\gamma}_k$ is the estimated range parameter of the $k$-th input to the GaSP emulator $\widehat{g}$. If $r_k$ is large (e.g., in the order of hundreds or thousands) for all $k$, the difference between the linked GaSPs with and without the dependence structure is then negligible. Note that given $\widehat{\gamma}^2_k$, $r_k$ increases as predictive variance $\sigma^2_k$ decreases. Thus, one could  safely neglect the impact of dependence by improving GaSP emulators in the feeding layer. We review these results in Section~\ref{sec:dependence} of supplementary materials. Since $r_k$ is calculated without the consideration of dependence and before invoking Theorem~\ref{thm:main}, it can be used as a measurement to determine whether one should consider the dependence before explicitly incorporating it to the emulation. 

However, $r_k$ may not be a valid measurement when kernels other than the squared exponential are used. It is also difficult in practice to have GaSP emulators producing sufficiently small predictive variances at the evaluated input positions to rule out the impact of dependence. Therefore, one may have to consider specifying the dependence structure between outputs of GaSP emulators from the feeding layer. One option for the dependence specification is to build multivariate GaSP emulators~\citep{rougier2009expert,fricker2013multivariate,zhang2015full}. However, existing literature on multivariate GaSP only consider the dependence among outputs from a single computer model, which means that in each layer of a system one has to treat all computer models, whose outputs are correlated, as a single model for the multivariate GaSP emulation, This is apparently an unpleasant feature because it reduces the benefit of system order reduction (i.e., GaSP emulators are constructed for individual computer models) offered by the linked GaSP emulation. A possible solution to this issue is to first build GaSP emulators ignoring the dependence and then model dependence structure separately, e.g., utilizing copulas~\citep{embrechts2003copula}. Nevertheless, one still need to conduct extra analytical efforts to derive more sophisticated closed form expressions for the mean and variance of linked emulator under the multivariate setting for different kernel choices.

Linked emulator gives the true distributional representation of coupled GaSP emulators of computer models in a system. Linked GaSP then serves as a Gaussian approximation to the analytically intractable linked emulator. The use of linked GaSP in replacement of linked emulator can be justified from two aspects. Firstly, with Gaussian distribution, one can construct closed form linked GaSP successively via the iterative procedure in Section~\ref{sec:feedforward}. Secondly, linked GaSP with its mean and variance matching to the linked emulator minimizes the Kullback–Leibler (KL) divergence (i.e., information loss) between the linked emulator and a Gaussian density~\citep{minka2013expectation}.  

The approximation accuracy of the linked GaSP to the linked emulator for a two layered system is explored in~\cite{kyzyurova2018coupling}, which indicate that the linked GaSP converges to the linked emulator when the predictive variances of GaSP emulators in the first layer reduce to zero. This statement is intuitive because GaSP emulators tend to be deterministic as their predictive variances drop. Consequently, the linked emulator decays to a Gaussian distribution that is equivalent to the corresponding linked GaSP. However, it is often not possible to ensure this condition for multi-layered systems, especially when systems are complex and the computational budget is limited. We explore provisionally the approximating performance of the linked GaSP in a three-layered synthetic system with a fairly small number of training points in Section~\ref{sec:approx} of supplementary materials. We found, and we also conjecture for systems with a moderate number of layers, that the linked GaSP approximates well the mean and variance of the linked emulator, while is unable to reconstruct sufficiently the full probabilistic distribution of the linked emulator. Therefore, the linked GaSP can be a good analytical replacement of a linked emulator for analysis, such as the history matching, where mean and variance are the key quantities of interest. However, if the full uncertainty description of an emulator is of concern (e.g., if tails are of specific interest), the linked GaSP may not be a fully adequate surrogate model.

Like all data-driven emulators, the linked GaSP is a simplified approximation to the underlying computer system, which can be both high-dimensional and extremely nonlinear. Thus, careful plans and implementations, such as computational budget allocation, design consideration and model validation, are essential for efficient emulation on systems of computer models. In addition, the accuracy of linked GaSPs is not only constrained by the assumptions listed in Section~\ref{sec:extension}, but also limited by those (e.g., stationarity) made for GaSP emulators. Therefore, further methodological and empirical advancements on both GaSP emulator and linked GaSP are required for robust uncertainty quantification of sophisticated real-world systems of computer models.

\section{Conclusion}
\label{sec:conclusion}
In this study, we generalize the linked GaSP  to  a class of Mat{\'e}rn kernels. The ability to use Mat{\'e}rn kernels is essential for wider applications of the linked GaSP on uncertainty quantification of systems of computer models.  The linked GaSP emulation can also be applied to any feed-forward systems with an iterative procedure. In combination with decoupling techniques, the linked GaSP can even be utilized for systems with internal loops.

The linked GaSP emulation can be further enhanced, in terms of the approximating accuracy and computational cost, via careful implementation of design strategies. We discuss pros and cons of several alternative designs, and introduce an adaptive design that improves the accuracy of the linked GaSP with reduced computational by allocating runs to different computer models in a system based on their heterogeneous functional complexity. The benefits of the adaptive design are illustrated via a synthetic example. Further refinements of the design and how it performs in real systems are directions worth exploring.

The linked GaSP outperforms the composite emulator by a ``divide-and-conquer'' strategy~\citep{kyzyurova2018coupling}, which converts the emulation of a bulky system into emulations of a number of simpler elements. However, when a single computer model dominates the functional complexity of the whole system the linked GaSP may not show a significant improvement over the composite emulator. Particularly, if the dimension of input to individual computer models is remarkably higher than that of global input, one might resort to dimension reduction techniques to construct GaSP emulators of individual computer models. Whether the benefits offered by the linked GaSP can overweight the approximation error induced by the dimension reduction methods needs to be studied in the future. Since the uncertainty quantification is now an integrated module in many research of multi-physics systems, one may consider split processes during the system development to facilitate surrogate modeling. 

Overall, we demonstrate both the effectiveness and efficiency of our new strategies to build linked GaSPs for systems of computer models. Another ambitious, but needed, task would be to investigate how our results can be exploited to emulate more complex feed-back coupled systems, such as climate models, than the one considered in this study.

\bibliographystyle{agsm}  
\bibliography{references}  

\begin{appendices}

\section{Closed Form Expressions}
\label{app:expression}
\subsection{Exponential Case}
\begin{align*}
 \xi_{ik}=&\exp\left\{\frac{\sigma^2_k+2\gamma_k\left(w^{\mathcal{T}}_{ik}-\mu_k\right)}{2\gamma_k^2}\right\}\Phi\left(\frac{\mu_A-w^{\mathcal{T}}_{ik}}{\sigma_k}\right)+\exp\left\{\frac{\sigma^2_k-2\gamma_k\left(w^{\mathcal{T}}_{ik}-\mu_k\right)}{2\gamma_k^2}\right\}\Phi\left(\frac{w^{\mathcal{T}}_{ik}-\mu_B}{\sigma_k}\right),\\
\zeta_{ijk}=&\begin{cases*}
	h_{\zeta}\left(w^{\mathcal{T}}_{ik},\,w^{\mathcal{T}}_{jk}\right), & $w^{\mathcal{T}}_{jk}\geq w^{\mathcal{T}}_{ik}$\;,\\
	h_{\zeta}\left(w^{\mathcal{T}}_{jk},\,w^{\mathcal{T}}_{ik}\right), & $w^{\mathcal{T}}_{jk}<w^{\mathcal{T}}_{ik}$\;,
\end{cases*}\\
\psi_{jk}=&\exp\left\{\frac{\sigma^2_k+2\gamma_k\left(w^{\mathcal{T}}_{jk}-\mu_k\right)}{2\gamma_k^2}\right\}\left[\mu_A\Phi\left(\frac{\mu_A-w^{\mathcal{T}}_{jk}}{\sigma_k}\right)+\frac{\sigma_k}{\sqrt{2\pi}}\exp\left\{-\frac{\left(w^{\mathcal{T}}_{jk}-\mu_A\right)^2}{2\sigma^2_k}\right\}\right]\\
    &-\exp\left\{\frac{\sigma^2_k-2\gamma_k\left(w^{\mathcal{T}}_{jk}-\mu_k\right)}{2\gamma_k^2}\right\}\left[\mu_B\Phi\left(\frac{w^{\mathcal{T}}_{jk}-\mu_B}{\sigma_k}\right)-\frac{\sigma_k}{\sqrt{2\pi}}\exp\left\{-\frac{\left(w^{\mathcal{T}}_{jk}-\mu_B\right)^2}{2\sigma^2_k}\right\}\right],
\end{align*}
where $\Phi(\cdot)$ denotes the cumulative density function of the standard normal;
\begin{align*}
    h_{\zeta}\left(x_1,\,x_2\right)=&\exp\left\{\frac{2\sigma^2_k+\gamma_k\left(x_1+x_2-2\mu_k\right)}{\gamma_k^2}\right\}\Phi\left(\frac{\mu_C-x_2}{\sigma_k}\right)\\
      &+\exp\left\{-\frac{x_2-x_1}{\gamma_k}\right\}\left[\Phi\left(\frac{x_2-\mu_k}{\sigma_k}\right)-\Phi\left(\frac{x_1-\mu_k}{\sigma_k}\right)\right]\\
      &+\exp\left\{\frac{2\sigma^2_k-\gamma_k\left(x_1+x_2-2\mu_k\right)}{\gamma_k^2}\right\}\Phi\left(\frac{x_1-\mu_D}{\sigma_k}\right);
\end{align*}
and
\begin{equation*}
 \mu_A=\mu_k-\frac{\sigma^2_k}{\gamma_k},\quad
 \mu_B=\mu_k+\frac{\sigma^2_k}{\gamma_k},\quad
 \mu_C=\mu_k-\frac{2\sigma^2_k}{\gamma_k}\quad\mathrm{and}\quad
 \mu_D=\mu_k+\frac{2\sigma^2_k}{\gamma_k}.
\end{equation*}
For notational convenience, in the above result we replace the index variable $l$ in the subscript of $\psi_{jl}$ by $k$, and $\mu_k(\mathbf{x}_k)$ and $\sigma_k(\mathbf{x}_k)$ by $\mu_k$ and $\sigma_k$. This change of notation is also applied in the remainder of the supplement. 

\subsection{Squared Exponential Case}
\begin{align*}
     \xi_{ik}&=\frac{1}{\sqrt{1+2\sigma^2_k/\gamma^2_k}}\exp\left\{-\frac{\left(\mu_k-w^{\mathcal{T}}_{ik}\right)^2}{2\sigma^2_k+\gamma^2_k}\right\},\\
    \zeta_{ijk}&=\frac{1}{\sqrt{1+4\sigma^2_k/\gamma^2_k}}\exp\left\{-\frac{\left(\frac{w^{\mathcal{T}}_{ik}+w^{\mathcal{T}}_{jk}}{2}-\mu_k\right)^2}{\gamma^2_k/2+2\sigma^2_k}-\frac{\left(w^{\mathcal{T}}_{ik}-w^{\mathcal{T}}_{jk}\right)^2}{2\gamma^2_k}\right\},\\
    \psi_{jk}&=\frac{1}{\sqrt{1+2\sigma^2_k/\gamma^2_k}}\exp\left\{-\frac{\left(\mu_k-w^{\mathcal{T}}_{jk}\right)^2}{2\sigma^2_k+\gamma^2_k}\right\}\frac{2\sigma^2_kw^{\mathcal{T}}_{jk}+\gamma^2_k\mu_k}{2\sigma^2_k+\gamma^2_k}.
\end{align*}

\subsection{Mat\'{e}rn-1.5 Case}
\begin{align*}
        \xi_{ik}=&\exp\left\{\frac{3\sigma^2_k+2\sqrt{3}\gamma_k\left(w^{\mathcal{T}}_{ik}-\mu_k\right)}{2\gamma_k^2}\right\}\left[\mathbf{E}^\top_1\boldsymbol{\Lambda}_{11}\Phi\left(\frac{\mu_A-w^{\mathcal{T}}_{ik}}{\sigma_k}\right)+\mathbf{E}^\top_1\boldsymbol{\Lambda}_{12}\frac{\sigma_k}{\sqrt{2\pi}}\exp\left\{-\frac{(w^{\mathcal{T}}_{ik}-\mu_A)^2}{2\sigma^2_k}\right\}\right]\\
         &+\exp\left\{\frac{3\sigma^2_k-2\sqrt{3}\gamma_k\left(w^{\mathcal{T}}_{ik}-\mu_k\right)}{2\gamma_k^2}\right\}\left[\mathbf{E}^\top_2\boldsymbol{\Lambda}_{21}\Phi\left(\frac{w^{\mathcal{T}}_{ik}-\mu_B}{\sigma_k}\right)+\mathbf{E}^\top_2\boldsymbol{\Lambda}_{22}\frac{\sigma_k}{\sqrt{2\pi}}\exp\left\{-\frac{(w^{\mathcal{T}}_{ik}-\mu_B)^2}{2\sigma^2_k}\right\}\right],\\
        \zeta_{ijk}=&\begin{cases*}
	h_{\zeta}\left(w^{\mathcal{T}}_{ik},\,w^{\mathcal{T}}_{jk}\right), & $w^{\mathcal{T}}_{jk}\geq w^{\mathcal{T}}_{ik}$\;,\\
	h_{\zeta}\left(w^{\mathcal{T}}_{jk},\,w^{\mathcal{T}}_{ik}\right), & $w^{\mathcal{T}}_{jk}<w^{\mathcal{T}}_{ik}$\;,\\
\end{cases*}\\
\psi_{jk}=&\exp\left\{\frac{3\sigma^2_k+2\sqrt{3}\gamma_k\left(w^{\mathcal{T}}_{jk}-\mu_k\right)}{2\gamma_k^2}\right\}\left[\mathbf{E}^\top_1\boldsymbol{\Lambda}_{61}\Phi\left(\frac{\mu_A-w^{\mathcal{T}}_{jk}}{\sigma_k}\right)+\mathbf{E}^\top_1\boldsymbol{\Lambda}_{62}\frac{\sigma_k}{\sqrt{2\pi}}\exp\left\{-\frac{(w^{\mathcal{T}}_{jk}-\mu_A)^2}{2\sigma^2_k}\right\}\right]\\
         &-\exp\left\{\frac{3\sigma^2_k-2\sqrt{3}\gamma_k\left(w^{\mathcal{T}}_{jk}-\mu_k\right)}{2\gamma_k^2}\right\}\left[\mathbf{E}^\top_2\boldsymbol{\Lambda}_{71}\Phi\left(\frac{w^{\mathcal{T}}_{jk}-\mu_B}{\sigma_k}\right)+\mathbf{E}^\top_2\boldsymbol{\Lambda}_{72}\frac{\sigma_k}{\sqrt{2\pi}}\exp\left\{-\frac{(w^{\mathcal{T}}_{jk}-\mu_B)^2}{2\sigma^2_k}\right\}\right],
\end{align*}
where
\begin{align*}
h_{\zeta}\left(x_1,\,x_2\right)=&\exp\left\{\frac{6\sigma_k^2+\sqrt{3}\gamma_k\left(x_1+x_2-2\mu_k\right)}{\gamma_k^2}\right\}\nonumber\\
        &\qquad\times\left[\mathbf{E}^\top_3\boldsymbol{\Lambda}_{31}\Phi\left(\frac{\mu_C-x_2}{\sigma_k}\right)+\mathbf{E}^\top_3\boldsymbol{\Lambda}_{32}\frac{\sigma_k}{\sqrt{2\pi}}\exp\left\{-\frac{(x_2-\mu_{C})^2}{2\sigma^2_k}\right\}\right]\nonumber\\
        &+\exp\left\{-\frac{\sqrt{3}\left(x_2-x_1\right)}{\gamma_k}\right\}\Bigg[\mathbf{E}^\top_4\boldsymbol{\Lambda}_{41}\left(\Phi\left(\frac{x_2-\mu_k}{\sigma_k}\right)-\Phi\left(\frac{x_1-\mu_k}{\sigma_k}\right)\right)\nonumber\\
        &\qquad+\mathbf{E}^\top_4\boldsymbol{\Lambda}_{42}\frac{\sigma_k}{\sqrt{2\pi}}\exp\left\{-\frac{(x_1-\mu_k)^2}{2\sigma^2_k}\right\}-\mathbf{E}^\top_4\boldsymbol{\Lambda}_{43}\frac{\sigma_k}{\sqrt{2\pi}}\exp\left\{-\frac{(x_2-\mu_k)^2}{2\sigma^2_k}\right\}\Bigg]\nonumber\\
        &+\exp\left\{\frac{6\sigma_k^2-\sqrt{3}\gamma_k\left(x_1+x_2-2\mu_k\right)}{\gamma_k^2}\right\}\nonumber\\
        &\qquad\times\left[\mathbf{E}^\top_5\boldsymbol{\Lambda}_{51}\Phi\left(\frac{x_1-\mu_D}{\sigma_k}\right)+\mathbf{E}^\top_5\boldsymbol{\Lambda}_{52}\frac{\sigma_k}{\sqrt{2\pi}}\exp\left\{-\frac{(x_1-\mu_{D})^2}{2\sigma^2_k}\right\}\right]
    \end{align*}
and
\begin{itemize}
         \item $\boldsymbol{\Lambda}_{11}=[1,\,\mu_A]^\top$, $\boldsymbol{\Lambda}_{12}=[0,\,1]^\top$, $\boldsymbol{\Lambda}_{21}=[1,\,-\mu_B]^\top$ and $\boldsymbol{\Lambda}_{22}=[0,\,1]^\top$;
         \item $\boldsymbol{\Lambda}_{31}=[1,\,\mu_C,\,\mu_C^2+\sigma^2_k]^\top$ and $\boldsymbol{\Lambda}_{32}=[0,\,1,\,\mu_C+x_2]^\top$; 
         \item $\boldsymbol{\Lambda}_{41}=[1,\,\mu_k,\,\mu_k^2+\sigma^2_k]^\top$, $\boldsymbol{\Lambda}_{42}=[0,\,1,\,\mu_k+x_1]^\top$ and $\boldsymbol{\Lambda}_{43}=[0,\,1,\,\mu_k+x_2]^\top$;
         \item $\boldsymbol{\Lambda}_{51}=[1,\,-\mu_D,\,\mu_D^2+\sigma^2_k]^\top$ and $\boldsymbol{\Lambda}_{52}=[0,\,1,\,-\mu_D-x_1]^\top\,$;
         \item $\boldsymbol{\Lambda}_{61}=[\mu_A,\,\mu^2_A+\sigma^2_k]^\top$ and $\boldsymbol{\Lambda}_{62}=[1,\,\mu_A+w^{\mathcal{T}}_{jk}]^\top$;
         \item $\boldsymbol{\Lambda}_{71}=[-\mu_B,\,\mu^2_B+\sigma^2_k]^\top$ and $\boldsymbol{\Lambda}_{72}=[1,\,-\mu_B-w^{\mathcal{T}}_{jk}]^\top$;
         \item $ \mathbf{E}_1=\left[1-\dfrac{\sqrt{3}w^{\mathcal{T}}_{ik}}{\gamma_k},\,\dfrac{\sqrt{3}}{\gamma_k}\right]^\top$ and $\mathbf{E}_2=\left[1+\dfrac{\sqrt{3}w^{\mathcal{T}}_{ik}}{\gamma_k},\,\dfrac{\sqrt{3}}{\gamma_k}\right]^\top$;
         \item $\mathbf{E}_3=\left[1+\dfrac{3x_1x_2-\sqrt{3}\gamma_k\left(x_1+x_2\right)}{\gamma_k^2},\,\dfrac{2\sqrt{3}\gamma_k-3\left(x_1+x_2\right)}{\gamma_k^2},\,\dfrac{3}{\gamma_k^2}\right]^\top$;
         \item $\mathbf{E}_4=\left[1+\dfrac{\sqrt{3}\gamma_k\left(x_2-x_1\right)-3x_1x_2}{\gamma_k^2},\,\dfrac{3\left(x_1+x_2\right)}{\gamma_k^2},\,-\dfrac{3}{\gamma_k^2}\right]^\top$;
         \item $\mathbf{E}_5=\left[1+\dfrac{3x_1x_2+\sqrt{3}\gamma_k\left(x_1+x_2\right)}{\gamma_k^2},\,\dfrac{2\sqrt{3}\gamma_k+3\left(x_1+x_2\right)}{\gamma_k^2},\,\dfrac{3}{\gamma_k^2}\right]^\top$;
         \item $\mu_A=\mu_k-\dfrac{\sqrt{3}\sigma^2_k}{\gamma_k}$, $\mu_B=\mu_k+\dfrac{\sqrt{3}\sigma^2_k}{\gamma_k}$, $\mu_{C}=\mu_k-\dfrac{2\sqrt{3}{\sigma^2_k}}{\gamma_k}$, $\mu_D=\mu_k+\dfrac{2\sqrt{3}{\sigma^2_k}}{\gamma_k}$.
\end{itemize}

\subsection{Mat\'{e}rn-2.5 Case}

\begin{align*}
\xi_{ik}=&\exp\left\{\frac{5\sigma^2_k+2\sqrt{5}\gamma_k\left(w^{\mathcal{T}}_{ik}-\mu_k\right)}{2\gamma_k^2}\right\}\left[\mathbf{E}^\top_1\boldsymbol{\Lambda}_{11}\Phi\left(\frac{\mu_A-w^{\mathcal{T}}_{ik}}{\sigma_k}\right)+\mathbf{E}^\top_1\boldsymbol{\Lambda}_{12}\frac{\sigma_k}{\sqrt{2\pi}}\exp\left\{-\frac{(w^{\mathcal{T}}_{ik}-\mu_A)^2}{2\sigma^2_k}\right\}\right]\\
         &+\exp\left\{\frac{5\sigma^2_k-2\sqrt{5}\gamma_k\left(w^{\mathcal{T}}_{ik}-\mu_k\right)}{2\gamma_k^2}\right\}\left[\mathbf{E}^\top_2\boldsymbol{\Lambda}_{21}\Phi\left(\frac{w^{\mathcal{T}}_{ik}-\mu_B}{\sigma_k}\right)+\mathbf{E}^\top_2\boldsymbol{\Lambda}_{22}\frac{\sigma_k}{\sqrt{2\pi}}\exp\left\{-\frac{(w^{\mathcal{T}}_{ik}-\mu_B)^2}{2\sigma^2_k}\right\}\right],\\[10pt]
\zeta_{ijk}=&\begin{cases*}
	h_{\zeta}\left(w^{\mathcal{T}}_{ik},\,w^{\mathcal{T}}_{jk}\right), & $w^{\mathcal{T}}_{jk}\geq w^{\mathcal{T}}_{ik}$\;,\\
	h_{\zeta}\left(w^{\mathcal{T}}_{jk},\,w^{\mathcal{T}}_{ik}\right), & $w^{\mathcal{T}}_{jk}<w^{\mathcal{T}}_{ik}$\;,\\
    \end{cases*}\\
 \psi_{jk}=&\exp\left\{\frac{5\sigma^2_k+2\sqrt{5}\gamma_k\left(w^{\mathcal{T}}_{jk}-\mu_k\right)}{2\gamma_k^2}\right\}\left[\mathbf{E}^\top_1\boldsymbol{\Lambda}_{61}\Phi\left(\frac{\mu_A-w^{\mathcal{T}}_{jk}}{\sigma_k}\right)+\mathbf{E}^\top_1\boldsymbol{\Lambda}_{62}\frac{\sigma_k}{\sqrt{2\pi}}\exp\left\{-\frac{(w^{\mathcal{T}}_{jk}-\mu_A)^2}{2\sigma^2_k}\right\}\right]\\
 &-\exp\left\{\frac{5\sigma^2_k-2\sqrt{5}\gamma_k\left(w^{\mathcal{T}}_{jk}-\mu_k\right)}{2\gamma_k^2}\right\}\left[\mathbf{E}^\top_2\boldsymbol{\Lambda}_{71}\Phi\left(\frac{w^{\mathcal{T}}_{jk}-\mu_B}{\sigma_k}\right)+\mathbf{E}^\top_2\boldsymbol{\Lambda}_{72}\frac{\sigma_k}{\sqrt{2\pi}}\exp\left\{-\frac{(w^{\mathcal{T}}_{jk}-\mu_B)^2}{2\sigma^2_k}\right\}\right],
\end{align*}
where
\begin{align*}
h_{\zeta}\left(x_1,\,x_2\right)=&\exp\left\{\frac{10\sigma_k^2+\sqrt{5}\gamma_k\left(x_1+x_2-2\mu_k\right)}{\gamma_k^2}\right\}\nonumber\\
        &\qquad\times\left[\mathbf{E}^\top_3\boldsymbol{\Lambda}_{31}\Phi\left(\frac{\mu_C-x_2}{\sigma_k}\right)+\mathbf{E}^\top_3\boldsymbol{\Lambda}_{32}\frac{\sigma_k}{\sqrt{2\pi}}\exp\left\{-\frac{(x_2-\mu_{C})^2}{2\sigma^2_k}\right\}\right]\nonumber\\
        &+\exp\left\{-\frac{\sqrt{5}\left(x_2-x_1\right)}{\gamma_k}\right\}\Bigg[\mathbf{E}^\top_4\boldsymbol{\Lambda}_{41}\left(\Phi\left(\frac{x_2-\mu_k}{\sigma_k}\right)-\Phi\left(\frac{x_1-\mu_k}{\sigma_k}\right)\right)\nonumber\\
        &\qquad+\mathbf{E}^\top_4\boldsymbol{\Lambda}_{42}\frac{\sigma_k}{\sqrt{2\pi}}\exp\left\{-\frac{(x_1-\mu_k)^2}{2\sigma^2_k}\right\}-\mathbf{E}^\top_4\boldsymbol{\Lambda}_{43}\frac{\sigma_k}{\sqrt{2\pi}}\exp\left\{-\frac{(x_2-\mu_k)^2}{2\sigma^2_k}\right\}\Bigg]\nonumber\\
        &+\exp\left\{\frac{10\sigma_k^2-\sqrt{5}\gamma_k\left(x_1+x_2-2\mu_k\right)}{\gamma_k^2}\right\}\nonumber\\
        &\qquad\times\left[\mathbf{E}^\top_5\boldsymbol{\Lambda}_{51}\Phi\left(\frac{x_1-\mu_D}{\sigma_k}\right)+\mathbf{E}^\top_5\boldsymbol{\Lambda}_{52}\frac{\sigma_k}{\sqrt{2\pi}}\exp\left\{-\frac{(x_1-\mu_{D})^2}{2\sigma^2_k}\right\}\right]
\end{align*}
and
\begin{itemize}
    \item $\boldsymbol{\Lambda}_{11}=[1,\,\mu_A,\,\mu^2_A+\sigma^2_k]^\top$ and $\boldsymbol{\Lambda}_{12}=[0,\,1,\,\mu_A+w^{\mathcal{T}}_{ik}]^\top$;
    \item $\boldsymbol{\Lambda}_{21}=[1,\,-\mu_B,\,\mu^2_B+\sigma^2_k]^\top$ and $\boldsymbol{\Lambda}_{22}=[0,\,1,\,-\mu_B-w^{\mathcal{T}}_{ik}]^\top$;
    \item $\boldsymbol{\Lambda}_{31}=[1,\,\mu_C,\,\mu_C^2+\sigma^2_k,\,\mu_C^3+3\sigma^2_k\mu_C,\,\mu_C^4+6\sigma^2_k\mu_C^2+3\sigma_k^4]^\top\,$;
    \item $\boldsymbol{\Lambda}_{32}=[0,\,1,\,\mu_C+x_2,\,\mu_C^2+2\sigma^2_k+x_2^2+\mu_Cx_2,\,\mu_C^3+x_2^3+x_2\mu_C^2+\mu_Cx_2^2+3\sigma_k^2x_2+5\sigma_k^2\mu_C]^\top\,$;
    \item $\boldsymbol{\Lambda}_{41}=[1,\,\mu_k,\,\mu_k^2+\sigma^2_k,\,\mu_k^3+3\sigma^2_k\mu_k,\,\mu_k^4+6\sigma^2_k\mu_k^2+3\sigma_k^4]^\top\,$;
    \item $\boldsymbol{\Lambda}_{42}=[0,\,1,\,\mu_k+x_1,\,\mu_k^2+2\sigma^2_k+x_1^2+\mu_kx_1,\,\mu_k^3+x_1^3+x_1\mu_k^2+\mu_kx_1^2+3\sigma_k^2x_1+5\sigma_k^2\mu_k]^\top\,$;
    \item $\boldsymbol{\Lambda}_{43}=[0,\,1,\,\mu_k+x_2,\,\mu_k^2+2\sigma^2_k+x_2^2+\mu_kx_2,\,\mu_k^3+x_2^3+x_2\mu_k^2+\mu_kx_2^2+3\sigma_k^2x_2+5\sigma_k^2\mu_k]^\top\,$;
    \item $\boldsymbol{\Lambda}_{51}=[1,\,-\mu_D,\,\mu_D^2+\sigma^2_k,\,-\mu_D^3-3\sigma^2_k\mu_D,\,\mu_D^4+6\sigma^2_k\mu_D^2+3\sigma_k^4]^\top\,$;
    \item $\boldsymbol{\Lambda}_{52}=[0,\,1,\,-\mu_D-x_1,\,\mu_D^2+2\sigma^2_k+x_1^2+\mu_Dx_1,\,-\mu_D^3-x_1^3-x_1\mu_D^2-\mu_Dx_1^2-3\sigma_k^2x_1-5\sigma_k^2\mu_D]^\top\,$;
    \item $\boldsymbol{\Lambda}_{61}=[\mu_A,\,\mu^2_A+\sigma^2_k,\,\mu^3_A+3\sigma_k^2\mu_A]^\top\,$;
    \item $\boldsymbol{\Lambda}_{62}=[1,\,\mu_A+w^{\mathcal{T}}_{jk},\,\mu^2_A+2\sigma^2_k+\left(w^{\mathcal{T}}_{jk}\right)^2+\mu_Aw^{\mathcal{T}}_{jk}]^\top\,$;
    \item $\boldsymbol{\Lambda}_{71}=[-\mu_B,\,\mu^2_B+\sigma^2_k,\,-\mu^3_B-3\sigma_k^2\mu_B]^\top\,$;
    \item $\boldsymbol{\Lambda}_{72}=[1,\,-\mu_B-w^{\mathcal{T}}_{jk},\,\mu^2_B+2\sigma^2_k+\left(w^{\mathcal{T}}_{jk}\right)^2+\mu_Bw^{\mathcal{T}}_{jk}]^\top\,$;
    \item $\mathbf{E}_1=\left[1-\dfrac{\sqrt{5}w^{\mathcal{T}}_{ik}}{\gamma_k}+\dfrac{5\left(w^{\mathcal{T}}_{ik}\right)^2}{3\gamma_k^2},\,\dfrac{\sqrt{5}}{\gamma_k}-\dfrac{10w^{\mathcal{T}}_{ik}}{3\gamma_k^2},\,\dfrac{5}{3\gamma_k^2}\right]^\top$;
    \item $\mathbf{E}_2=\left[1+\dfrac{\sqrt{5}w^{\mathcal{T}}_{ik}}{\gamma_k}+\dfrac{5\left(w^{\mathcal{T}}_{ik}\right)^2}{3\gamma_k^2},\,\dfrac{\sqrt{5}}{\gamma_k}+\dfrac{10w^{\mathcal{T}}_{ik}}{3\gamma_k^2},\,\dfrac{5}{3\gamma_k^2}\right]^\top$;
    \item $\mathbf{E}_3=[E_{30},\,E_{31},\,E_{32},\,E_{33},\,E_{34}]^\top\,$;
    \item $\mathbf{E}_4=[E_{40},\,E_{41},\,E_{42},\,E_{43},\,E_{44}]^\top\,$;
    \item $\mathbf{E}_5=[E_{50},\,E_{51},\,E_{52},\,E_{53},\,E_{54}]^\top\,$;
    \item 
    $\begin{aligned}[t]
    E_{30}=&1+\dfrac{25x_1^2x_2^2-3\sqrt{5}\left(3\gamma_k^3+5\gamma_kx_1x_2\right)\left(x_1+x_2\right)+15\gamma_k^2\left(x_1^2+x_2^2+3x_1x_2\right)}{9\gamma_k^4}\\
    E_{31}=&\dfrac{18\sqrt{5}\gamma_k^3+15\sqrt{5}\gamma_k\left(x_1^2+x_2^2\right)-(75\gamma_k^2+50x_1x_2)\left(x_1+x_2\right)+60\sqrt{5}\gamma_kx_1x_2}{9\gamma_k^4}\\
    E_{32}=&\dfrac{5\left[5x_1^2+5x_2^2+15\gamma_k^2-9\sqrt{5}\gamma_k\left(x_1+x_2\right)+20x_1x_2\right]}{9\gamma_k^4}\\
    E_{33}=&\dfrac{10\left(3\sqrt{5}\gamma_k-5x_1-5x_2\right)}{9\gamma_k^4}\quad\mathrm{and}\quad
    E_{34}=\dfrac{25}{9\gamma_k^4};
    \end{aligned}$
    \item
    $\begin{aligned}[t]
    E_{40}=&1+\dfrac{25x_1^2x_2^2+3\sqrt{5}\left(3\gamma_k^3-5\gamma_kx_1x_2\right)\left(x_2-x_1\right)+15\gamma_k^2\left(x_1^2+x_2^2-3x_1x_2\right)}{9\gamma_k^4}\\
    E_{41}=&\frac{5\left[3\sqrt{5}\gamma_k\left(x_2^2-x_1^2\right)+3\gamma_k^2\left(x_1+x_2\right)-10x_1x_2\left(x_1+x_2\right)\right]}{9\gamma_k^4}\\
    E_{42}=&\dfrac{5\left[5x_1^2+5x_2^2-3\gamma_k^2-3\sqrt{5}\gamma_k\left(x_2-x_1\right)+20x_1x_2\right]}{9\gamma_k^4}\\
    E_{43}=&-\dfrac{50\left(x_1+x_2\right)}{9\gamma_k^4}\quad\mathrm{and}\quad
    E_{44}=\dfrac{25}{9\gamma_k^4};
    \end{aligned}$
    \item
    $\begin{aligned}[t]
    E_{50}=&1+\dfrac{25x_1^2x_2^2+3\sqrt{5}\left(3\gamma_k^3+5\gamma_kx_1x_2\right)\left(x_1+x_2\right)+15\gamma_k^2\left(x_1^2+x_2^2+3x_1x_2\right)}{9\gamma_k^4}\\
    E_{51}=&\dfrac{18\sqrt{5}\gamma_k^3+15\sqrt{5}\gamma_k\left(x_1^2+x_2^2\right)+(75\gamma_k^2+50x_1x_2)\left(x_1+x_2\right)+60\sqrt{5}\gamma_kx_1x_2}{9\gamma_k^4}\\
    E_{52}=&\dfrac{5\left[5x_1^2+5x_2^2+15\gamma_k^2+9\sqrt{5}\gamma_k\left(x_1+x_2\right)+20x_1x_2\right]}{9\gamma_k^4}\\
    E_{53}=&\dfrac{10\left(3\sqrt{5}\gamma_k+5x_1+5x_2\right)}{9\gamma_k^4}\quad\mathrm{and}\quad
    E_{54}=\dfrac{25}{9\gamma_k^4};
    \end{aligned}$
    \item
    $\mu_A=\mu_k-\dfrac{\sqrt{5}\sigma^2_k}{\gamma_k}$, $\mu_B=\mu_k+\dfrac{\sqrt{5}\sigma^2_k}{\gamma_k}$, $\mu_C=\mu_k-\dfrac{2\sqrt{5}{\sigma^2_k}}{\gamma_k}$, $\mu_D=\mu_k+\dfrac{2\sqrt{5}{\sigma^2_k}}{\gamma_k}$. 
\end{itemize}
\end{appendices}

\pagebreak
\begin{center}
\label{supp}
\textbf{\Large Supplementary Materials}
\end{center}
\setcounter{equation}{0}
\setcounter{section}{0}
\makeatletter
\renewcommand{\thesection}{S.\arabic{section}}
\renewcommand\thesubsubsection{\arabic{subsubsection}}
\renewcommand{\theequation}{S\arabic{equation}} 

\section[An Example on the Deficiency of Independent Designs]{An Example on the Deficiency of Independent Designs}
\label{sec:ind_design}
In this section, we illustrate a scenario where the independent designs for the linked GaSP emulation can be problematic. Consider the computer system shown in Figure~\ref{fig:exp_supp}, which consists three computer models with the following analytical functional forms:
\begin{equation*}
f_1=0.5+0.5x\sin(10x),\; f_2=\exp(-10x),\; f_3=\sin\left(\frac{1}{(0.7w_1+0.3)(0.7w_2+0.3)}\right)
\end{equation*}
with $x\in[0,\,1]$.

\begin{figure}[htbp]
\centering
\begin{tikzpicture}[shorten >=1pt,->,draw=black!50, node distance=4.5cm]
    \tikzstyle{every pin edge}=[<-,shorten <=1pt]
    \tikzstyle{neuron}=[circle,fill=black!25,minimum size=17pt,inner sep=0pt]
    \tikzstyle{layer1}=[neuron, fill=green!50];
    \tikzstyle{layer2}=[neuron, fill=red!50];
    \tikzstyle{layer3}=[neuron, fill=blue!50];
    \tikzstyle{annot} = [text width=4em, text centered]

    \node[layer1, pin=left:$x$] (I-1) at (0,-0.75) {$f_1$};
    \node[layer1, pin=left:$x$] (I-2) at (0,-2) {$f_2$};

    \path[yshift=-1.375cm]
            node[layer2, pin={[pin edge={->}]right:$y$}] (H-1) at (4.5cm,0 cm) {$f_3$};
            \path [draw] (I-1) -- (H-1) node[font=\small,midway,fill=white,align=left,sloped] {$w_1$};
            \path [draw] (I-2) -- (H-1) node[font=\small,midway,fill=white,align=left,sloped] {$w_2$};

    \node[annot,above of=I-1, node distance=0.7cm] (hl) {Layer 1};
    \node[annot,right of=hl] {Layer 2};
\end{tikzpicture}
\caption{A synthetic computer system where $f_1$ and $f_2$ are two computer models with a common one-dimensional input but different scalar-valued outputs, and $f_3$ is a computer model with two-dimensional input and one-dimensional output.}
\label{fig:exp_supp}
\end{figure}
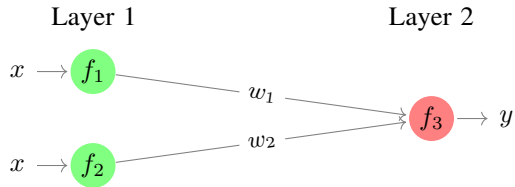

We construct the linked GaSP by building GaSP emulators of individual computer models independently with their own one-shot LHD. It can be seen from Figure~\ref{fig:exp_supp_ind} that ignoring the structural dependence causes a poor LHD of $f_3$, where only one design point falls close to the input space of $f_3$ (see the solid trajectory in Figure~\subref*{fig:exp_supp_ind_f3}) that is significant to the global output, whereas the rest of design points are exploring regions that are insignificant to the global output. As a result, most of the computational resources are wasted and the resulting linked GaSP (see Figure~\subref*{fig:exp_supp_ind_emulator}) is unsatisfactory. It is worth noting that when implementing the LHD for $f_3$ we assume that we have perfect knowledge about the ranges of $w_1$ and $w_2$ that are produced by $f_1$ and $f_2$ (i.e., $w_1\in[0,1]$ and $w_2\in[0,1]$). However, it is often impossible in practice to have good prior knowledge about these ranges and therefore independent designs can result in excessive computational efforts when the input ranges are set too wide or an inadequate linked GaSP when the input ranges are set to narrow. All these mentioned issues related to independent designs could become severer when the input dimensions of individual computer models become high. 

For comparison, Figure~\ref{fig:exp_supp_dep} gives the linked GaSP constructed using the sequential LHD, where the design of $f_3$ is determined by propagating the one-shot LHD of the global input $x$ through $f_1$ and $f_2$. It is apparent that by taking the system structure into account, the design for $f_3$ only explores the region that is significant to the global out (i.e., all training points in Figure~\subref*{fig:exp_supp_dep_f3} fall on the solid trajectory). Consequently, the resulting linked GaSP (see Figure~\subref*{fig:exp_supp_dep_emulator}) provides a much better approximation to the underlying system. 

\begin{figure}[htbp]
\centering 
\subfloat[$f_1$]{\label{fig:exp_supp_ind_f1}\includegraphics[width=0.25\linewidth]{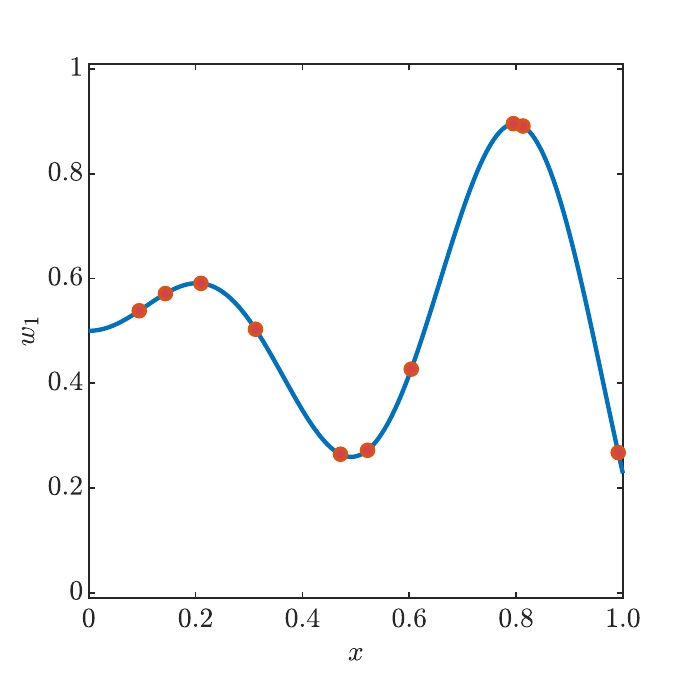}} 
\subfloat[$f_2$]{\label{fig:exp_supp_ind_f2}\includegraphics[width=0.25\linewidth]{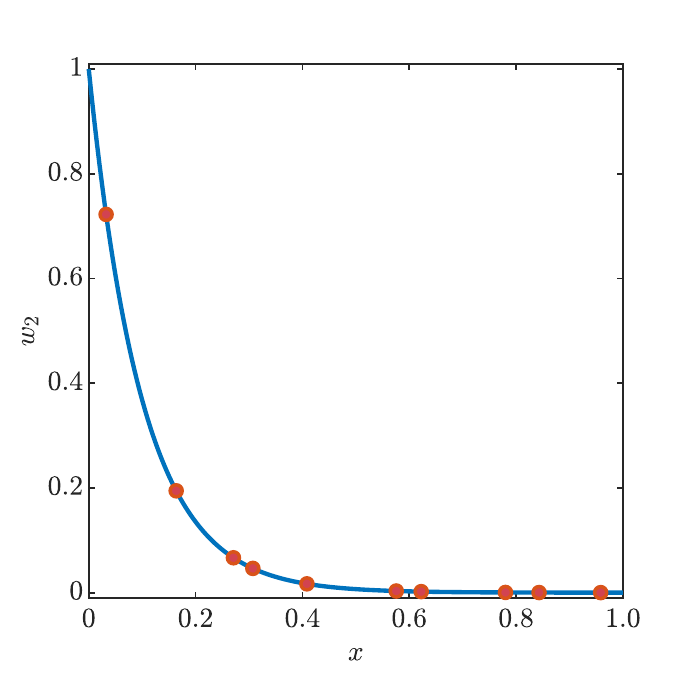}}
\subfloat[$f_3$]{\label{fig:exp_supp_ind_f3}\includegraphics[width=0.25\linewidth]{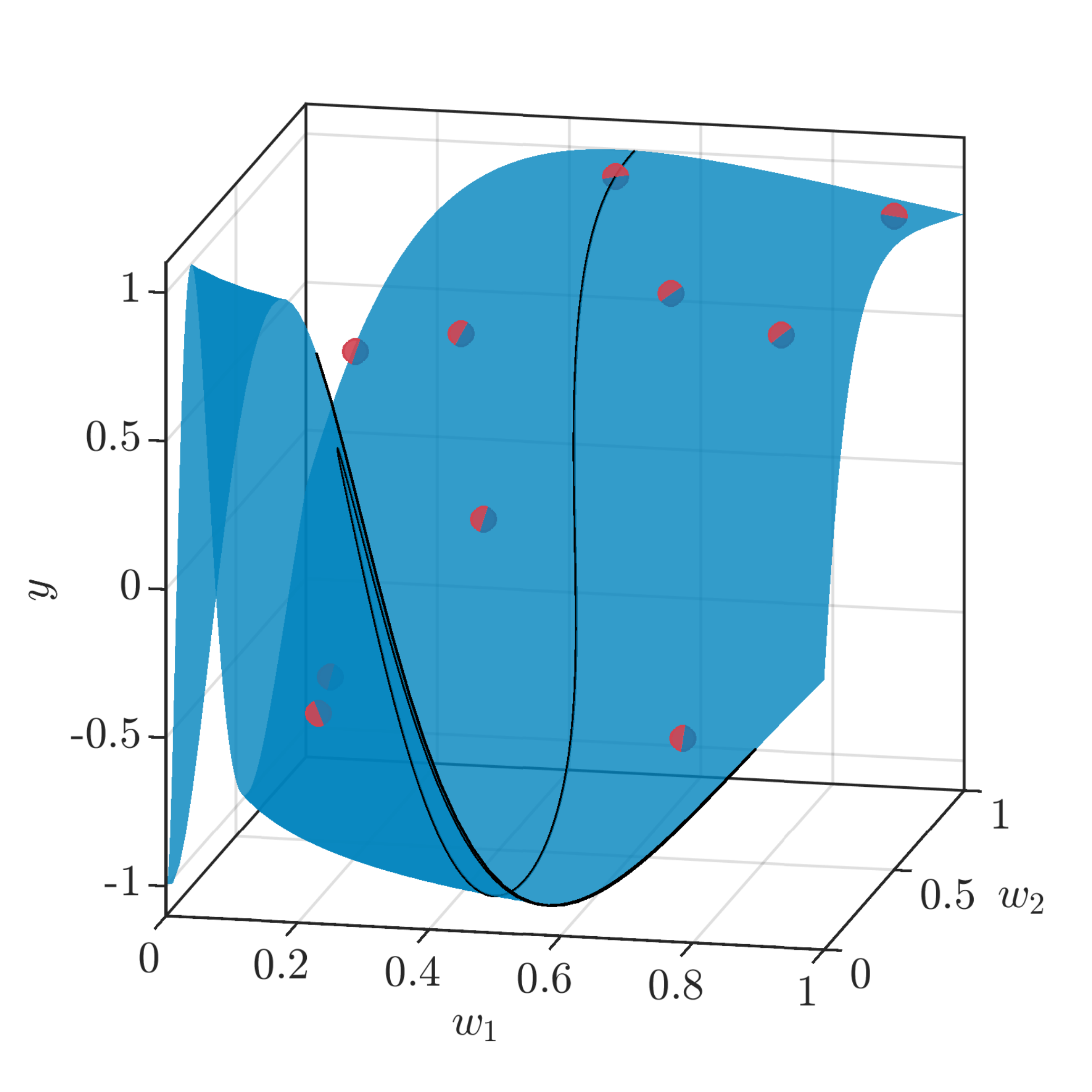}}
\subfloat[linked GaSP]{\label{fig:exp_supp_ind_emulator}\includegraphics[width=0.25\linewidth]{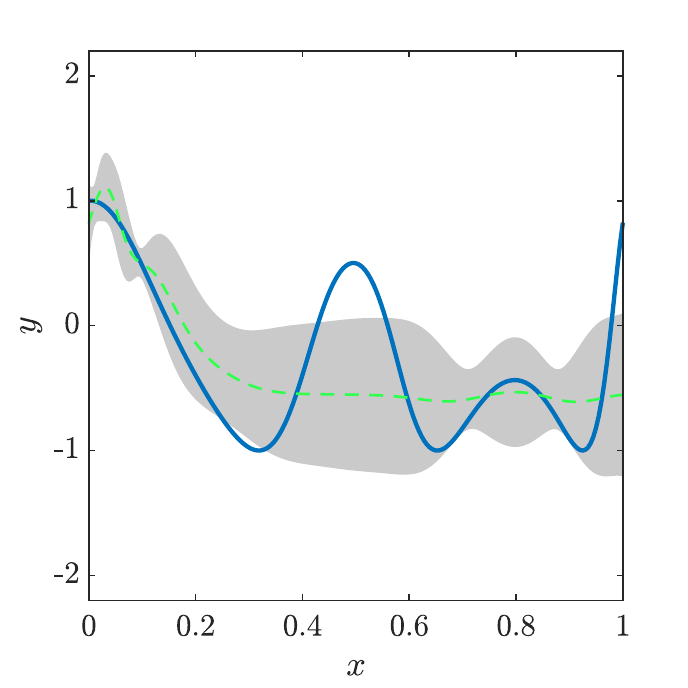}}
\caption{The linked GaSP constructed using the independent LHD. The solid lines in (a), (b) and (d) are true functional forms of $f_1$, $f_2$ and the coupled system; the surface in (c) is the true functional form of $f_3$; the solid trajectory on the surface in (c) corresponds to the region of $f_3$ that has impact on the global output given the interested range of the global input $x$. The dashed line and shaded area in (d) represent the mean and predictive interval of the constructed linked GaSP. The filled circles are training points generated by the LHD to construct the GaSP emulators of individual sub-models.}
\label{fig:exp_supp_ind}
\end{figure}

\begin{figure}[htbp]
\centering 
\subfloat[$f_1$]{\label{fig:exp_supp_dep_f1}\includegraphics[width=0.25\linewidth]{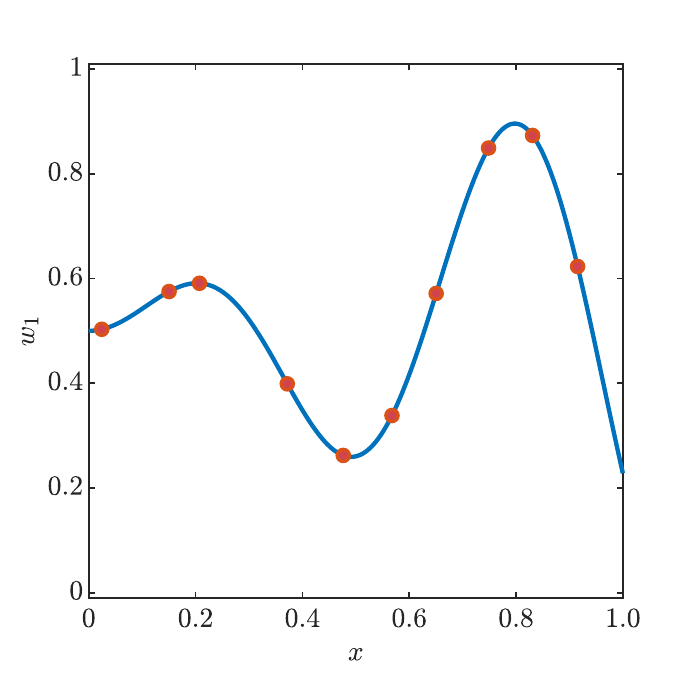}} 
\subfloat[$f_2$]{\label{fig:exp_supp_dep_f2}\includegraphics[width=0.25\linewidth]{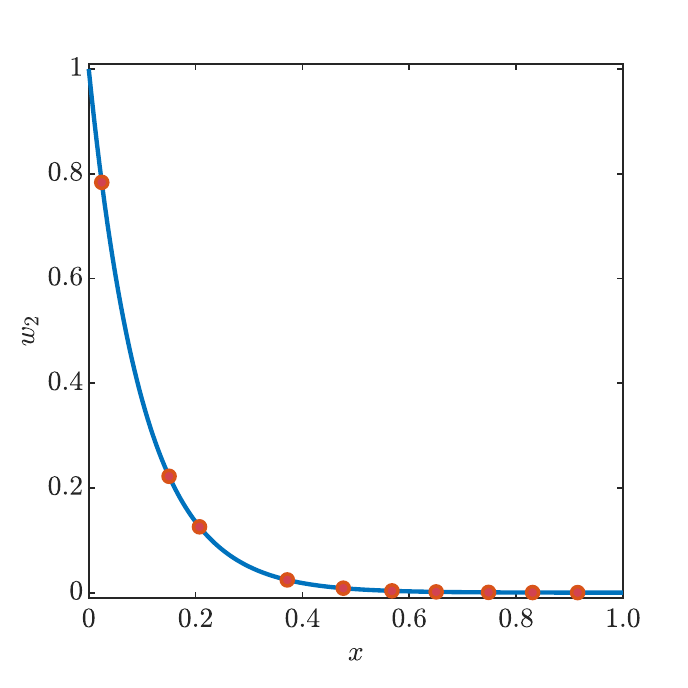}}
\subfloat[$f_3$]{\label{fig:exp_supp_dep_f3}\includegraphics[width=0.25\linewidth]{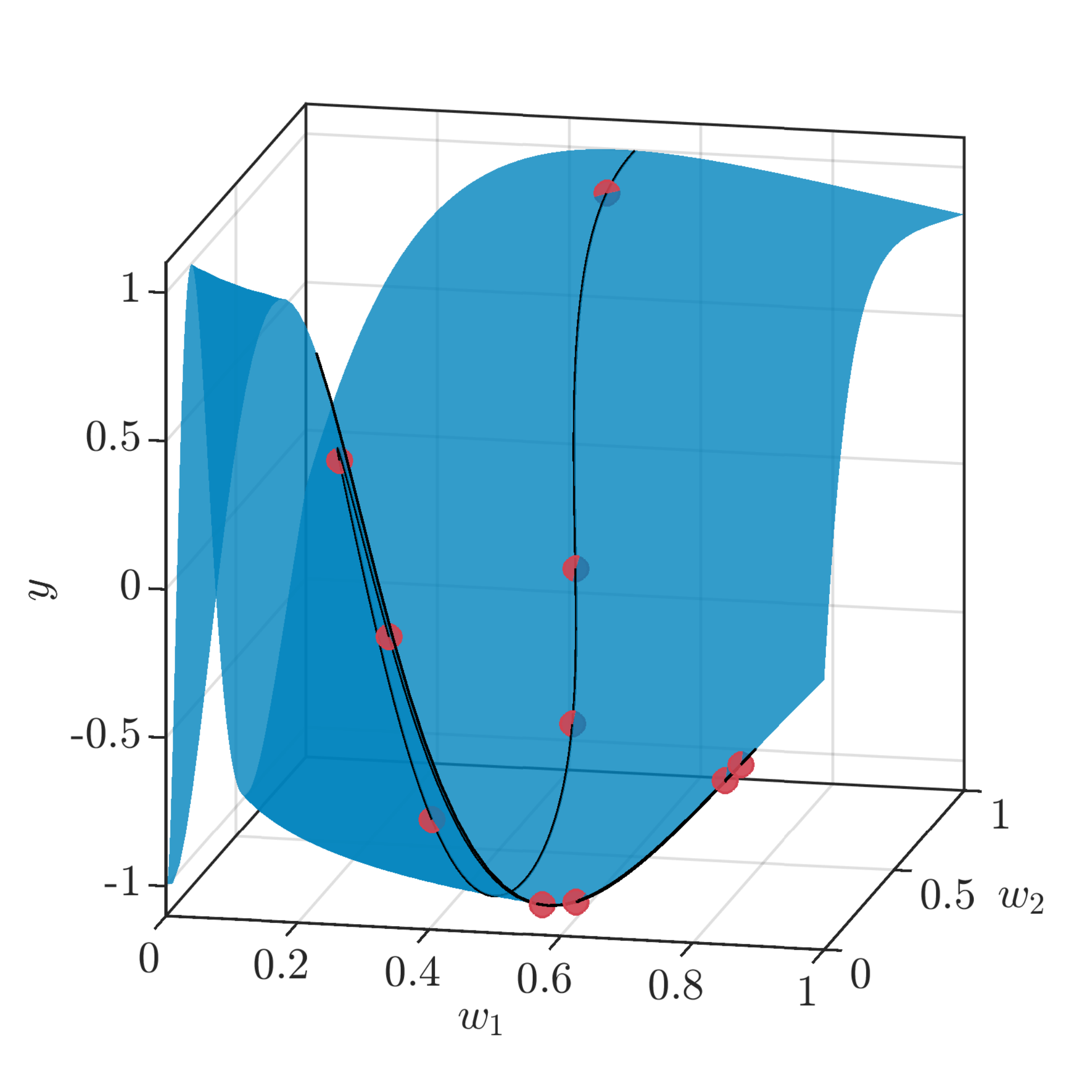}}
\subfloat[linked GaSP]{\label{fig:exp_supp_dep_emulator}\includegraphics[width=0.25\linewidth]{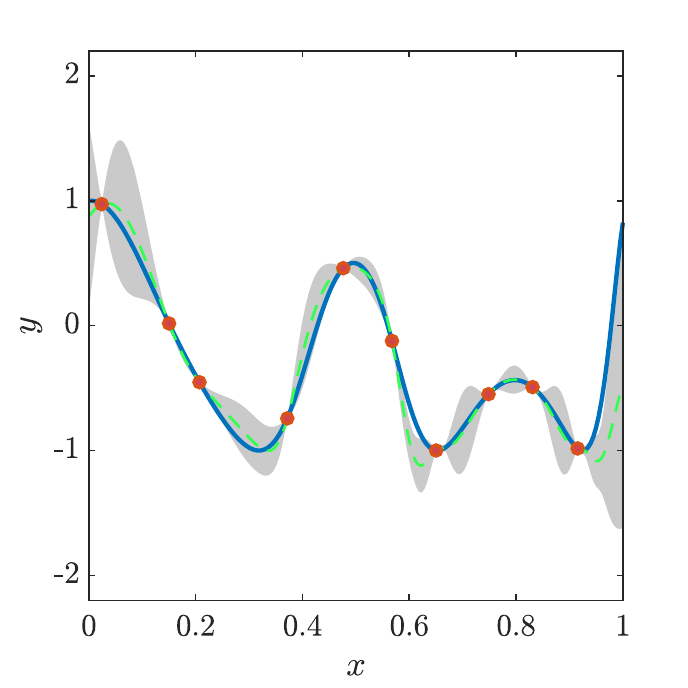}}
\caption{The linked GaSP constructed using the sequential LHD, where the design of $f_3$ is determined by propagating the LHD on the global input $x$ through $f_1$ and $f_2$. The solid lines in (a), (b) and (d) are true functional forms of $f_1$, $f_2$ and the coupled system; the surface in (c) is the true functional form of $f_3$; the solid trajectory on the surface in (c) corresponds to the region of $f_3$ that has impact on the global output given the interested range of the global input $x$. The dashed line and shaded area in (d) represent the mean and predictive interval of the constructed linked GaSP. The filled circles are training points.}
\label{fig:exp_supp_dep}
\end{figure}

\section[Diagnosis of Significance of Dependence among Outputs of Feeding Computer Models]{Diagnosis of Significance of Dependence among Outputs of Feeding Computer Models}
\label{sec:dependence}
In this section, we review the result given in~\cite{kyzyurova2018coupling} that can be used to diagnose whether the ignorance of dependence between the outputs of computer models in the feeding layers has significant impacts on the resultant linked GaSP. We reproduce the following theorem of~\cite{kyzyurova2018coupling} with proof and in consistency with our notations.

\begin{theorem}
\label{thm:fullcov}
Replace Assumption~\ref{ass:2} by the following assumption:
\begin{equation*}
\mathbf{W}\sim\mathcal{MN}(\boldsymbol{\mu},\boldsymbol{\Sigma}),
\end{equation*}
where $\boldsymbol{\Sigma}$ is the covariance matrix of $\mathbf{W}$ with diagonal elements being $\sigma^2_1(\mathbf{x}_1),\dots,\sigma^2_d(\mathbf{x}_d)$. Then, when $\widehat{g}$ is built with the squared exponential kernel, the mean and variance of the linked emulator are given by those from Theorem~\ref{thm:main} with $\boldsymbol{\Omega}=\boldsymbol{\Sigma}$ and
\begin{itemize}
\item the $i$-th element of $\mathbf{I}$:
\begin{equation*}
    I_i=\widetilde{\xi}_{i}\prod_{k=1}^pc_k(z_k,z^{\mathcal{T}}_{ik}),
\end{equation*}
where 
\begin{equation*}
\widetilde{\xi}_{i}=\frac{1}{\sqrt{|(\boldsymbol\Lambda+\boldsymbol\Sigma)\boldsymbol\Lambda^{-1}|}}\exp\left\{-\frac{1}{2}(\boldsymbol{\omega}^{\mathcal{T}}_i-\boldsymbol{\mu})^\top(\boldsymbol{\Lambda}+\boldsymbol{\Sigma})^{-1}(\boldsymbol{\omega}^{\mathcal{T}}_i-\boldsymbol{\mu})\right\}
\end{equation*}
with $\boldsymbol{\Lambda}=\mathrm{diag}(\frac{\gamma^2_1}{2},\dots,\frac{\gamma^2_d}{2})$;
\item the $ij$-th element of $\mathbf{J}$:
\begin{equation*}
J_{ij}=\widetilde{\zeta}_{ij}\prod_{k=1}^p c_k(z_k,\,z^{\mathcal{T}}_{ik})\,c_k(z_k,\,z^{\mathcal{T}}_{jk}),
\end{equation*} 
where 
\begin{multline*}
\widetilde{\zeta}_{ij}=\frac{1}{\sqrt{|(\boldsymbol\Gamma+\boldsymbol\Sigma)\boldsymbol\Gamma^{-1}|}}\exp\left\{-\frac{1}{8}(\boldsymbol{\omega}^{\mathcal{T}}_i-\boldsymbol{\omega}^{\mathcal{T}}_j)^\top\boldsymbol{\Gamma}^{-1}(\boldsymbol{\omega}^{\mathcal{T}}_i-\boldsymbol{\omega}^{\mathcal{T}}_j)\right\}\\
\times\exp\left\{-\frac{1}{2}\left(\frac{\boldsymbol{\omega}^{\mathcal{T}}_i+\boldsymbol{\omega}^{\mathcal{T}}_j}{2}-\boldsymbol{\mu}\right)^\top(\boldsymbol{\Gamma}+\boldsymbol{\Sigma})^{-1}\left(\frac{\boldsymbol{\omega}^{\mathcal{T}}_i+\boldsymbol{\omega}^{\mathcal{T}}_j}{2}-\boldsymbol{\mu}\right)\right\}  
\end{multline*}
with $\boldsymbol{\Gamma}=\mathrm{diag}(\frac{\gamma^2_1}{4},\dots,\frac{\gamma^2_d}{4})$;
\item the $lj$-th elemen of $\mathbf{B}$:
\begin{equation*}
B_{lj}=\widetilde{\psi}_{jl}\prod_{k=1}^pc_k(z_k,z^{\mathcal{T}}_{jk}),
\end{equation*}
where 
\begin{equation*}
\widetilde{\psi}_{jl}=\mathbf{e}_l[\boldsymbol{\Lambda}(\boldsymbol{\Lambda}+\boldsymbol{\Sigma})^{-1}\boldsymbol{\mu}+\boldsymbol{\Sigma}(\boldsymbol{\Lambda}+\boldsymbol{\Sigma})^{-1}\boldsymbol{\omega}^{\mathcal{T}}_j]\,\widetilde{\xi}_{j}.
\end{equation*}
\end{itemize}
\end{theorem}
\begin{proof}
The proof is in Section~\ref{sec:prooffullcov}.
\end{proof}

It can be seen from Theorem~\ref{thm:fullcov} that the covariance matrix $\boldsymbol\Sigma$ appears in the forms of inversions and determinants of  $\boldsymbol{\Lambda}+\boldsymbol{\Sigma}$ and $\boldsymbol{\Gamma}+\boldsymbol{\Sigma}$ in most cases and appears only in these two forms when the trend function is set to a constant (i.e., $\mathbf{B}$ has no effects on the mean and variance of the linked emulator). Thus, how significant the dependence (i.e., the off-diagonal elements of $\boldsymbol{\Sigma}$) between outputs $\mathbf{W}$ is to the  linked emulator depends on the magnitudes of $\gamma^2_1,\dots,\gamma^2_d$. When the magnitudes of $\gamma^2_1,\dots,\gamma^2_d$ are sufficiently large such that $\boldsymbol{\Lambda}+\boldsymbol{\Sigma}$ and $\boldsymbol{\Gamma}+\boldsymbol{\Sigma}$ become diagonally dominant, the inversions and determinants of $\boldsymbol{\Lambda}+\boldsymbol{\Sigma}$ and $\boldsymbol{\Gamma}+\boldsymbol{\Sigma}$ can be well approximated by those of $\boldsymbol{\Lambda}+\mathrm{diag}(\boldsymbol{\Sigma})$ and $\boldsymbol{\Gamma}+\mathrm{diag}(\boldsymbol{\Sigma})$~\citep{demmel1992componentwise,ipsen2011determinant}. As a result, in practice one could first construct GaSP emulators of individual computer models without considering the possible dependence between their outputs, and then check the ratios of $\gamma^2_k$ to $\sigma^2_k$ for all $k=1,\dots,d$ to determine whether the dependence structure is non-negligible. Note that given $\gamma^2_k$, the ratio of $\gamma^2_k$ to $\sigma^2_k$ increases as $\sigma^2_k$ drops. Therefore, at least in the squared exponential case given in Theorem~\ref{thm:fullcov}, one can safely neglect the dependence as long as emulators in the feeding layer are produce small variances at the global input positions to be evaluated. This point is intuitive because when the predictive variances go to zero at a given input position, GaSP emulators converge to the corresponding predictive means and become constants. Therefore, incorporating the dependence structure is unnecessary. Figure~\ref{fig:ratio} presents ratios of the synthetic system in Figure~\ref{fig:exp_supp} at various testing global input positions. It can be seen that for most of the global input positions, ratios of $\gamma^2_k$ to $\sigma^2_k$ for $k=1,\,2$ are higher than $100$, meaning that linked GaSPs can be constructed without the consideration of the dependence between $w_1$ and $w_2$. Even though ratios of $\gamma^2_k$ to $\sigma^2_k$ are relative low over $x\in[0.9,1.0]$, these ratios can be raised by improving the GaSP emulators of $f_1$ and $f_2$ over that region.

\begin{figure}[htbp]
\centering 
\includegraphics[width=0.5\linewidth]{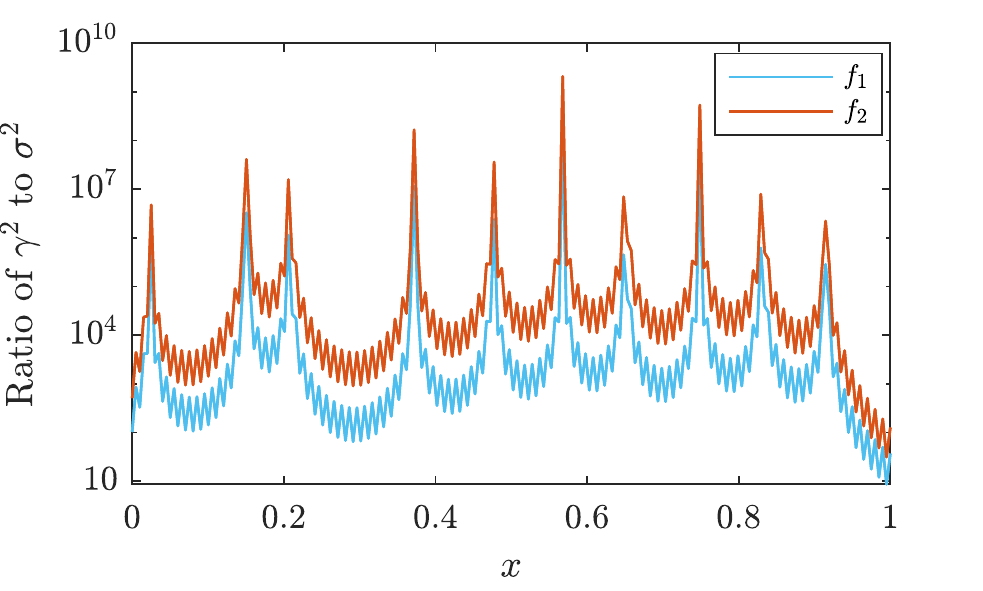}
\caption{Ratios of $\gamma^2_k$ to $\sigma^2_k$ with $k=1,\,2$ for the synthetic system in Figure~\ref{fig:exp_supp}. The upper and lower solid lines give respectively ratios that associate to the outputs of GaSP emulators of $f_2$ and $f_1$, over the global input domain. Ten large spikes correspond to input positions near the training data points. Ratios of $\gamma^2$ to $\sigma^2$ are plotted in log-scale.}
\label{fig:ratio}
\end{figure}

\section[The Approximation Performance of Linked GaSP to Linked Emulator]{The Approximating Performance of Linked GaSP to Linked Emulator}
\label{sec:approx}
In this section, we explore the approximation accuracy of linked GaSP to linked emulator in a three-layered synthetic system shown in Figure~\ref{fig:exp1}. The individual computer models $f_1$, $f_2$ and $f_3$ with scalar-valued output $w_1$, $w_2$ and $y$ respectively are defined by the following analytical forms:
\begin{equation*}
f_1=\sin(\pi x),\quad f_2=\cos(5w_1)\quad\mathrm{and}\quad f_3=\sin(w_2^2),
\end{equation*}
where the global input $x\in[-1,1]$.

\begin{figure}[htbp]
\centering
\begin{tikzpicture}[shorten >=1pt,->,draw=black!50, node distance=4.5cm]
    \tikzstyle{every pin edge}=[<-,shorten <=1pt]
    \tikzstyle{neuron}=[circle,fill=black!25,minimum size=17pt,inner sep=0pt]
    \tikzstyle{layer1}=[neuron, fill=green!50];
    \tikzstyle{layer2}=[neuron, fill=red!50];
    \tikzstyle{layer3}=[neuron, fill=blue!50];
    \tikzstyle{layer4}=[neuron, fill=purple!50];
    \tikzstyle{annot} = [text width=4em, text centered]

    \node[layer1, pin=left:$x$] (I-1) at (0,0) {${f}_1$};
    \node[layer2] (O-1) at (4.5,0) {${f}_2$};
    \node[layer3,pin={[pin edge={->}]right:$y$}] (K-1) at (9cm, 0cm) {${f}_3$};

    \path[draw] (I-1) -- (O-1) node[font=\small,midway,fill=white,align=left,sloped] {$w_1$};
    \path[draw] (O-1) -- (K-1) node[font=\small,midway,fill=white,align=left,sloped] {$w_2$};
    \node[annot,above of=O-1, node distance=0.7cm] (hl) {Layer 2};
    \node[annot,above of=K-1, node distance=0.7cm] (hk) {Layer 3};
    \node[annot,left of=hl] {Layer 1};
\end{tikzpicture}
\caption{A synthetic three-layered computer system with three computer models $f_1$, $f_2$ and $f_3$, all of which have 1-D input and output.}
\label{fig:exp1}
\end{figure}
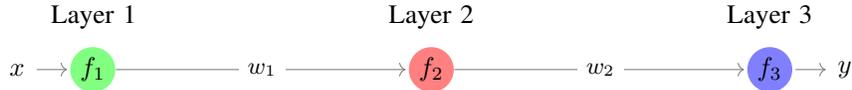 

We draw eight training points from the sequential LHD to construct the linked GaSP and the linked emulator. The linked emulator is represented by $500$ random samples drawn sequentially through GaSP emulators of $f_1$, $f_2$ and $f_3$. Figure~\subref*{fig:exp_supp_density} compares the full probabilistic descriptions between the linked GaSP and linked emulator. It is clear that the linked emulator is not Gaussian distributed because it is skewed and most of its densities are concentrated near zero. As a Gaussian approximation to the linked emulator, the linked GaSP puts some probability masses below zero, giving overestimated and unrealistic uncertainty descriptions of the underlying system at unrealized input positions. This discrepancy on the probability density can cause inaccurate uncertainty assessment based on the linked GaSP if the probability distribution of an emulator is critical, e.g., the tail is of specific interest. However, Figure~\subref*{fig:exp_supp_mean} and~\subref*{fig:exp_supp_std} indicate that the linked GaSP approximates well the mean and variance of the linked emulator. Therefore, if mean and variance are essential quantities of an uncertainty analysis, linked GaSP is an adequate replacement of the linked emulator and one can benefit analytical expressions of the linked GaSP for efficient and effective analysis of the underlying computer system. 

\begin{figure}[htbp]
\centering 
\subfloat[Density]{\label{fig:exp_supp_density}\includegraphics[width=0.33\linewidth]{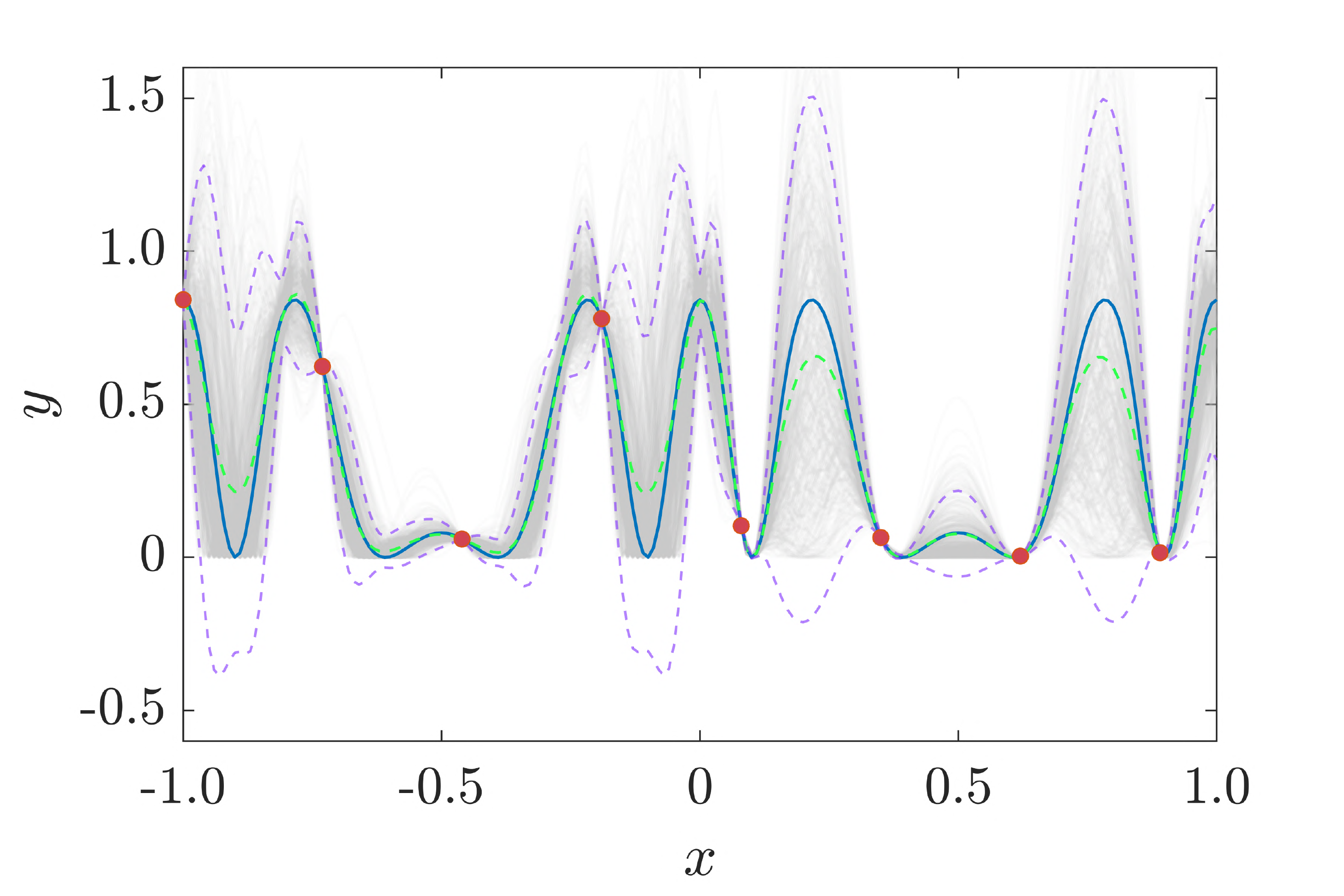}} 
\subfloat[Mean]{\label{fig:exp_supp_mean}\includegraphics[width=0.33\linewidth]{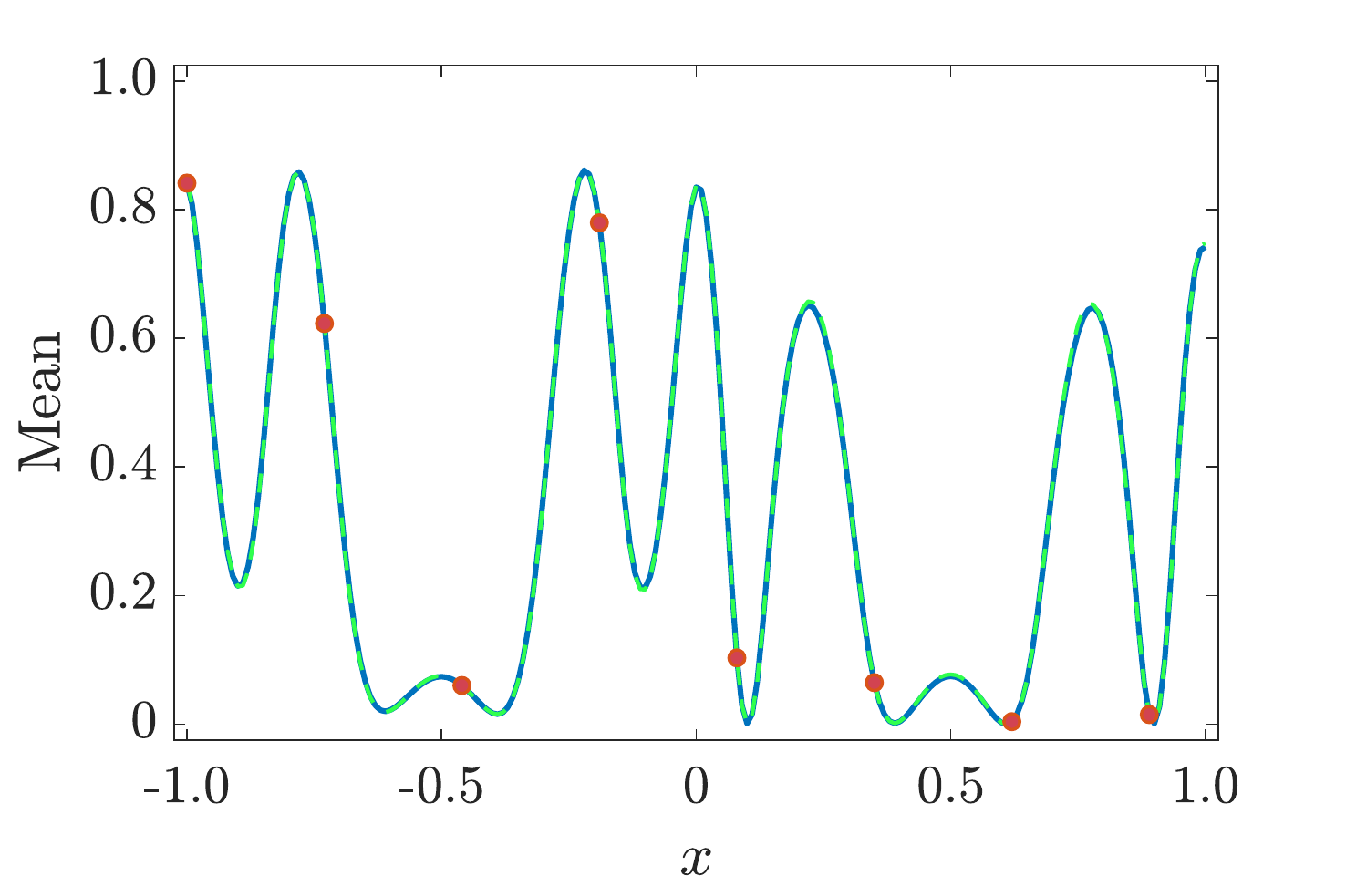}}
\subfloat[Standard Deviation]{\label{fig:exp_supp_std}\includegraphics[width=0.33\linewidth]{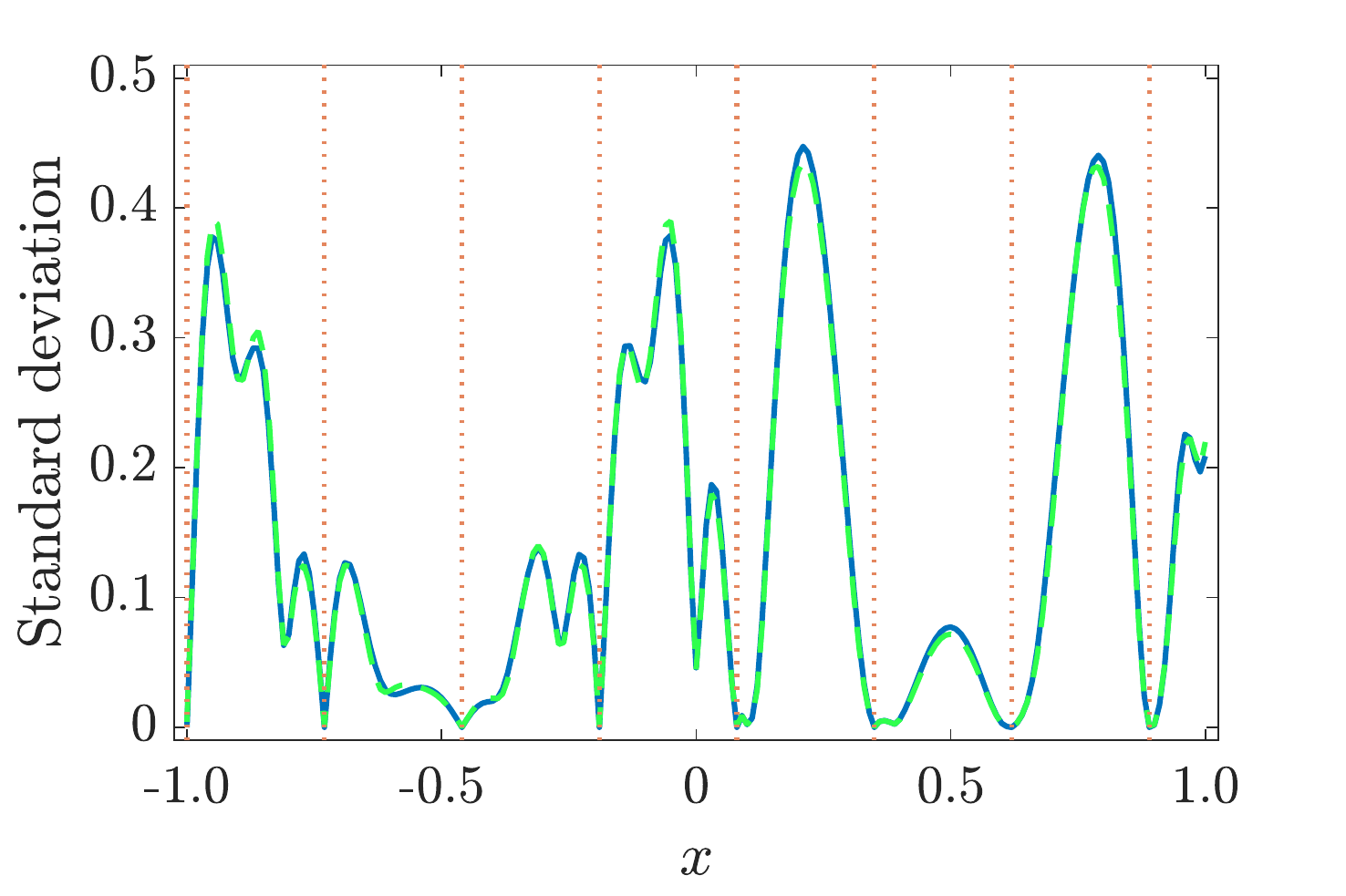}}
\caption{\emph{(a)} Comparison of probability densities of the linked GaSP and the linked emulator. The grey-shaded lines are $500$ random sample paths representing the linked emulator; the blue solid line is the true functional form between the global input and output of the system in Figure~\ref{fig:exp1}; the dashed green line is the mean prediction of the linked GaSP; the dashed purple lines represent $5$-th and $95$-th percentiles of the linked GaSP; the filled circles are training points used to construct the linked GaSP and linked emulator; \emph{(b)} Comparison of predictive mean between the linked GaSP and the linked emulator. The blue solid line is the predictive mean of the linked emulator (that is calculated using the empirical mean of the $500$ sample paths) and the dashed line is the mean of the linked GaSP; the filled circles are training points used to construct the linked GaSP and linked emulator; \emph{(c)} Comparison of standard deviation between the linked GaSP and the linked emulator. The blue solid line is the standard deviation of the linked emulator (that is calculated using the empirical standard deviation of the $500$ sample paths) and the dashed line is the standard deviation of the linked GaSP; the dashed vertical lines indicate the input locations of training points.}
\label{fig:exp_supp_approx}
\end{figure}

\section[Proof of Theorem~3.1]{Proof of Theorem~\ref{thm:main}}
\label{sec:thmproof}
In this section, we prove Theorem~\ref{thm:main} by considering not only the multiplicative form of the kernel function but also the additive form given by
\begin{equation*}
c(\mathbf{X}_i,\,\mathbf{X}_j)=\sum_{k=1}^p c_k(X_{ik},\,X_{jk}).
\end{equation*}

\subsection{\texorpdfstring{{Derivation of $\mu_I$}}{Derivation of the mean}}
We first derive the expression for $\mu_I$. Let $\mu_g(\mathbf{W},\mathbf{z})$ and $\sigma^2_g(\mathbf{W},\mathbf{z})$ be the mean and variance of the GP emulator $\widehat{g}$. Then, by the tower rule, we have
\begin{equation*}
 \mu_I=\mathbb{E}[\mu_g(\mathbf{W},\mathbf{z})],
\end{equation*}
where the expectation is taken respect to $\mathbf{W}$. Replace $\mu_g(\mathbf{W},\mathbf{z})$ by equation~\eqref{eq:krigingmean} with Assumption~\ref{ass:1}, we have
\begin{align}
\label{eq:muL}
 \mu_I=&\mathbb{E}\left[\mathbf{W}^\top\widehat{\boldsymbol{\theta}}+\mathbf{h}(\mathbf{z})^\top\widehat{\boldsymbol{\beta}}+\mathbf{r}^\top(\mathbf{W},\,\mathbf{z})\mathbf{R}^{-1}\left(\mathbf{y}^{\mathcal{T}}-\mathbf{w}^{\mathcal{T}}\widehat{\boldsymbol{\theta}}-\mathbf{H}(\mathbf{z}^{\mathcal{T}})\widehat{\boldsymbol{\beta}}\right)\right]\nonumber\\
 =&\mathbb{E}\left[\mathbf{W}^\top\right]\widehat{\boldsymbol{\theta}}+\mathbf{h}(\mathbf{z})^\top\widehat{\boldsymbol{\beta}}+\mathbb{E}\left[\mathbf{r}^\top(\mathbf{W},\,\mathbf{z})\right]\mathbf{R}^{-1}\left(\mathbf{y}^{\mathcal{T}}-\mathbf{w}^{\mathcal{T}}\widehat{\boldsymbol{\theta}}-\mathbf{H}(\mathbf{z}^{\mathcal{T}})\widehat{\boldsymbol{\beta}}\right)\nonumber\\
 =&\boldsymbol{\mu}^\top\widehat{\boldsymbol{\theta}}+\mathbf{h}(\mathbf{z})^\top\widehat{\boldsymbol{\beta}}+\mathbf{I}^\top\mathbf{A},
\end{align}
where
\begin{itemize}
    \item $\boldsymbol{\mu}=[\mu_1(\mathbf{x}_1),\dots,\mu_d(\mathbf{x}_d)]^\top\in\mathbb{R}^{d\times 1}\,$;
    \item $\mathbf{A}=\mathbf{R}^{-1}\left(\mathbf{y}^{\mathcal{T}}-\mathbf{w}^{\mathcal{T}}\widehat{\boldsymbol{\theta}}-\mathbf{H}(\mathbf{z}^{\mathcal{T}})\widehat{\boldsymbol{\beta}}\right)\in\mathbb{R}^{m\times 1}\,$;
    \item $\left[\widehat{\boldsymbol{\theta}}^\top,\,\widehat{\boldsymbol{\beta}}^\top\right]^\top\eqdef\left(\widetilde{\mathbf{H}}^\top\mathbf{R}^{-1}\widetilde{\mathbf{H}}\right)^{-1}\widetilde{\mathbf{H}}^\top\mathbf{R}^{-1}\mathbf{y}^{\mathcal{T}}$ with $\widetilde{\mathbf{H}}=\left[\mathbf{w}^{\mathcal{T}},\mathbf{H}(\mathbf{z}^{\mathcal{T}})\right]\in\mathbb{R}^{m\times(d+q)}$;
    \item $\mathbf{I}=\mathbb{E}\left[\mathbf{r}(\mathbf{W},\,\mathbf{z})\right]\in\mathbb{R}^{m\times 1}$ with its $i$-th element:
    \begin{align*}
        I_i=&\mathbb{E}\left[c(\mathbf{W},\,\mathbf{w}^{\mathcal{T}}_i)c(\mathbf{z},\,\mathbf{z}^{\mathcal{T}}_i)\right]\\
        =&\mathbb{E}\left[c(\mathbf{W},\,\mathbf{w}^{\mathcal{T}}_i)\right]c(\mathbf{z},\,\mathbf{z}^{\mathcal{T}}_i)\\
        =&\prod_{k=1}^d\mathbb{E}\left[c_k(W_k,\,w^{\mathcal{T}}_{ik})\right]\prod_{k=1}^pc_k(z_k,\,z^{\mathcal{T}}_{ik})\\
        =&\prod_{k=1}^d\xi_{ik} \prod_{k=1}^pc_k(z_k,\,z^{\mathcal{T}}_{ik})
    \end{align*}
    in case of multiplicative form, and
    \begin{align*}
        I_i=&\mathbb{E}\left[c(\mathbf{W},\,\mathbf{w}^{\mathcal{T}}_i)+c(\mathbf{z},\,\mathbf{z}^{\mathcal{T}}_i)\right]\\
        =&\mathbb{E}\left[c(\mathbf{W},\,\mathbf{w}^{\mathcal{T}}_i)\right]+c(\mathbf{z},\,\mathbf{z}^{\mathcal{T}}_i)\\
        =&\sum_{k=1}^d\mathbb{E}\left[c_k(W_k,\,w^{\mathcal{T}}_{ik})\right]+\sum_{k=1}^pc_k(z_k,\,z^{\mathcal{T}}_{ik})\\
        =&\sum_{k=1}^d\xi_{ik}+\sum_{k=1}^pc_k(z_k,\,z^{\mathcal{T}}_{ik})
    \end{align*}
    in case of additive form, where 
    \begin{equation*}
      \xi_{ik}\eqdef\mathbb{E}\left[c_k(W_k,\,w^{\mathcal{T}}_{ik})\right]
    \end{equation*}
    and in the derivation above we use the independence of $W_{i=1,\dots,d}$.
\end{itemize}

\subsection{\texorpdfstring{{Derivation of $\sigma^2_I$}}{Derivation of the variance}}
We now derive the expression for the variance $\sigma^2_I\,$. Using the law of total variance, we have
\begin{align}
\label{eq:sigma2L}
 \sigma^2_I=&\mathbb{E}\left[\sigma^2_g(\mathbf{W},\mathbf{z})\right]+\mathrm{Var}\left(\mu_g(\mathbf{W},\mathbf{z})\right)\nonumber\\
      =&\mathbb{E}\left[\sigma^2_g(\mathbf{W},\mathbf{z})\right]+\mathbb{E}\left[\mu_g^2(\mathbf{W},\mathbf{z})\right]-\mathbb{E}\left[\mu_g(\mathbf{W},\mathbf{z})\right]^2\nonumber\\
      =&\mathbb{E}\left[\sigma^2_g(\mathbf{W},\mathbf{z})\right]+\mathbb{E}\left[\mu_g^2(\mathbf{W},\mathbf{z})\right]-\mu_I^2.
\end{align}

\subsubsection{\texorpdfstring{{Derivation of $\mathbb{E}\left[\mu_g^2(\mathbf{W},\mathbf{z})\right]$}}{Derivation of the variance - 1}} 

Replace $\mu_g(\mathbf{W},\mathbf{z})$ by equation~\eqref{eq:krigingmean}, we have
    \begin{align*}
   \mu_g(\mathbf{W},\mathbf{z})=&\left[\mathbf{W}^\top\widehat{\boldsymbol{\theta}}+\mathbf{h}(\mathbf{z})^\top\widehat{\boldsymbol{\beta}}+\mathbf{r}^\top(\mathbf{W},\,\mathbf{z})\mathbf{R}^{-1}\left(\mathbf{y}^{\mathcal{T}}-\mathbf{w}^{\mathcal{T}}\widehat{\boldsymbol{\theta}}-\mathbf{H}(\mathbf{z}^{\mathcal{T}})\widehat{\boldsymbol{\beta}}\right)\right]^2\\
    =&\mathbf{W}^\top\widehat{\boldsymbol{\theta}}\widehat{\boldsymbol{\theta}}^\top\mathbf{W}+\left(\mathbf{h}(\mathbf{z})^\top\widehat{\boldsymbol{\beta}}\right)^2+2\widehat{\boldsymbol{\theta}}^\top\mathbf{W}\mathbf{h}(\mathbf{z})^\top\widehat{\boldsymbol{\beta}}\\
    &+2\widehat{\boldsymbol{\theta}}^\top\mathbf{W}\mathbf{r}^\top(\mathbf{W},\,\mathbf{z})\mathbf{A}+2\mathbf{h}(\mathbf{z})^\top\widehat{\boldsymbol{\beta}}\mathbf{r}^\top(\mathbf{W},\,\mathbf{z})\mathbf{A}+\mathbf{r}^\top(\mathbf{W},\,\mathbf{z})\mathbf{A}\mathbf{A}^\top\mathbf{r}(\mathbf{W},\,\mathbf{z}).
    \end{align*}
    Then, we have
    \begin{align*}
    \mathbb{E}\left[\mu_g(\mathbf{W},\mathbf{z})^2\right]=&\mathbb{E}\left[\mathbf{W}^\top\widehat{\boldsymbol{\theta}}\widehat{\boldsymbol{\theta}}^\top\mathbf{W}\right]+\left(\mathbf{h}(\mathbf{z})^\top\widehat{\boldsymbol{\beta}}\right)^2+2\widehat{\boldsymbol{\theta}}^\top\mathbb{E}\left[\mathbf{W}\right]\mathbf{h}(\mathbf{z})^\top\widehat{\boldsymbol{\beta}}\nonumber\\
    &+2\widehat{\boldsymbol{\theta}}^\top\mathbb{E}\left[\mathbf{W}\mathbf{r}^\top(\mathbf{W},\,\mathbf{z})\right]\mathbf{A}+2\mathbf{h}(\mathbf{z})^\top\widehat{\boldsymbol{\beta}}\mathbb{E}\left[\mathbf{r}^\top(\mathbf{W},\,\mathbf{z})\right]\mathbf{A}\nonumber\\
    &+\mathbb{E}\left[\mathbf{r}^\top(\mathbf{W},\,\mathbf{z})\mathbf{A}\mathbf{A}^\top\mathbf{r}(\mathbf{W},\,\mathbf{z})\right]\nonumber\\
    =&\mathbb{E}\left[\mathbf{W}^\top\widehat{\boldsymbol{\theta}}\widehat{\boldsymbol{\theta}}^\top\mathbf{W}\right]+\left(\mathbf{h}(\mathbf{z})^\top\widehat{\boldsymbol{\beta}}\right)^2+2\widehat{\boldsymbol{\theta}}^\top\boldsymbol{\mu}\mathbf{h}(\mathbf{z})^\top\widehat{\boldsymbol{\beta}}\nonumber\\
    &+2\widehat{\boldsymbol{\theta}}^\top\mathbf{B}\mathbf{A}+2\mathbf{h}(\mathbf{z})^\top\widehat{\boldsymbol{\beta}}\mathbf{I}^\top\mathbf{A}+\mathbb{E}\left[\mathbf{r}^\top(\mathbf{W},\,\mathbf{z})\mathbf{A}\mathbf{A}^\top\mathbf{r}(\mathbf{W},\,\mathbf{z})\right]\nonumber
    \end{align*}
    The first expectation in the above equation can be solved as follow:
    \begin{align}
    \label{eq:expectation1}
\mathbb{E}\left[\mathbf{W}^\top\widehat{\boldsymbol{\theta}}\widehat{\boldsymbol{\theta}}^\top\mathbf{W}\right]=&\mathrm{tr}\left\{\widehat{\boldsymbol{\theta}}\widehat{\boldsymbol{\theta}}^\top\mathrm{var}(\mathbf{W})\right\}+\mathbb{E}_{\mathbf{W}}\left[\mathbf{W}\right]^\top\widehat{\boldsymbol{\theta}}\widehat{\boldsymbol{\theta}}^\top\mathbb{E}_{\mathbf{W}}\left[\mathbf{W}\right]\nonumber\\
=&\mathrm{tr}\left\{\widehat{\boldsymbol{\theta}}\widehat{\boldsymbol{\theta}}^\top\boldsymbol{\Omega}\right\}+\boldsymbol{\mu}^\top\widehat{\boldsymbol{\theta}}\widehat{\boldsymbol{\theta}}^\top\boldsymbol{\mu}\nonumber\\
=&\mathrm{tr}\left\{\widehat{\boldsymbol{\theta}}\widehat{\boldsymbol{\theta}}^\top\boldsymbol{\Omega}\right\}+\mathrm{tr}\left\{\widehat{\boldsymbol{\theta}}\widehat{\boldsymbol{\theta}}^\top\boldsymbol{\mu}\boldsymbol{\mu}^\top\right\}\nonumber\\
=&\mathrm{tr}\left\{\widehat{\boldsymbol{\theta}}\widehat{\boldsymbol{\theta}}^\top\left(\boldsymbol{\mu}\boldsymbol{\mu}^\top+\boldsymbol{\Omega}\right)\right\}.
\end{align}
The second expectation can be solved in a similar manner:
    \begin{align}
    \label{eq:expectation2}
 \mathbb{E}\left[\mathbf{r}^\top(\mathbf{W},\,\mathbf{z})\mathbf{A}\mathbf{A}^\top\mathbf{r}(\mathbf{W},\,\mathbf{z})\right]=&\mathrm{tr}\left\{\mathbb{E}\left[\mathbf{r}^\top(\mathbf{W},\,\mathbf{z})\mathbf{A}\mathbf{A}^\top\mathbf{r}(\mathbf{W},\,\mathbf{z})\right]\right\}\nonumber\\
 =&\mathbb{E}\left[\mathrm{tr}\left\{\mathbf{r}^\top(\mathbf{W},\,\mathbf{z})\mathbf{A}\mathbf{A}^\top\mathbf{r}(\mathbf{W},\,\mathbf{z})\right\}\right]\nonumber\\
 =&\mathrm{tr}\left\{\mathbf{A}\mathbf{A}^\top\mathbb{E}\left[\mathbf{r}(\mathbf{W},\,\mathbf{z})\mathbf{r}^\top(\mathbf{W},\,\mathbf{z})\right]\right\}\nonumber\\
 =&\mathrm{tr}\left\{\mathbf{A}\mathbf{A}^\top\mathbf{J}\right\}.
\end{align}
Thus, we obtain that

    \begin{align*}
    \mathbb{E}\left[\mu_g(\mathbf{W},\mathbf{z})^2\right]=&\mathrm{tr}\left\{\widehat{\boldsymbol{\theta}}\widehat{\boldsymbol{\theta}}^\top\mathrm{var}(\mathbf{W})\right\}+\mathbb{E}\left[\mathbf{W}\right]^\top\widehat{\boldsymbol{\theta}}\widehat{\boldsymbol{\theta}}^\top\mathbb{E}\left[\mathbf{W}\right]+\left(\mathbf{h}(\mathbf{z})^\top\widehat{\boldsymbol{\beta}}\right)^2+2\widehat{\boldsymbol{\theta}}^\top\boldsymbol{\mu}\mathbf{h}(\mathbf{z})^\top\widehat{\boldsymbol{\beta}}\nonumber\\
    &+2\widehat{\boldsymbol{\theta}}^\top\mathbf{B}\mathbf{A}+2\mathbf{h}(\mathbf{z})^\top\widehat{\boldsymbol{\beta}}\mathbf{I}^\top\mathbf{A}+\mathrm{tr}\left\{\mathbf{A}\mathbf{A}^\top\mathbb{E}\left[\mathbf{r}(\mathbf{W},\,\mathbf{z})\mathbf{r}^\top(\mathbf{W},\,\mathbf{z})\right]\right\}\nonumber\\	
    =&\mathrm{tr}\left\{\widehat{\boldsymbol{\theta}}\widehat{\boldsymbol{\theta}}^\top\left(\boldsymbol{\mu}\boldsymbol{\mu}^\top+\boldsymbol{\Omega}\right)\right\}+\left(\mathbf{h}(\mathbf{z})^\top\widehat{\boldsymbol{\beta}}\right)^2+2\widehat{\boldsymbol{\theta}}^\top\boldsymbol{\mu}\mathbf{h}(\mathbf{z})^\top\widehat{\boldsymbol{\beta}}\nonumber\\
    &+2\left[\widehat{\boldsymbol{\theta}}^\top\mathbf{B}+\mathbf{h}(\mathbf{z})^\top\widehat{\boldsymbol{\beta}}\mathbf{I}^\top\right]\mathbf{A}+\mathrm{tr}\left\{\mathbf{A}\mathbf{A}^\top\mathbf{J}\right\},
    \end{align*}
where 
\begin{itemize}
\item $\boldsymbol{\Omega}=\mathrm{diag}(\sigma^2_1(\mathbf{x}_1),\dots,\sigma^2_d(\mathbf{x}_d))\in\mathbb{R}^{d\times d}\,$;
\item $\mathbf{B}=\mathbb{E}\left[\mathbf{W}\mathbf{r}^\top(\mathbf{W},\,\mathbf{z})\right]\in\mathbb{R}^{d\times m}$ with its $lj$-th element:	
\begin{align*}
    B_{lj}=&\mathbb{E}\left[W_l\,c(\mathbf{W},\,\mathbf{w}^{\mathcal{T}}_j)c(\mathbf{z},\,\mathbf{z}^{\mathcal{T}}_j)\right]\\
    =&\mathbb{E}\left[W_l c(\mathbf{W},\,\mathbf{w}^{\mathcal{T}}_j)\right]c(\mathbf{z},\,\mathbf{z}^{\mathcal{T}}_j)\\
    =&\mathbb{E}\left[W_l\prod_{k=1}^{d}c_k(W_k,\,w^{\mathcal{T}}_{jk})\right]\prod_{k=1}^pc_k(z_k,\,z^{\mathcal{T}}_{jk})\\
    =&\mathbb{E}\left[W_lc_l(W_l,\,w^{\mathcal{T}}_{jl})\right]\prod^d_{\substack{k=1\\k\neq l}}\mathbb{E}\left[c_k(W_k,\,w^{\mathcal{T}}_{jk})\right]\prod_{k=1}^pc_k(z_k,\,z^{\mathcal{T}}_{jk})\\
    =&\psi_{jl}\prod^d_{\substack{k=1\\k\neq l}}\xi_{jk}\prod_{k=1}^pc_k(z_k,\,z^{\mathcal{T}}_{jk})
    \end{align*}
in case of multiplicative form, and
\begin{align*}
    B_{lj}=&\mathbb{E}\left[W_l\left(c(\mathbf{W},\,\mathbf{w}^{\mathcal{T}}_j)+c(\mathbf{z},\,\mathbf{z}^{\mathcal{T}}_j)\right)\right]\\
    =&\mathbb{E}\left[W_l c(\mathbf{W},\,\mathbf{w}^{\mathcal{T}}_j)\right]+\mathbb{E}\left[W_l\right]c(\mathbf{z},\,\mathbf{z}^{\mathcal{T}}_j)\\
    =&\mathbb{E}\left[W_l\sum_{k=1}^{d}c_k(W_k,\,w^{\mathcal{T}}_{jk})\right]+\mu_l\sum_{k=1}^pc_k(z_k,\,z^{\mathcal{T}}_{jk})\\
    =&\mathbb{E}\left[W_lc_l(W_l,\,w^{\mathcal{T}}_{jl})\right]+\mu_l\sum^d_{\substack{k=1\\k\neq l}}\mathbb{E}\left[c_k(W_k,\,w^{\mathcal{T}}_{jk})\right]+\mu_l\sum_{k=1}^pc_k(z_k,\,z^{\mathcal{T}}_{jk})\\
    =&\psi_{jl}+\mu_l\sum^d_{\substack{k=1\\k\neq l}}\xi_{jk}+\mu_l\sum_{k=1}^pc_k(z_k,\,z^{\mathcal{T}}_{jk})
    \end{align*}
in case of additive form, in which
\begin{equation*}
  \psi_{jl}\eqdef\mathbb{E}\left[W_lc_l(W_l,\,w^{\mathcal{T}}_{jl})\right]\,;  
\end{equation*}
\item $\mathbf{J}=\mathbb{E}\left[\mathbf{r}(\mathbf{W},\,\mathbf{z})\mathbf{r}^\top(\mathbf{W},\,\mathbf{z})\right]\in\mathbb{R}^{m\times m}$ with its $ij$-th element:
\begin{align*}
J_{ij}=&\mathbb{E}\Big[c(\mathbf{W},\,\mathbf{w}^{\mathcal{T}}_i)c(\mathbf{z},\,\mathbf{z}^{\mathcal{T}}_i)\,c(\mathbf{W},\,\mathbf{w}^{\mathcal{T}}_j)c(\mathbf{z},\,\mathbf{z}^{\mathcal{T}}_j)\Big]\\
=&\mathbb{E}\left[c(\mathbf{W},\,\mathbf{w}^{\mathcal{T}}_i)c(\mathbf{W},\,\mathbf{w}^{\mathcal{T}}_j)\right]c(\mathbf{z},\,\mathbf{z}^{\mathcal{T}}_i)c(\mathbf{z},\,\mathbf{z}^{\mathcal{T}}_j)\\
=&\prod_{k=1}^d\mathbb{E}\left[c_k(W_k,\,w^{\mathcal{T}}_{ik})c_k(W_k,\,w^{\mathcal{T}}_{jk})\right]\prod_{k=1}^pc_k(z_k,\,z^{\mathcal{T}}_{ik})c_k(z_k,\,z^{\mathcal{T}}_{jk})\\
=&\prod_{k=1}^d \zeta_{ijk}\prod_{k=1}^pc_k(z_k,\,z^{\mathcal{T}}_{ik})c_k(z_k,\,z^{\mathcal{T}}_{jk})
\end{align*}
in case of multiplicative form, and
\begin{align*}
J_{ij}=&\mathbb{E}\Big[\left(c(\mathbf{W},\,\mathbf{w}^{\mathcal{T}}_i)+c(\mathbf{z},\,\mathbf{z}^{\mathcal{T}}_i)\right)\left(c(\mathbf{W},\,\mathbf{w}^{\mathcal{T}}_j)+c(\mathbf{z},\,\mathbf{z}^{\mathcal{T}}_j)\right)\Big]\\
=&\mathbb{E}\left[c(\mathbf{W},\,\mathbf{w}^{\mathcal{T}}_i)c(\mathbf{W},\,\mathbf{w}^{\mathcal{T}}_j)\right]+\mathbb{E}\left[c(\mathbf{W},\,\mathbf{w}^{\mathcal{T}}_i)\right]c(\mathbf{z},\,\mathbf{z}^{\mathcal{T}}_j)\\
&+\mathbb{E}\left[c(\mathbf{W},\,\mathbf{w}^{\mathcal{T}}_j)\right]c(\mathbf{z},\,\mathbf{z}^{\mathcal{T}}_i)+c(\mathbf{z},\,\mathbf{z}^{\mathcal{T}}_i)c(\mathbf{z},\,\mathbf{z}^{\mathcal{T}}_j)\\
=&\sum^d_{\substack{k,l=1\\k\neq l}}\mathbb{E}\left[c_k(W_k,\,w^{\mathcal{T}}_{ik})\right]\mathbb{E}\left[c_l(W_l,\,w^{\mathcal{T}}_{jl})\right]+\sum_{k=1}^d\mathbb{E}\left[c_k(W_k,\,w^{\mathcal{T}}_{ik})c_k(W_k,\,w^{\mathcal{T}}_{jk})\right]\\
&+\sum_{k=1}^d\xi_{ik}\sum_{k=1}^pc_k(z_k,\,z^{\mathcal{T}}_{jk})+\sum_{k=1}^d\xi_{jk}\sum_{k=1}^pc_k(z_k,\,z^{\mathcal{T}}_{ik})+\sum_{k=1}^pc_k(z_k,\,z^{\mathcal{T}}_{ik})\sum_{k=1}^pc_k(z_k,\,z^{\mathcal{T}}_{jk})\\
=&\sum^d_{\substack{k,l=1\\k\neq l}}\xi_{ik}\xi_{jl}+\sum_{k=1}^d\zeta_{ijk}+\sum_{k=1}^d\xi_{ik}\sum_{k=1}^pc_k(z_k,\,z^{\mathcal{T}}_{jk})\\
&+\sum_{k=1}^d\xi_{jk}\sum_{k=1}^pc_k(z_k,\,z^{\mathcal{T}}_{ik})+\sum_{k=1}^pc_k(z_k,\,z^{\mathcal{T}}_{ik})\sum_{k=1}^pc_k(z_k,\,z^{\mathcal{T}}_{jk})
\end{align*}
in case of additive form, in which
\begin{equation*}
\zeta_{ijk}\eqdef\mathbb{E}\left[c_k(W_k,\,w^{\mathcal{T}}_{ik})c_k(W_k,\,w^{\mathcal{T}}_{jk})\right].   
\end{equation*}
\end{itemize}

\subsubsection{\texorpdfstring{{Derivation of $\mathbb{E}\left[\sigma^2_g(\mathbf{W},\mathbf{z})\right]$}}{Derivation of the variance - 2}}

Replacing $\sigma^2_g(\mathbf{W},\mathbf{z})$ by equation~\eqref{eq:krigingvar}:
\begin{align*}
  \mathbb{E}\left[\sigma^2_g(\cdot,\cdot)\right]=&\sigma^2\,\mathbb{E}\Big[1+\eta-\mathbf{r}^\top(\mathbf{W},\,\mathbf{z})\mathbf{R}^{-1}\mathbf{r}(\mathbf{W},\,\mathbf{z})+\left(\mathbf{h}(\mathbf{W},\,\mathbf{z})-\widetilde{\mathbf{H}}^\top\mathbf{R}^{-1}\mathbf{r}(\mathbf{W},\,\mathbf{z})\right)^\top\nonumber\\
  &\times\left(\widetilde{\mathbf{H}}^\top\mathbf{R}^{-1}\widetilde{\mathbf{H}}\right)^{-1}\left(\mathbf{h}(\mathbf{W},\,\mathbf{z})-\widetilde{\mathbf{H}}^\top\mathbf{R}^{-1}\mathbf{r}(\mathbf{W},\,\mathbf{z})\right)\Big]\nonumber\\
 =&\sigma^2(1+\eta)+\sigma^2\,\mathbb{E}\Bigg[\mathbf{h}^\top(\mathbf{W},\,\mathbf{z})\left(\widetilde{\mathbf{H}}^\top\mathbf{R}^{-1}\widetilde{\mathbf{H}}\right)^{-1}\mathbf{h}(\mathbf{W},\,\mathbf{z})\nonumber\\
 &+\mathbf{r}^\top(\mathbf{W},\,\mathbf{z})\left\{\mathbf{R}^{-1}\widetilde{\mathbf{H}}\left(\widetilde{\mathbf{H}}^\top\mathbf{R}^{-1}\widetilde{\mathbf{H}}\right)^{-1}\widetilde{\mathbf{H}}^\top\mathbf{R}^{-1}-\mathbf{R}^{-1}\right\}\mathbf{r}(\mathbf{W},\,\mathbf{z})\nonumber\\
 &-2\mathrm{tr}\left\{\mathbf{h}^\top(\mathbf{W},\,\mathbf{z})\left(\widetilde{\mathbf{H}}^\top\mathbf{R}^{-1}\widetilde{\mathbf{H}}\right)^{-1}\widetilde{\mathbf{H}}^\top\mathbf{R}^{-1}\mathbf{r}(\mathbf{W},\,\mathbf{z})\right\}\Bigg]\nonumber\\
=&\sigma^2(1+\eta)+\sigma^2\,\mathbb{E}\left[\mathbf{h}^\top(\mathbf{W},\,\mathbf{z})\left(\widetilde{\mathbf{H}}^\top\mathbf{R}^{-1}\widetilde{\mathbf{H}}\right)^{-1}\mathbf{h}(\mathbf{W},\,\mathbf{z})\right]\nonumber\\
&+\sigma^2\,\mathbb{E}\left[\mathbf{r}^\top(\mathbf{W},\,\mathbf{z})\left\{\mathbf{R}^{-1}\widetilde{\mathbf{H}}\left(\widetilde{\mathbf{H}}^\top\mathbf{R}^{-1}\widetilde{\mathbf{H}}\right)^{-1}\widetilde{\mathbf{H}}^\top\mathbf{R}^{-1}-\mathbf{R}^{-1}\right\}\mathbf{r}(\mathbf{W},\,\mathbf{z})\right]\nonumber\\
&-2\sigma^2\,\mathbb{E}\left[\mathrm{tr}\left\{\mathbf{h}^\top(\mathbf{W},\,\mathbf{z})\left(\widetilde{\mathbf{H}}^\top\mathbf{R}^{-1}\widetilde{\mathbf{H}}\right)^{-1}\widetilde{\mathbf{H}}^\top\mathbf{R}^{-1}\mathbf{r}(\mathbf{W},\,\mathbf{z})\right\}\right]\nonumber\\
=&\sigma^2\,\left[1+\eta+\mathrm{tr}\left\{\mathbf{C}\mathbf{P}\right\}+\mathbf{G}^\top\mathbf{C}\mathbf{G}+\mathrm{tr}\left\{\mathbf{Q}\mathbf{J}\right\}-2\mathrm{tr}\left\{\mathbf{C}\widetilde{\mathbf{H}}^\top\mathbf{R}^{-1}\mathbf{K}\right\}\right],
\end{align*}
where
\begin{itemize}
\item $\mathbf{C}=\left(\widetilde{\mathbf{H}}^\top\mathbf{R}^{-1}\widetilde{\mathbf{H}}\right)^{-1}\in\mathbb{R}^{(d+q)\times(d+q)}$ with $\widetilde{\mathbf{H}}=\left[\mathbf{w}^{\mathcal{T}},\mathbf{H}(\mathbf{z}^{\mathcal{T}})\right]\in\mathbb{R}^{m\times(d+q)}$;
\item $\mathbf{P}=\mathrm{Var}\left[\mathbf{h}(\mathbf{W},\,\mathbf{z})\right]=\mathrm{Var}\left[\left(\mathbf{W}^\top,\,\mathbf{h}(\mathbf{z})^\top\right)^\top\right]=\mathrm{blkdiag}(\boldsymbol{\Omega},\,\mathbf{0})\in\mathbb{R}^{(d+q)\times (d+q)}\,$;
\item $\mathbf{G}=\mathbb{E}\left[\mathbf{h}(\mathbf{W},\,\mathbf{z})\right]=\mathbb{E}\left[\left(\mathbf{W}^\top,\,\mathbf{h}(\mathbf{z})^\top\right)^\top\right]=[\boldsymbol{\mu}^\top,\,\mathbf{h}(\mathbf{z})^\top]^\top\in\mathbb{R}^{(d+q)\times 1}\,$;
\item $\mathbf{Q}=\mathbf{R}^{-1}\widetilde{\mathbf{H}}\left(\widetilde{\mathbf{H}}^\top\mathbf{R}^{-1}\widetilde{\mathbf{H}}\right)^{-1}\widetilde{\mathbf{H}}^\top\mathbf{R}^{-1}-\mathbf{R}^{-1}\in\mathbb{R}^{m\times m}\,$;
\end{itemize}
and
\begin{equation*}
\mathbf{K}=\mathbb{E}\left[\mathbf{h}(\mathbf{W},\,\mathbf{z})\mathbf{r}^\top(\mathbf{W},\,\mathbf{z})\right]^\top=\left[\mathbf{B}^\top,\,\mathbf{I}\mathbf{h}(\mathbf{z})^\top\right]\in\mathbb{R}^{m\times(d+q)}.
\end{equation*}

\subsubsection{\texorpdfstring{{Derivation of $\mu^2_I$}}{Derivation of the variance - 3}}

Using equation~\eqref{eq:krigingmean}, we have
\begin{align*}
\mu^2_I=&\left(\boldsymbol{\mu}^\top\widehat{\boldsymbol{\theta}}+\mathbf{h}(\mathbf{z})^\top\widehat{\boldsymbol{\beta}}+\mathbf{I}^\top\mathbf{A}\right)\left(\boldsymbol{\mu}^\top\widehat{\boldsymbol{\theta}}+\mathbf{h}(\mathbf{z})^\top\widehat{\boldsymbol{\beta}}+\mathbf{I}^\top\mathbf{A}\right)^\top\nonumber\\
=&\left(\boldsymbol{\mu}^\top\widehat{\boldsymbol{\theta}}+\mathbf{h}(\mathbf{z})^\top\widehat{\boldsymbol{\beta}}+\mathbf{I}^\top\mathbf{A}\right)\left(\widehat{\boldsymbol{\theta}}^\top\boldsymbol{\mu}+\widehat{\boldsymbol{\beta}}^\top\mathbf{h}(\mathbf{z})+\mathbf{A}^\top\mathbf{I}\right)\nonumber\\
=&\boldsymbol{\mu}^\top\widehat{\boldsymbol{\theta}}\widehat{\boldsymbol{\theta}}^\top\boldsymbol{\mu}+\left(\mathbf{h}(\mathbf{z})^\top\widehat{\boldsymbol{\beta}}\right)^2+\mathbf{I}^\top\mathbf{A}\mathbf{A}^\top\mathbf{I}+2\widehat{\boldsymbol{\theta}}^\top\boldsymbol{\mu}\mathbf{h}(\mathbf{z})^\top\widehat{\boldsymbol{\beta}}+2\widehat{\boldsymbol{\theta}}^\top\boldsymbol{\mu}\mathbf{I}^\top\mathbf{A}+2\mathbf{h}(\mathbf{z})^\top\widehat{\boldsymbol{\beta}}\mathbf{I}^\top\mathbf{A}\nonumber\\
=&\mathrm{tr}\left\{\widehat{\boldsymbol{\theta}}\widehat{\boldsymbol{\theta}}^\top\boldsymbol{\mu}\boldsymbol{\mu}^\top\right\}+\left(\mathbf{h}(\mathbf{z})^\top\widehat{\boldsymbol{\beta}}\right)^2+\mathrm{tr}\left\{\mathbf{A}\mathbf{A}^\top\mathbf{I}\mathbf{I}^\top\right\}+2\widehat{\boldsymbol{\theta}}^\top\boldsymbol{\mu}\mathbf{h}(\mathbf{z})^\top\widehat{\boldsymbol{\beta}}+2\left[\widehat{\boldsymbol{\theta}}^\top\boldsymbol{\mu}+\mathbf{h}(\mathbf{z})^\top\widehat{\boldsymbol{\beta}}\right]\mathbf{I}^\top\mathbf{A}
\end{align*}

Finally, we obtain the expression for~\eqref{eq:sigma2L}, which is given by
\begin{align}
\label{eq:sigma2}
\sigma^2_I=&\mathrm{tr}\left\{\mathbf{A}\mathbf{A}^\top\mathbf{J}\right\}-\mathrm{tr}\left\{\mathbf{A}\mathbf{A}^\top\mathbf{I}\mathbf{I}^\top\right\}+2\widehat{\boldsymbol{\theta}}^\top\mathbf{B}\mathbf{A}-2\widehat{\boldsymbol{\theta}}^\top\boldsymbol{\mu}\mathbf{I}^\top\mathbf{A}+\mathrm{tr}\left\{\widehat{\boldsymbol{\theta}}\widehat{\boldsymbol{\theta}}^\top\boldsymbol{\Omega}\right\}\nonumber\\
&+\sigma^2\,\left(1+\eta+\mathrm{tr}\left\{\mathbf{C}\mathbf{P}\right\}+\mathbf{G}^\top\mathbf{C}\mathbf{G}+\mathrm{tr}\left\{\mathbf{Q}\mathbf{J}\right\}-2\mathrm{tr}\left\{\mathbf{C}\widetilde{\mathbf{H}}^\top\mathbf{R}^{-1}\mathbf{K}\right\}\right)\nonumber\\\
=&\mathbf{A}^\top\left(\mathbf{J}-\mathbf{I}\mathbf{I}^\top\right)\mathbf{A}+2\widehat{\boldsymbol{\theta}}^\top\left(\mathbf{B}-\boldsymbol{\mu}\mathbf{I}^\top\right)\mathbf{A}+\mathrm{tr}\left\{\widehat{\boldsymbol{\theta}}\widehat{\boldsymbol{\theta}}^\top\boldsymbol{\Omega}\right\}\nonumber\\\
&+\sigma^2\,\left(1+\eta+\mathrm{tr}\left\{\mathbf{Q}\mathbf{J}\right\}+\mathbf{G}^\top\mathbf{C}\mathbf{G}+\mathrm{tr}\left\{\mathbf{C}\mathbf{P}-2\mathbf{C}\widetilde{\mathbf{H}}^\top\mathbf{R}^{-1}\mathbf{K}\right\}\right).
\end{align}

This together with equation~\eqref{eq:muL} completes the proof. In case that the trend is assumed constant, the expressions for $\mu_I$ and $\sigma^2_I$ can be simplified to the following:

\begin{align*}
\mu_I=&\left(\mathbf{1}_m^\top\mathbf{R}^{-1}\mathbf{1}_m\right)^{-1}\mathbf{1}_m^\top\mathbf{R}^{-1}\mathbf{y}^{\mathcal{T}}+\mathbf{I}^\top\mathbf{A},\\  
\sigma^2_I=&\mathbf{A}^\top\left(\mathbf{J}-\mathbf{I}\mathbf{I}^\top\right)\mathbf{A}+\sigma^2\,\left(1+\eta+\mathrm{tr}\left\{\mathbf{Q}\mathbf{J}\right\}+\mathbf{C}-\mathrm{tr}\left\{2\mathbf{C}\mathbf{1}_m^\top\mathbf{R}^{-1}\mathbf{I}\right\}\right),
\end{align*}
where 
\begin{itemize}
\item $\mathbf{A}=\mathbf{R}^{-1}\left(\mathbf{y}^{\mathcal{T}}-\mathbf{1}_m\left(\mathbf{1}_m^\top\mathbf{R}^{-1}\mathbf{1}_m\right)^{-1}\mathbf{1}_m^\top\mathbf{R}^{-1}\mathbf{y}^{\mathcal{T}}\right)$;
\item $\mathbf{Q}=\mathbf{R}^{-1}\mathbf{1}_m\mathbf{C}\mathbf{1}_m^\top\mathbf{R}^{-1}-\mathbf{R}^{-1}$;
\item $\mathbf{C}=\left(\mathbf{1}_m^\top\mathbf{R}^{-1}\mathbf{1}_m\right)^{-1}$.
\end{itemize}

\section[Proof of Proposition~3.2]{Proof of Proposition~\ref{prop:kernel}}
\label{sec:proofprop}
\begin{lemma}
\label{lemma:partial}
Denote
\begin{equation*}
    \Gamma[m]=\int^{a}_{b}\frac{x^m}{\sigma\sqrt{2\pi}}\exp\left\{-\frac{(x-\mu)^2}{2\sigma^2}\right\}\mathrm{d}x
\end{equation*}
for $m\in\mathbb{N}_0\,$, where $a\in\mathbb{R}\,$, $b\in\mathbb{R}\,$, $\mu\in\mathbb{R}$ and $\sigma\in\mathbb{R}_{\geq0}\,$. Then, we have
\begin{align*}
\Gamma[0]=&\Phi\left(\frac{a-\mu}{\sigma}\right)-\Phi\left(\frac{b-\mu}{\sigma}\right),\\
\Gamma[1]=&\mu\left[\Phi\left(\frac{a-\mu}{\sigma}\right)-\Phi\left(\frac{b-\mu}{\sigma}\right)\right]+\frac{\sigma}{\sqrt{2\pi}}\left[\exp\left\{-\frac{(b-\mu)^2}{2\sigma^2}\right\}-\exp\left\{-\frac{(a-\mu)^2}{2\sigma^2}\right\}\right],\\
\Gamma[2]=&\left(\mu^2+\sigma^2\right)\left[\Phi\left(\frac{a-\mu}{\sigma}\right)-\Phi\left(\frac{b-\mu}{\sigma}\right)\right]\\
&\qquad\qquad+\frac{(\mu+b)\sigma}{\sqrt{2\pi}}\exp\left\{-\frac{(b-\mu)^2}{2\sigma^2}\right\}-\frac{(\mu+a)\sigma}{\sqrt{2\pi}}\exp\left\{-\frac{(a-\mu)^2}{2\sigma^2}\right\},\\
\Gamma[3]=&\left(\mu^3+3\mu\sigma^2\right)\left[\Phi\left(\frac{a-\mu}{\sigma}\right)-\Phi\left(\frac{b-\mu}{\sigma}\right)\right]\\
&\qquad\qquad+\frac{(b^2+\mu b+\mu^2+2\sigma^2)\sigma}{\sqrt{2\pi}}\exp\left\{-\frac{(b-\mu)^2}{2\sigma^2}\right\}\\
&\qquad\qquad-\frac{(a^2+\mu a+\mu^2+2\sigma^2)\sigma}{\sqrt{2\pi}}\exp\left\{-\frac{(a-\mu)^2}{2\sigma^2}\right\},\\
\Gamma[4]=&\left(\mu^4+3\sigma^4+6\mu^2\sigma^2\right)\left[\Phi\left(\frac{a-\mu}{\sigma}\right)-\Phi\left(\frac{b-\mu}{\sigma}\right)\right]\\
&\qquad\qquad+\frac{(b^3+\mu^3+\mu^2 b+\mu b^2+3\sigma^2b+5\sigma^2\mu)\sigma}{\sqrt{2\pi}}\exp\left\{-\frac{(b-\mu)^2}{2\sigma^2}\right\}\\
&\qquad\qquad-\frac{(a^3+\mu^3+\mu^2 a+\mu a^2+3\sigma^2a+5\sigma^2\mu)\sigma}{\sqrt{2\pi}}\exp\left\{-\frac{(a-\mu)^2}{2\sigma^2}\right\},
\end{align*}
where $\Phi(\cdot)$ denotes the cumulative density function of the standard normal.
\end{lemma}

\begin{proof}
Denote
\begin{equation*}
    \kappa[m]=\int^{s}_{t}\frac{x^m}{\sqrt{2\pi}}\exp\left\{-\frac{x^2}{2}\right\}\mathrm{d}x
\end{equation*}
for $m\in\mathbb{N}_0\,$, where $s\in\mathbb{R}$ and $t\in\mathbb{R}\,$. Then via integration by parts, we have
\begin{align*}
\kappa[m]=&\frac{1}{\sqrt{2\pi}}\left(-x^{m-1}e^{-\frac{x^2}{2}}\bigg\vert^s_t+(m-1)\int^{s}_{t}x^{m-2}e^{-\frac{x^2}{2}}\mathrm{d}x\right)\nonumber\\
=&\frac{1}{\sqrt{2\pi}}\left(t^{m-1}e^{-\frac{t^2}{2}}-s^{m-1}e^{-\frac{s^2}{2}}\right)+(m-1)\int^{s}_{t}x^{m-2}e^{-\frac{x^2}{2}}\mathrm{d}x\nonumber\\
=&\frac{1}{\sqrt{2\pi}}\left(t^{m-1}e^{-\frac{t^2}{2}}-s^{m-1}e^{-\frac{s^2}{2}}\right)+(m-1)\kappa[m-2]\nonumber.
\end{align*}
Thus, we have
\begin{align}
\label{eq:eq1}
\kappa[0]=&\int^{s}_{t}\frac{1}{\sqrt{2\pi}}\exp\left\{-\frac{x^2}{2}\right\}\mathrm{d}x=\Phi(s)-\Phi(t),\\
\label{eq:eq2}
\kappa[1]=&\int^{s}_{t}\frac{x}{\sqrt{2\pi}}\exp\left\{-\frac{x^2}{2}\right\}\mathrm{d}x\nonumber\\
=&-\frac{1}{\sqrt{2\pi}}e^{-\frac{x^2}{2}}\bigg\vert^s_t\nonumber\\
=&\frac{1}{\sqrt{2\pi}}\left(e^{-\frac{t^2}{2}}-e^{-\frac{s^2}{2}}\right),\\
\label{eq:eq3}
\kappa[2]=&\frac{1}{\sqrt{2\pi}}\left(te^{-\frac{t^2}{2}}-se^{-\frac{s^2}{2}}\right)+\kappa[0]\nonumber\\
=&\frac{1}{\sqrt{2\pi}}\left(te^{-\frac{t^2}{2}}-se^{-\frac{s^2}{2}}\right)+\Phi(s)-\Phi(t),
\end{align}
and
\begin{align}
\label{eq:eq4}
\kappa[3]=&\frac{1}{\sqrt{2\pi}}\left(t^{2}e^{-\frac{t^2}{2}}-s^{2}e^{-\frac{s^2}{2}}\right)+2\kappa[1]\nonumber\\
=&\frac{1}{\sqrt{2\pi}}\left(t^{2}e^{-\frac{t^2}{2}}-s^{2}e^{-\frac{s^2}{2}}\right)+\frac{2}{\sqrt{2\pi}}\left(e^{-\frac{t^2}{2}}-e^{-\frac{s^2}{2}}\right),\\
\label{eq:eq5}
\kappa[4]=&\frac{1}{\sqrt{2\pi}}\left(t^{3}e^{-\frac{t^2}{2}}-s^{3}e^{-\frac{s^2}{2}}\right)+3\kappa[2]\nonumber\\
=&\frac{1}{\sqrt{2\pi}}\left(t^{3}e^{-\frac{t^2}{2}}-s^{3}e^{-\frac{s^2}{2}}\right)+\frac{3}{\sqrt{2\pi}}\left(te^{-\frac{t^2}{2}}-se^{-\frac{s^2}{2}}\right)+3\left[\Phi(s)-\Phi(t)\right],
\end{align}
where $\Phi(\cdot)$ denotes the cumulative density function of the standard normal. 

Denote
\begin{equation*}
    \Gamma[m]=\int^{a}_{b}\frac{x^m}{\sigma\sqrt{2\pi}}\exp\left\{-\frac{(x-\mu)^2}{2\sigma^2}\right\}\mathrm{d}x
\end{equation*}
for $m\in\mathbb{N}_0\,$, where $a\in\mathbb{R}\,$, $b\in\mathbb{R}\,$, $\mu\in\mathbb{R}$ and $\sigma\in\mathbb{R}_{\geq0}\,$. Let
\begin{equation*}
s=\frac{x-\mu}{\sigma},
\end{equation*}
then we have
\begin{equation*}
\Gamma[m]=\int^{\frac{a-\mu}{\sigma}}_{\frac{b-\mu}{\sigma}}\frac{(\sigma s+\mu)^m}{\sqrt{2\pi}}\exp\left\{-\frac{s^2}{2}\right\}\mathrm{d}s    
\end{equation*}
for $m\in\mathbb{N}_0\,$. The lemma is subsequently proved by using equations~\eqref{eq:eq1},~\eqref{eq:eq2},~\eqref{eq:eq3},~\eqref{eq:eq4} and~\eqref{eq:eq5} for all $m\in\{0,\dots,4\}$. 
\end{proof}

\subsection{Derivation for Exponential Case}
\label{sec:exp}

\subsubsection{\texorpdfstring{{Derivation of $\xi_{ik}$}}{Derivation of the first expectation}}
\begin{align*}
    \xi_{ik}=&\mathbb{E}\left[c_k(W_k,\,w^{\mathcal{T}}_{ik})\right]\\
    =&\int\exp\left\{-\frac{|w-w^{\mathcal{T}}_{ik}|}{\gamma_k}\right\}\frac{1}{\sigma_k\sqrt{2\pi}}\exp\left\{-\frac{(w-\mu_k)^2}{2\sigma_k^2}\right\}\mathrm{d}w\\
    =&\int^{+\infty}_{w^{\mathcal{T}}_{ik}}\frac{1}{\sigma_k\sqrt{2\pi}}\exp\left\{-\frac{w-w^{\mathcal{T}}_{ik}}{\gamma_k}-\frac{(w-\mu_k)^2}{2\sigma_k^2}\right\}\mathrm{d}w+\int^{w^{\mathcal{T}}_{ik}}_{-\infty}\frac{1}{\sigma_k\sqrt{2\pi}}\exp\left\{\frac{w-w^{\mathcal{T}}_{ik}}{\gamma_k}-\frac{(w-\mu_k)^2}{2\sigma_k^2}\right\}\mathrm{d}w\\
    =&\exp\left\{\frac{\sigma^2_k+2\gamma_k\left(w^{\mathcal{T}}_{ik}-\mu_k\right)}{2\gamma_k^2}\right\}\int^{+\infty}_{w^{\mathcal{T}}_{ik}}\frac{1}{\sigma_k\sqrt{2\pi}}\exp\left\{-\frac{(w-\mu_A)^2}{2\sigma_k^2}\right\}\mathrm{d}w\\
    &+\exp\left\{\frac{\sigma^2_k-2\gamma_k\left(w^{\mathcal{T}}_{ik}-\mu_k\right)}{2\gamma_k^2}\right\}\int^{w^{\mathcal{T}}_{ik}}_{-\infty}\frac{1}{\sigma_k\sqrt{2\pi}}\exp\left\{-\frac{(w-\mu_B)^2}{2\sigma_k^2}\right\}\mathrm{d}w,
    \end{align*}
    where the last step is obtained by completing the square. Using Lemma~\ref{lemma:partial}, we then have  
    \begin{align*}
    \xi_{ik}=&\exp\left\{\frac{\sigma^2_k+2\gamma_k\left(w^{\mathcal{T}}_{ik}-\mu_k\right)}{2\gamma_k^2}\right\}\Phi\left(\frac{\mu_A-w^{\mathcal{T}}_{ik}}{\sigma_k}\right)+\exp\left\{\frac{\sigma^2_k-2\gamma_k\left(w^{\mathcal{T}}_{ik}-\mu_k\right)}{2\gamma_k^2}\right\}\Phi\left(\frac{w^{\mathcal{T}}_{ik}-\mu_B}{\sigma_k}\right),
\end{align*}
where
\begin{equation*}
    \mu_A=\mu_k-\frac{\sigma^2_k}{\gamma_k}\quad\mathrm{and}\quad\mu_B=\mu_k+\frac{\sigma^2_k}{\gamma_k}.
\end{equation*}
    
\subsubsection{\texorpdfstring{{Derivation of $\zeta_{ijk}$}}{Derivation of the second expectation}}
    \begin{align}
      \zeta_{ijk}=&\mathbb{E}\left[c_k(W_k,\,w^{\mathcal{T}}_{ik})c_k(W_k,\,w^{\mathcal{T}}_{jk})\right]\nonumber\\
      =&\int\frac{1}{\sigma_k\sqrt{2\pi}}\exp\left\{-\frac{|w-w^{\mathcal{T}}_{ik}|}{\gamma_k}-\frac{|w-w^{\mathcal{T}}_{jk}|}{\gamma_k}-\frac{(w-\mu_k)^2}{2\sigma_k^2}\right\}\mathrm{d}w\nonumber\\
      \label{eq:c1}
      =&\int^{+\infty}_{w^{\mathcal{T}}_{jk}}\frac{1}{\sigma_k\sqrt{2\pi}}\exp\left\{-\frac{w-w^{\mathcal{T}}_{ik}}{\gamma_k}-\frac{w-w^{\mathcal{T}}_{jk}}{\gamma_k}-\frac{(w-\mu_k)^2}{2\sigma_k^2}\right\}\mathrm{d}w\\
      \label{eq:c2}
      &+\int^{w^{\mathcal{T}}_{jk}}_{w^{\mathcal{T}}_{ik}}\frac{1}{\sigma_k\sqrt{2\pi}}\exp\left\{-\frac{w-w^{\mathcal{T}}_{ik}}{\gamma_k}-\frac{w^{\mathcal{T}}_{jk}-w}{\gamma_k}-\frac{(w-\mu_k)^2}{2\sigma_k^2}\right\}\mathrm{d}w\\
      \label{eq:c3}
      &+\int^{w^{\mathcal{T}}_{ik}}_{-\infty}\frac{1}{\sigma_k\sqrt{2\pi}}\exp\left\{-\frac{w^{\mathcal{T}}_{ik}-w}{\gamma_k}-\frac{w^{\mathcal{T}}_{jk}-w}{\gamma_k}-\frac{(w-\mu_k)^2}{2\sigma_k^2}\right\}\mathrm{d}w,
    \end{align}
    where $w^{\mathcal{T}}_{ik}\leq w^{\mathcal{T}}_{jk}$ is assumed.
    
   By completing the square, term~\eqref{eq:c1} can be rewritten as follow:
    \begin{align*}
        \eqref{eq:c1}=&\exp\left\{\frac{2\sigma^2_k+\gamma_k\left(w^{\mathcal{T}}_{ik}+w^{\mathcal{T}}_{jk}-2\mu_k\right)}{\gamma_k^2}\right\}\int^{+\infty}_{w^{\mathcal{T}}_{jk}}\frac{1}{\sigma_k\sqrt{2\pi}}\exp\left\{-\frac{(w-\mu_C)^2}{2\sigma_k^2}\right\}\mathrm{d}w,
    \end{align*}
    where 
    \begin{equation*}
    \mu_C=\mu_k-\dfrac{2\sigma^2_k}{\gamma_k}.   
    \end{equation*}
    Then by Lemma~\ref{lemma:partial}, we obtain
    \begin{align*}
        \eqref{eq:c1}=&\exp\left\{\frac{2\sigma^2_k+\gamma_k\left(w^{\mathcal{T}}_{ik}+w^{\mathcal{T}}_{jk}-2\mu_k\right)}{\gamma_k^2}\right\}\Phi\left(\frac{\mu_C-w^{\mathcal{T}}_{jk}}{\sigma_k}\right).
    \end{align*}
    
    Since term~\eqref{eq:c3} can be rewritten as
    \begin{align*}
        \eqref{eq:c3}=&\int^{w^{\mathcal{T}}_{ik}}_{-\infty}\frac{1}{\sigma_k\sqrt{2\pi}}\exp\left\{-\frac{w^{\mathcal{T}}_{ik}-w}{\gamma_k}-\frac{w^{\mathcal{T}}_{jk}-w}{\gamma_k}-\frac{(w-\mu_k)^2}{2\sigma_k^2}\right\}\mathrm{d}w\\
        =&\int^{+\infty}_{-w^{\mathcal{T}}_{ik}}\frac{1}{\sigma_k\sqrt{2\pi}}\exp\left\{-\frac{w+w^{\mathcal{T}}_{ik}}{\gamma_k}-\frac{w+w^{\mathcal{T}}_{jk}}{\gamma_k}-\frac{(w+\mu_k)^2}{2\sigma_k^2}\right\}\mathrm{d}w,
    \end{align*}
the form of which allows us to obtain solution of term~\eqref{eq:c3} by simply using that of term~\eqref{eq:c1}. Thus, we have   
\begin{align*}
        \eqref{eq:c3}=&\exp\left\{\frac{2\sigma^2_k-\gamma_k\left(w^{\mathcal{T}}_{ik}+w^{\mathcal{T}}_{jk}-2\mu_k\right)}{\gamma_k^2}\right\}\Phi\left(\frac{w^{\mathcal{T}}_{ik}-\mu_D}{\sigma_k}\right),
    \end{align*}
    where 
    \begin{equation*}
    \mu_D=\mu_k+\dfrac{2\sigma^2_k}{\gamma_k}\,. 
    \end{equation*}
    Term~\eqref{eq:c2} is obtained as follow:
    \begin{align*}
        \eqref{eq:c2}=&\int^{w^{\mathcal{T}}_{jk}}_{w^{\mathcal{T}}_{ik}}\frac{1}{\sigma_k\sqrt{2\pi}}\exp\left\{-\frac{w^{\mathcal{T}}_{jk}-w^{\mathcal{T}}_{ik}}{\gamma_k}-\frac{(w-\mu_k)^2}{2\sigma_k^2}\right\}\mathrm{d}w\\
        =&\exp\left\{-\frac{w^{\mathcal{T}}_{jk}-w^{\mathcal{T}}_{ik}}{\gamma_k}\right\}\int^{w^{\mathcal{T}}_{jk}}_{w^{\mathcal{T}}_{ik}}\frac{1}{\sigma_k\sqrt{2\pi}}\exp\left\{-\frac{(w-\mu_k)^2}{2\sigma_k^2}\right\}\mathrm{d}w\\
        =&\exp\left\{-\frac{w^{\mathcal{T}}_{jk}-w^{\mathcal{T}}_{ik}}{\gamma_k}\right\}\left[\Phi\left(\frac{w^{\mathcal{T}}_{jk}-\mu_k}{\sigma_k}\right)-\Phi\left(\frac{w^{\mathcal{T}}_{ik}-\mu_k}{\sigma_k}\right)\right],
    \end{align*}
    where the last step uses Lemma~\ref{lemma:partial}. Therefore, we obtain that
   \begin{align}
   \label{eq:c4}
      \zeta_{ijk}=&\exp\left\{\frac{2\sigma^2_k+\gamma_k\left(w^{\mathcal{T}}_{ik}+w^{\mathcal{T}}_{jk}-2\mu_k\right)}{\gamma_k^2}\right\}\Phi\left(\frac{\mu_C-w^{\mathcal{T}}_{jk}}{\sigma_k}\right)\nonumber\\
      &+\exp\left\{-\frac{w^{\mathcal{T}}_{jk}-w^{\mathcal{T}}_{ik}}{\gamma_k}\right\}\left[\Phi\left(\frac{w^{\mathcal{T}}_{jk}-\mu_k}{\sigma_k}\right)-\Phi\left(\frac{w^{\mathcal{T}}_{ik}-\mu_k}{\sigma_k}\right)\right]\nonumber\\
      &+\exp\left\{\frac{2\sigma^2_k-\gamma_k\left(w^{\mathcal{T}}_{ik}+w^{\mathcal{T}}_{jk}-2\mu_k\right)}{\gamma_k^2}\right\}\Phi\left(\frac{w^{\mathcal{T}}_{ik}-\mu_D}{\sigma_k}\right)
      \end{align}
      for $w^{\mathcal{T}}_{ik}\leq w^{\mathcal{T}}_{jk}$. Observe that
      \begin{equation*}
       \mathbb{E}\left[c_k(W_k,\,w^{\mathcal{T}}_{ik})c_k(W_k,\,w^{\mathcal{T}}_{jk})\right]=\mathbb{E}\left[c_k(W_k,\,w^{\mathcal{T}}_{jk})c_k(W_k,\,w^{\mathcal{T}}_{ik})\right],  
      \end{equation*}
       Thus, the expression for $\zeta_{ijk}$ when $w^{\mathcal{T}}_{ik}> w^{\mathcal{T}}_{jk}$ is obtained by simply interchanging the positions of $w^{\mathcal{T}}_{ik}$ and $w^{\mathcal{T}}_{jk}$ in formula~\eqref{eq:c4}.
      
\subsubsection{\texorpdfstring{{Derivation of $\psi_{jk}$}}{Derivation of the third expectation}}
    \begin{align*}
     \psi_{jk}=&\mathbb{E}\left[W_kc_k(W_k,\,w^{\mathcal{T}}_{jk})\right]\\
     =&\int\exp\left\{-\frac{|w-w^{\mathcal{T}}_{jk}|}{\gamma_k}\right\}\frac{w}{\sigma_k\sqrt{2\pi}}\exp\left\{-\frac{(w-\mu_k)^2}{2\sigma_k^2}\right\}\mathrm{d}w\\
    =&\int^{+\infty}_{w^{\mathcal{T}}_{jk}}\frac{w}{\sigma_k\sqrt{2\pi}}\exp\left\{-\frac{w-w^{\mathcal{T}}_{jk}}{\gamma_k}-\frac{(w-\mu_k)^2}{2\sigma_k^2}\right\}\mathrm{d}w+\int^{w^{\mathcal{T}}_{jk}}_{-\infty}\frac{w}{\sigma_k\sqrt{2\pi}}\exp\left\{\frac{w-w^{\mathcal{T}}_{jk}}{\gamma_k}-\frac{(w-\mu_k)^2}{2\sigma_k^2}\right\}\mathrm{d}w\\
    =&\exp\left\{\frac{\sigma^2_k+2\gamma_k\left(w^{\mathcal{T}}_{jk}-\mu_k\right)}{2\gamma_k^2}\right\}\int^{+\infty}_{w^{\mathcal{T}}_{jk}}\frac{w}{\sigma_k\sqrt{2\pi}}\exp\left\{-\frac{(w-\mu_A)^2}{2\sigma_k^2}\right\}\mathrm{d}w\\
    &+\exp\left\{\frac{\sigma^2_k-2\gamma_k\left(w^{\mathcal{T}}_{jk}-\mu_k\right)}{2\gamma_k^2}\right\}\int^{w^{\mathcal{T}}_{jk}}_{-\infty}\frac{w}{\sigma_k\sqrt{2\pi}}\exp\left\{-\frac{(w-\mu_B)^2}{2\sigma_k^2}\right\}\mathrm{d}w,
    \end{align*}
    where the last step is obtained by completing the square. 
    
    Thus, by Lemma~\ref{lemma:partial} we have
    \begin{align*}
    \psi_{jk}=&\exp\left\{\frac{\sigma^2_k+2\gamma_k\left(w^{\mathcal{T}}_{jk}-\mu_k\right)}{2\gamma_k^2}\right\}\left[\mu_A\Phi\left(\frac{\mu_A-w^{\mathcal{T}}_{jk}}{\sigma_k}\right)+\frac{\sigma_k}{\sqrt{2\pi}}\exp\left\{-\frac{\left(w^{\mathcal{T}}_{jk}-\mu_A\right)^2}{2\sigma^2_k}\right\}\right]\\
    &+\exp\left\{\frac{\sigma^2_k-2\gamma_k\left(w^{\mathcal{T}}_{jk}-\mu_k\right)}{2\gamma_k^2}\right\}\left[-\mu_B\Phi\left(\frac{w^{\mathcal{T}}_{jk}-\mu_B}{\sigma_k}\right)+\frac{\sigma_k}{\sqrt{2\pi}}\exp\left\{-\frac{\left(w^{\mathcal{T}}_{jk}-\mu_B\right)^2}{2\sigma^2_k}\right\}\right].
     \end{align*}

\subsection{Derivation for Squared Exponential Case}
\label{sec:sexp}
\subsubsection{\texorpdfstring{{Derivation of $\xi_{ik}$}}{Derivation of the first expectation}}

\begin{align*}
    \xi_{ik}=&\mathbb{E}\left[c_k(W_k,\,w^{\mathcal{T}}_{ik})\right]\\
    =&\int\exp\left\{-\left(\frac{w-w^{\mathcal{T}}_{ik}}{\gamma_k}\right)^2\right\}\frac{1}{\sigma_k\sqrt{2\pi}}\exp\left\{-\frac{(w-\mu_k)^2}{2\sigma_k^2}\right\}\mathrm{d}w\\
    =&\int\frac{1}{\sigma_k\sqrt{2\pi}}\exp\left\{-\frac{\left(w-w^{\mathcal{T}}_{ik}\right)^2}{\gamma_k^2}-\frac{(w-\mu_k)^2}{2\sigma_k^2}\right\}\mathrm{d}w\\
    =&\exp\left\{-\frac{\left(\mu_k-w^{\mathcal{T}}_{ik}\right)^2}{2\sigma^2_k+\gamma_k^2}\right\}\int\frac{1}{\sigma_k\sqrt{2\pi}}\exp\left\{-\frac{2\sigma_k^2+\gamma_k^2}{2\sigma^2_k\gamma_k^2}\left[w-\frac{2\sigma^2_kw^{\mathcal{T}}_{ik}+\gamma_k^2\mu_k}{2\sigma^2_k+\gamma_k^2}\right]^2\right\}\mathrm{d}w,
    \end{align*}   
    where the last step is obtained by completing the square. Consequently, 
    
    \begin{align*}
    \xi_{ik}=&\frac{1}{\sqrt{1+2\sigma^2_k/\gamma_k^2}}\exp\left\{-\frac{\left(\mu_k-w^{\mathcal{T}}_{ik}\right)^2}{2\sigma^2_k+\gamma_k^2}\right\}\int\frac{\sqrt{2\sigma_k^2+\gamma_k^2}}{\sigma_k\gamma_k\sqrt{2\pi}}\exp\left\{-\frac{2\sigma_k^2+\gamma_k^2}{2\sigma^2_k\gamma_k^2}\left[w-\frac{2\sigma^2_kw^{\mathcal{T}}_{ik}+\gamma_k^2\mu_k}{2\sigma^2_k+\gamma_k^2}\right]^2\right\}\mathrm{d}w\\
    =&\frac{1}{\sqrt{1+2\sigma^2_k/\gamma_k^2}}\exp\left\{-\frac{\left(\mu_k-w^{\mathcal{T}}_{ik}\right)^2}{2\sigma^2_k+\gamma_k^2}\right\},
    \end{align*}
    where the last step uses the fact that the integral in the first step equals to one because it integrates the probability density function of a normal distribution with mean and variance equal to  
   \begin{equation*}
    \frac{2\sigma^2_kw^{\mathcal{T}}_{ik}+\gamma_k^2\mu_k}{2\sigma^2_k+\gamma_k^2}\quad\mathrm{and}\quad \frac{\sigma^2_k\gamma_k^2}{2\sigma_k^2+\gamma_k^2}   
   \end{equation*}
    respectively.   
\subsubsection{\texorpdfstring{{Derivation of $\zeta_{ijk}$}}{Derivation of the second expectation}}

 \begin{align*}
      \zeta_{ijk}=&\mathbb{E}\left[c_k(W_k,\,w^{\mathcal{T}}_{ik})c_k(W_k,\,w^{\mathcal{T}}_{jk})\right]\\
      =&\int\frac{1}{\sigma_k\sqrt{2\pi}}\exp\left\{-\frac{\left(w-w^{\mathcal{T}}_{ik}\right)^2}{\gamma_k^2}-\frac{\left(w-w^{\mathcal{T}}_{jk}\right)^2}{\gamma_k^2}-\frac{(w-\mu_k)^2}{2\sigma_k^2}\right\}\mathrm{d}w.
    \end{align*}
    By applying the completing in square, we can obtain the following:
    \begin{align*}
     \zeta_{ijk}=&\frac{1}{\sqrt{1+4\sigma^2_k/\gamma_k^2}}\exp\left\{-\frac{\left(\frac{w^{\mathcal{T}}_{ik}+w^{\mathcal{T}}_{jk}}{2}-\mu_k\right)^2}{\gamma_k^2/2+2\sigma^2_k}-\frac{\left(w^{\mathcal{T}}_{ik}-w^{\mathcal{T}}_{jk}\right)^2}{2\gamma_k^2}\right\}\int\frac{1}{\sigma_*\sqrt{2\pi}}\exp\left\{-\frac{(w-\mu_*)^2}{2\sigma_*^2}\right\}\mathrm{d}w,
    \end{align*}
    where
    \begin{equation*}
        \mu_*=\frac{2\sigma^2_k\left(w^{\mathcal{T}}_{ik}+w^{\mathcal{T}}_{jk}\right)+\gamma_k^2\mu_k}{4\sigma^2_k+\gamma_k^2}\quad\mathrm{and}\quad\sigma^2_*=\frac{\sigma^2_k\gamma_k^2}{4\sigma_k^2+\gamma_k^2}.
    \end{equation*}
    Thus, we have
    \begin{equation*}
    \zeta_{ijk}=\frac{1}{\sqrt{1+4\sigma^2_k/\gamma_k^2}}\exp\left\{-\frac{\left(\frac{w^{\mathcal{T}}_{ik}+w^{\mathcal{T}}_{jk}}{2}-\mu_k\right)^2}{\gamma_k^2/2+2\sigma^2_k}-\frac{\left(w^{\mathcal{T}}_{ik}-w^{\mathcal{T}}_{jk}\right)^2}{2\gamma_k^2}\right\}. 
    \end{equation*}

\subsubsection{\texorpdfstring{{Derivation of $\psi_{jk}$}}{Derivation of the third expectation}}

\begin{align*}
     \psi_{jk}&=\mathbb{E}\left[W_kc_k(W_k,\,w^{\mathcal{T}}_{jk})\right]\\
     &=\int\frac{w}{\sigma_k\sqrt{2\pi}}\exp\left\{-\frac{\left(w-w^{\mathcal{T}}_{jk}\right)^2}{\gamma_k^2}-\frac{(w-\mu_k)^2}{2\sigma_k^2}\right\}\mathrm{d}w\\
     &=\frac{1}{\sqrt{1+2\sigma^2_k/\gamma_k^2}}\exp\left\{-\frac{\left(\mu_k-w^{\mathcal{T}}_{jk}\right)^2}{2\sigma^2_k+\gamma_k^2}\right\}\int\frac{w}{\sigma_*\sqrt{2\pi}}\exp\left\{-\frac{(w-\mu_*)^2}{2\sigma_*^2}\right\}\mathrm{d}w,
     \end{align*}
     where the last step is obtained by completing in square; and
     \begin{equation*}
        \mu_*=\frac{2\sigma^2_kw^{\mathcal{T}}_{jk}+\gamma_k^2\mu_k}{2\sigma^2_k+\gamma_k^2}\quad\mathrm{and}\quad\sigma^2_*=\frac{\sigma^2_k\gamma_k^2}{2\sigma_k^2+\gamma_k^2}. 
    \end{equation*}
    Realising that the integral 
    \begin{equation*}
     \int\frac{w}{\sigma_*\sqrt{2\pi}}\exp\left\{-\frac{(w-\mu_*)^2}{2\sigma_*^2}\right\}\mathrm{d}w   
    \end{equation*}
    is in fact the expectation of a normal random variable with mean $\mu_*$ and variance $\sigma^2_*\,$, we have
    \begin{equation*}
    \psi_{jk}=\frac{1}{\sqrt{1+2\sigma^2_k/\gamma_k^2}}\exp\left\{-\frac{\left(\mu_k-w^{\mathcal{T}}_{jk}\right)^2}{2\sigma^2_k+\gamma_k^2}\right\}\frac{2\sigma^2_kw^{\mathcal{T}}_{jk}+\gamma_k^2\mu_k}{2\sigma^2_k+\gamma_k^2}.
    \end{equation*}

\subsection{Derivation for Mat\'{e}rn-1.5 Case}
\label{sec:matern1.5}
\subsubsection{\texorpdfstring{{Derivation of $\xi_{ik}$}}{Derivation of the first expectation}}
\begin{align}
   \xi_{ik}=&\mathbb{E}\left[c_k(W_k,\,w^{\mathcal{T}}_{ik})\right]\nonumber\\ 
    =&\int\left(1+\frac{\sqrt{3}|w-w^{\mathcal{T}}_{ik}|}{\gamma_k}\right)\frac{1}{\sigma_k\sqrt{2\pi}}\exp\left\{-\frac{\sqrt{3}|w-w^{\mathcal{T}}_{ik}|}{\gamma_k}-\frac{(w-\mu_k)^2}{2\sigma^2_k}\right\}\mathrm{d}w\nonumber\\
    \label{eq:e1}
    =&\int^{+\infty}_{w^{\mathcal{T}}_{ik}}\left(1+\frac{\sqrt{3}\left(w-w^{\mathcal{T}}_{ik}\right)}{\gamma_k}\right)\frac{1}{\sigma_k\sqrt{2\pi}}\exp\left\{-\frac{\sqrt{3}\left(w-w^{\mathcal{T}}_{ik}\right)}{\gamma_k}-\frac{(w-\mu_k)^2}{2\sigma^2_k}\right\}\mathrm{d}w\\
    \label{eq:e2}
     &+\int^{w^{\mathcal{T}}_{ik}}_{-\infty}\left(1+\frac{\sqrt{3}\left(w^{\mathcal{T}}_{ik}-w\right)}{\gamma_k}\right)\frac{1}{\sigma_k\sqrt{2\pi}}\exp\left\{\frac{\sqrt{3}\left(w-w^{\mathcal{T}}_{ik}\right)}{\gamma_k}-\frac{(w-\mu_k)^2}{2\sigma^2_k}\right\}\mathrm{d}w.
 \end{align}   
 
 We first calculate term~\eqref{eq:e1} by completing in square:
     \begin{equation*}
         \eqref{eq:e1}=\exp\left\{\frac{3\sigma^2_k+2\sqrt{3}\gamma_k\left(w^{\mathcal{T}}_{ik}-\mu_k\right)}{2\gamma_k^2}\right\}\int^{+\infty}_{w^{\mathcal{T}}_{ik}}\left[E_{11}w+E_{10}\right]\frac{1}{\sigma_k\sqrt{2\pi}}\exp\left\{-\frac{(w-\mu_A)^2}{2\sigma^2_k}\right\},
     \end{equation*}
     where
     \begin{equation*}
         E_{10}=1-\frac{\sqrt{3}w^{\mathcal{T}}_{ik}}{\gamma_k},\quad
         E_{11}=\frac{\sqrt{3}}{\gamma_k}\quad\mathrm{and}\quad
         \mu_A=\mu_k-\frac{\sqrt{3}\sigma^2_k}{\gamma_k}.
     \end{equation*}
     
     By Lemma~\ref{lemma:partial}, we then obtain
     \begin{align*}
        \eqref{eq:e1}=&\exp\left\{\frac{3\sigma^2_k+2\sqrt{3}\gamma_k\left(w^{\mathcal{T}}_{ik}-\mu_k\right)}{2\gamma_k^2}\right\}\left[\mathbf{E}^\top_1\boldsymbol{\Lambda}_{11}\Phi\left(\frac{\mu_A-w^{\mathcal{T}}_{ik}}{\sigma_k}\right)+\mathbf{E}^\top_1\boldsymbol{\Lambda}_{12}\frac{\sigma_k}{\sqrt{2\pi}}\exp\left\{-\frac{(w^{\mathcal{T}}_{ik}-\mu_A)^2}{2\sigma^2_k}\right\}\right],
     \end{align*}
     where
     \begin{equation*}
      \mathbf{E}_1=[E_{10},\,E_{11}]^\top,\quad\boldsymbol{\Lambda}_{11}=[1,\,\mu_A]^\top\quad\mathrm{and}\quad\boldsymbol{\Lambda}_{12}=[0,\,1]^\top.
     \end{equation*}
         
     Term~\eqref{eq:e2} can be rewritten as follow:
     \begin{align*}
         \eqref{eq:e2}=&\int^{w^{\mathcal{T}}_{ik}}_{-\infty}\left(1+\frac{\sqrt{3}\left(w^{\mathcal{T}}_{ik}-w\right)}{\gamma_k}\right)\frac{1}{\sigma_k\sqrt{2\pi}}\exp\left\{\frac{\sqrt{3}\left(w-w^{\mathcal{T}}_{ik}\right)}{\gamma_k}-\frac{(w-\mu_k)^2}{2\sigma^2_k}\right\}\mathrm{d}w\\
         =&\int^{+\infty}_{-w^{\mathcal{T}}_{ik}}\left(1+\frac{\sqrt{3}\left(w+w^{\mathcal{T}}_{ik}\right)}{\gamma_k}\right)\frac{1}{\sigma_k\sqrt{2\pi}}\exp\left\{-\frac{\sqrt{3}\left(w+w^{\mathcal{T}}_{ik}\right)}{\gamma_k}-\frac{(w+\mu_k)^2}{2\sigma^2_k}\right\}\mathrm{d}w,
     \end{align*}
     the form of which allows us to obtain solution of term~\eqref{eq:e2} by simply using that of term~\eqref{eq:e1}. Thus, we have
      \begin{align*}
        \eqref{eq:e2}=&\exp\left\{\frac{3\sigma^2_k-2\sqrt{3}\gamma_k\left(w^{\mathcal{T}}_{ik}-\mu_k\right)}{2\gamma_k^2}\right\}\left[\mathbf{E}^\top_2\boldsymbol{\Lambda}_{21}\Phi\left(\frac{w^{\mathcal{T}}_{ik}-\mu_B}{\sigma_k}\right)+\mathbf{E}^\top_2\boldsymbol{\Lambda}_{22}\frac{\sigma_k}{\sqrt{2\pi}}\exp\left\{-\frac{(w^{\mathcal{T}}_{ik}-\mu_B)^2}{2\sigma^2_k}\right\}\right],
     \end{align*}
     where
     \begin{equation*}
     \mathbf{E}_2=[E_{20},\,E_{21}]^\top,\quad\boldsymbol{\Lambda}_{21}=[1,\,-\mu_B]^\top\quad\mathrm{and}\quad\boldsymbol{\Lambda}_{22}=[0,\,1]^\top	
     \end{equation*}
     with
     \begin{equation*}
         E_{20}=1+\frac{\sqrt{3}w^{\mathcal{T}}_{ik}}{\gamma_k},\quad
         E_{21}=\frac{\sqrt{3}}{\gamma_k}\quad\mathrm{and}\quad
         \mu_B=\mu_k+\frac{\sqrt{3}\sigma^2_k}{\gamma_k}.
     \end{equation*}
     Finally, we have
     \begin{align*}
        \xi_{ik}=&\exp\left\{\frac{3\sigma^2_k+2\sqrt{3}\gamma_k\left(w^{\mathcal{T}}_{ik}-\mu_k\right)}{2\gamma_k^2}\right\}\left[\mathbf{E}^\top_1\boldsymbol{\Lambda}_{11}\Phi\left(\frac{\mu_A-w^{\mathcal{T}}_{ik}}{\sigma_k}\right)+\mathbf{E}^\top_1\boldsymbol{\Lambda}_{12}\frac{\sigma_k}{\sqrt{2\pi}}\exp\left\{-\frac{(w^{\mathcal{T}}_{ik}-\mu_A)^2}{2\sigma^2_k}\right\}\right]\\
         &+\exp\left\{\frac{3\sigma^2_k-2\sqrt{3}\gamma_k\left(w^{\mathcal{T}}_{ik}-\mu_k\right)}{2\gamma_k^2}\right\}\left[\mathbf{E}^\top_2\boldsymbol{\Lambda}_{21}\Phi\left(\frac{w^{\mathcal{T}}_{ik}-\mu_B}{\sigma_k}\right)+\mathbf{E}^\top_2\boldsymbol{\Lambda}_{22}\frac{\sigma_k}{\sqrt{2\pi}}\exp\left\{-\frac{(w^{\mathcal{T}}_{ik}-\mu_B)^2}{2\sigma^2_k}\right\}\right].
     \end{align*}

\subsubsection{\texorpdfstring{{Derivation of $\zeta_{ijk}$}}{Derivation of the second expectation}}

\begin{align*}
      \zeta_{ijk}=&\mathbb{E}\left[c_k(W_k,\,w^{\mathcal{T}}_{ik})c_k(W_k,\,w^{\mathcal{T}}_{jk})\right]\\
      =&\int\left(1+\frac{\sqrt{3}|w-w^{\mathcal{T}}_{ik}|}{\gamma_k}\right)\left(1+\frac{\sqrt{3}|w-w^{\mathcal{T}}_{jk}|}{\gamma_k}\right)\\
      &\qquad\qquad\times\frac{1}{\sigma_k\sqrt{2\pi}}\exp\left\{-\frac{\sqrt{3}|w-w^{\mathcal{T}}_{ik}|+\sqrt{3}|w-w^{\mathcal{T}}_{jk}|}{\gamma_k}-\frac{(w-\mu_k)^2}{2\sigma^2_k}\right\}\mathrm{d}w.
      \end{align*}
      
      Assume that $w^{\mathcal{T}}_{ik}\leq w^{\mathcal{T}}_{jk}\,$, we have
      \begin{align}
      \label{eq:e4}
      \zeta_{ijk}=&\int^{+\infty}_{w^{\mathcal{T}}_{jk}}\left(1+\frac{\sqrt{3}(w-w^{\mathcal{T}}_{ik})}{\gamma_k}\right)\left(1+\frac{\sqrt{3}(w-w^{\mathcal{T}}_{jk})}{\gamma_k}\right)\nonumber\\
      &\times\frac{1}{\sigma_k\sqrt{2\pi}}\exp\left\{-\frac{\sqrt{3}(w-w^{\mathcal{T}}_{ik})+\sqrt{3}(w-w^{\mathcal{T}}_{jk})}{\gamma_k}-\frac{(w-\mu_k)^2}{2\sigma^2_k}\right\}\mathrm{d}w\\
      \label{eq:e5}
      &+\int^{w^{\mathcal{T}}_{jk}}_{w^{\mathcal{T}}_{ik}}\left(1+\frac{\sqrt{3}(w-w^{\mathcal{T}}_{ik})}{\gamma_k}\right)\left(1+\frac{\sqrt{3}(w^{\mathcal{T}}_{jk}-w)}{\gamma_k}\right)\nonumber\\
      &\times\frac{1}{\sigma_k\sqrt{2\pi}}\exp\left\{-\frac{\sqrt{3}(w-w^{\mathcal{T}}_{ik})+\sqrt{3}(w^{\mathcal{T}}_{jk}-w)}{\gamma_k}-\frac{(w-\mu_k)^2}{2\sigma^2_k}\right\}\mathrm{d}w\\
       \label{eq:e6}
      &+\int^{w^{\mathcal{T}}_{ik}}_{-\infty}\left(1+\frac{\sqrt{3}(w^{\mathcal{T}}_{ik}-w)}{\gamma_k}\right)\left(1+\frac{\sqrt{3}(w^{\mathcal{T}}_{jk}-w)}{\gamma_k}\right)\nonumber\\
      &\times\frac{1}{\sigma_k\sqrt{2\pi}}\exp\left\{-\frac{\sqrt{3}(w^{\mathcal{T}}_{ik}-w)+\sqrt{3}(w^{\mathcal{T}}_{jk}-w)}{\gamma_k}-\frac{(w-\mu_k)^2}{2\sigma^2_k}\right\}\mathrm{d}w.
      \end{align}
      
      We first calculate term~\eqref{eq:e4} by expanding the product of two brackets after the integral sign:
      \begin{multline*}
       \eqref{eq:e4}=\int^{+\infty}_{w^{\mathcal{T}}_{jk}}(E_{32}w^2+E_{31}w+E_{30})\frac{1}{\sigma_k\sqrt{2\pi}}\exp\left\{-\frac{\sqrt{3}(w-w^{\mathcal{T}}_{ik})+\sqrt{3}(w-w^{\mathcal{T}}_{jk})}{\gamma_k}-\frac{(w-\mu_k)^2}{2\sigma^2_k}\right\}\mathrm{d}w,
      \end{multline*}
      where
      \begin{equation*}
       E_{30}=1+\frac{3w^{\mathcal{T}}_{ik}w^{\mathcal{T}}_{jk}-\sqrt{3}\gamma_k\left(w^{\mathcal{T}}_{ik}+w^{\mathcal{T}}_{jk}\right)}{\gamma_k^2},\quad
       E_{31}=\frac{2\sqrt{3}\gamma_k-3\left(w^{\mathcal{T}}_{ik}+w^{\mathcal{T}}_{jk}\right)}{\gamma_k^2}\quad\mathrm{and}\quad E_{32}=\frac{3}{\gamma_k^2}.
      \end{equation*}
      Then by completing in square, we have
      \begin{multline*}
        \eqref{eq:e4}=\exp\left\{\frac{6\sigma_k^2+\sqrt{3}\gamma_k\left(w^{\mathcal{T}}_{ik}+w^{\mathcal{T}}_{jk}-2\mu_k\right)}{\gamma_k^2}\right\}\\
       \times\int^{+\infty}_{w^{\mathcal{T}}_{jk}}(E_{32}w^2+E_{31}w+E_{30})\frac{1}{\sigma_k\sqrt{2\pi}}\exp\left\{-\frac{(w-\mu_{C})^2}{2\sigma^2_k}\right\}\mathrm{d}w,
      \end{multline*}
      where
      \begin{equation*}
        \mu_{C}=\mu_k-2\sqrt{3}\frac{\sigma^2_k}{\gamma_k}. 
      \end{equation*}
      
      Using Lemma~\ref{lemma:partial} and arranging terms, we obtain
      \begin{multline*}
        \eqref{eq:e4}=\exp\left\{\frac{6\sigma_k^2+\sqrt{3}\gamma_k\left(w^{\mathcal{T}}_{ik}+w^{\mathcal{T}}_{jk}-2\mu_k\right)}{\gamma_k^2}\right\}\\
       \times\left[\mathbf{E}^\top_3\boldsymbol{\Lambda}_{31}\Phi\left(\frac{\mu_C-w^{\mathcal{T}}_{jk}}{\sigma_k}\right)+\mathbf{E}^\top_3\boldsymbol{\Lambda}_{32}\frac{\sigma_k}{\sqrt{2\pi}}\exp\left\{-\frac{(w^{\mathcal{T}}_{jk}-\mu_{C})^2}{2\sigma^2_k}\right\}\right],
      \end{multline*}
      where
      \begin{equation*}
      \mathbf{E}_3=[E_{30},\,E_{31},\,E_{32}]^\top,\quad	\boldsymbol{\Lambda}_{31}=[1,\,\mu_C,\,\mu_C^2+\sigma^2_k]^\top\quad\mathrm{and}\quad\boldsymbol{\Lambda}_{32}=[0,\,1,\,\mu_C+w^{\mathcal{T}}_{jk}]^\top.
      \end{equation*}

      The derivation of term~\eqref{eq:e5} is analogue to that of term~\eqref{eq:e4}. By expanding the product of two brackets after the integral sign, we have
      \begin{equation*}
       \eqref{eq:e5}=\int^{w^{\mathcal{T}}_{jk}}_{w^{\mathcal{T}}_{ik}}(E_{42}w^2+E_{41}w+E_{40})\frac{1}{\sigma_k\sqrt{2\pi}}\exp\left\{-\frac{\sqrt{3}(w-w^{\mathcal{T}}_{ik})+\sqrt{3}(w^{\mathcal{T}}_{jk}-w)}{\gamma_k}-\frac{(w-\mu_k)^2}{2\sigma^2_k}\right\}\mathrm{d}w,
      \end{equation*}
      where
      \begin{equation*}
        E_{40}=1+\frac{\sqrt{3}\gamma_k\left(w^{\mathcal{T}}_{jk}-w^{\mathcal{T}}_{ik}\right)-3w^{\mathcal{T}}_{ik}w^{\mathcal{T}}_{jk}}{\gamma_k^2},\quad
        E_{41}=\frac{3\left(w^{\mathcal{T}}_{ik}+w^{\mathcal{T}}_{jk}\right)}{\gamma_k^2}\quad\mathrm{and}\quad E_{42}=-\frac{3}{\gamma_k^2}.	
      \end{equation*}

      Then by completing in square, we have
      \begin{equation*}
        \eqref{eq:e5}=\exp\left\{-\frac{\sqrt{3}\left(w^{\mathcal{T}}_{jk}-w^{\mathcal{T}}_{ik}\right)}{\gamma_k}\right\}\int^{w^{\mathcal{T}}_{jk}}_{w^{\mathcal{T}}_{ik}}(E_{42}w^2+E_{41}w+E_{40})\frac{1}{\sigma_k\sqrt{2\pi}}\exp\left\{-\frac{(w-\mu_k)^2}{2\sigma^2_k}\right\}\mathrm{d}w.
      \end{equation*}
      
      Using Lemma~\ref{lemma:partial} and arranging terms, we obtain 
      \begin{align*}
        \eqref{eq:e5}=&\exp\left\{-\frac{\sqrt{3}\left(w^{\mathcal{T}}_{jk}-w^{\mathcal{T}}_{ik}\right)}{\gamma_k}\right\}\Bigg[\mathbf{E}^\top_4\boldsymbol{\Lambda}_{41}\left(\Phi\left(\frac{w^{\mathcal{T}}_{jk}-\mu_k}{\sigma_k}\right)-\Phi\left(\frac{w^{\mathcal{T}}_{ik}-\mu_k}{\sigma_k}\right)\right)\\
        &\qquad+\mathbf{E}^\top_4\boldsymbol{\Lambda}_{42}\frac{\sigma_k}{\sqrt{2\pi}}\exp\left\{-\frac{(w^{\mathcal{T}}_{ik}-\mu_k)^2}{2\sigma^2_k}\right\}-\mathbf{E}^\top_4\boldsymbol{\Lambda}_{43}\frac{\sigma_k}{\sqrt{2\pi}}\exp\left\{-\frac{(w^{\mathcal{T}}_{jk}-\mu_k)^2}{2\sigma^2_k}\right\}\Bigg],
      \end{align*}
      where
      \begin{equation*}
      \mathbf{E}_4=[E_{40},\,E_{41},\,E_{42}]^\top,\quad \boldsymbol{\Lambda}_{41}=[1,\,\mu_k,\,\mu_k^2+\sigma^2_k]^\top,\quad \boldsymbol{\Lambda}_{42}=[0,\,1,\,\mu_k+w^{\mathcal{T}}_{ik}]^\top\quad\mathrm{and}\quad \boldsymbol{\Lambda}_{43}=[0,\,1,\,\mu_k+w^{\mathcal{T}}_{jk}]^\top.  
      \end{equation*}
      Term~\eqref{eq:e6} can then be computed in the following way:
      \begin{align*}
          \eqref{eq:e6}=&\int^{w^{\mathcal{T}}_{ik}}_{-\infty}\left(1+\frac{\sqrt{3}(w^{\mathcal{T}}_{ik}-w)}{\gamma_k}\right)\left(1+\frac{\sqrt{3}(w^{\mathcal{T}}_{jk}-w)}{\gamma_k}\right)\\
      &\qquad\times\frac{1}{\sigma_k\sqrt{2\pi}}\exp\left\{-\frac{\sqrt{3}(w^{\mathcal{T}}_{ik}-w)+\sqrt{3}(w^{\mathcal{T}}_{jk}-w)}{\gamma_k}-\frac{(w-\mu_k)^2}{2\sigma^2_k}\right\}\mathrm{d}w\\
      =&\int^{+\infty}_{-w^{\mathcal{T}}_{ik}}\left(1+\frac{\sqrt{3}(w+w^{\mathcal{T}}_{ik})}{\gamma_k}\right)\left(1+\frac{\sqrt{3}(w+w^{\mathcal{T}}_{jk})}{\gamma_k}\right)\\
      &\qquad\times\frac{1}{\sigma_k\sqrt{2\pi}}\exp\left\{-\frac{\sqrt{3}(w+w^{\mathcal{T}}_{ik})+\sqrt{3}(w+w^{\mathcal{T}}_{jk})}{\gamma_k}-\frac{(w+\mu_k)^2}{2\sigma^2_k}\right\}\mathrm{d}w,
      \end{align*}
      the form of which allows us to obtain solution of term~\eqref{eq:e6} by simply using that of term~\eqref{eq:e4}. Thus, we have
      \begin{multline*}
        \eqref{eq:e6}=\exp\left\{\frac{6\sigma_k^2-\sqrt{3}\gamma_k\left(w^{\mathcal{T}}_{ik}+w^{\mathcal{T}}_{jk}-2\mu_k\right)}{\gamma_k^2}\right\}\\
        \times\left[\mathbf{E}^\top_5\boldsymbol{\Lambda}_{51}\Phi\left(\frac{w^{\mathcal{T}}_{ik}-\mu_D}{\sigma_k}\right)+\mathbf{E}^\top_5\boldsymbol{\Lambda}_{52}\frac{\sigma_k}{\sqrt{2\pi}}\exp\left\{-\frac{(w^{\mathcal{T}}_{ik}-\mu_{D})^2}{2\sigma^2_k}\right\}\right],
      \end{multline*}
      where
      \begin{equation*}
      \mathbf{E}_5=[E_{50},\,E_{51},\,E_{52}]^\top,\quad \boldsymbol{\Lambda}_{51}=[1,\,-\mu_D,\,\mu_D^2+\sigma^2_k]^\top\quad\mathrm{and}\quad \boldsymbol{\Lambda}_{52}=[0,\,1,\,-\mu_D-w^{\mathcal{T}}_{ik}]^\top 
      \end{equation*}
      with
      \begin{itemize}
      \item $E_{50}=1+\dfrac{3w^{\mathcal{T}}_{ik}w^{\mathcal{T}}_{jk}+\sqrt{3}\gamma_k\left(w^{\mathcal{T}}_{ik}+w^{\mathcal{T}}_{jk}\right)}{\gamma_k^2}\quad \mathrm{and}\quad E_{51}=\dfrac{2\sqrt{3}\gamma_k+3\left(w^{\mathcal{T}}_{ik}+w^{\mathcal{T}}_{jk}\right)}{\gamma_k^2}$;
      \item $E_{52}=\dfrac{3}{\gamma_k^2}\quad\mathrm{and}\quad \mu_D=\mu_k+2\sqrt{3}\dfrac{\sigma^2_k}{\gamma_k}$.
      \end{itemize}
      
      Therefore, the expression for $\zeta_{ijk}$ when $w^{\mathcal{T}}_{ik}\leq w^{\mathcal{T}}_{jk}$ is given by
      \begin{align}
      \label{eq:e7}
      \zeta_{ijk}=&\exp\left\{\frac{6\sigma_k^2+\sqrt{3}\gamma_k\left(w^{\mathcal{T}}_{ik}+w^{\mathcal{T}}_{jk}-2\mu_k\right)}{\gamma_k^2}\right\}\nonumber\\
        &\quad\times\left[\mathbf{E}^\top_3\boldsymbol{\Lambda}_{31}\Phi\left(\frac{\mu_C-w^{\mathcal{T}}_{jk}}{\sigma_k}\right)+\mathbf{E}^\top_3\boldsymbol{\Lambda}_{32}\frac{\sigma_k}{\sqrt{2\pi}}\exp\left\{-\frac{(w^{\mathcal{T}}_{jk}-\mu_{C})^2}{2\sigma^2_k}\right\}\right]\nonumber\\
        &+\exp\left\{-\frac{\sqrt{3}\left(w^{\mathcal{T}}_{jk}-w^{\mathcal{T}}_{ik}\right)}{\gamma_k}\right\}\Bigg[\mathbf{E}^\top_4\boldsymbol{\Lambda}_{41}\left(\Phi\left(\frac{w^{\mathcal{T}}_{jk}-\mu_k}{\sigma_k}\right)-\Phi\left(\frac{w^{\mathcal{T}}_{ik}-\mu_k}{\sigma_k}\right)\right)\nonumber\\
        &\quad+\mathbf{E}^\top_4\boldsymbol{\Lambda}_{42}\frac{\sigma_k}{\sqrt{2\pi}}\exp\left\{-\frac{(w^{\mathcal{T}}_{ik}-\mu_k)^2}{2\sigma^2_k}\right\}-\mathbf{E}^\top_4\boldsymbol{\Lambda}_{43}\frac{\sigma_k}{\sqrt{2\pi}}\exp\left\{-\frac{(w^{\mathcal{T}}_{jk}-\mu_k)^2}{2\sigma^2_k}\right\}\Bigg]\nonumber\\
        &+\exp\left\{\frac{6\sigma_k^2-\sqrt{3}\gamma_k\left(w^{\mathcal{T}}_{ik}+w^{\mathcal{T}}_{jk}-2\mu_k\right)}{\gamma_k^2}\right\}\nonumber\\
        &\quad\times\left[\mathbf{E}^\top_5\boldsymbol{\Lambda}_{51}\Phi\left(\frac{w^{\mathcal{T}}_{ik}-\mu_D}{\sigma_k}\right)+\mathbf{E}^\top_5\boldsymbol{\Lambda}_{52}\frac{\sigma_k}{\sqrt{2\pi}}\exp\left\{-\frac{(w^{\mathcal{T}}_{ik}-\mu_{D})^2}{2\sigma^2_k}\right\}\right].\nonumber
      \end{align}
      
      Observe that
      \begin{equation*}
        \mathbb{E}\left[c_k(W_k,\,w^{\mathcal{T}}_{ik})c_k(W_k,\,w^{\mathcal{T}}_{jk})\right]=\mathbb{E}\left[c_k(W_k,\,w^{\mathcal{T}}_{jk})c_k(W_k,\,w^{\mathcal{T}}_{ik})\right].  
      \end{equation*}
      Thus, the expression for $\zeta_{ijk}$ when $w^{\mathcal{T}}_{ik}> w^{\mathcal{T}}_{jk}$ is obtained by simply interchanging the positions of $w^{\mathcal{T}}_{ik}$ and $w^{\mathcal{T}}_{jk}$ in the above formula of $\zeta_{ijk}$ when $w^{\mathcal{T}}_{ik}\leq w^{\mathcal{T}}_{jk}$.

\subsubsection{\texorpdfstring{{Derivation of $\psi_{jk}$}}{Derivation of the third expectation}}
\begin{align}
     \psi_{jk}=&\mathbb{E}\left[W_kc_k(W_k,\,w^{\mathcal{T}}_{jk})\right]\nonumber\\
     =&\int w\left(1+\frac{\sqrt{3}|w-w^{\mathcal{T}}_{jk}|}{\gamma_k}\right)\frac{1}{\sigma_k\sqrt{2\pi}}\exp\left\{-\frac{\sqrt{3}|w-w^{\mathcal{T}}_{jk}|}{\gamma_k}-\frac{(w-\mu_k)^2}{2\sigma^2_k}\right\}\mathrm{d}w\nonumber\\
     \label{eq:e8}
     =&\int^{+\infty}_{w^{\mathcal{T}}_{jk}}\left(w+\frac{\sqrt{3}w\left(w-w^{\mathcal{T}}_{jk}\right)}{\gamma_k}\right)\frac{1}{\sigma_k\sqrt{2\pi}}\exp\left\{-\frac{\sqrt{3}\left(w-w^{\mathcal{T}}_{jk}\right)}{\gamma_k}-\frac{(w-\mu_k)^2}{2\sigma^2_k}\right\}\mathrm{d}w\\
     \label{eq:e9}
     &+\int^{w^{\mathcal{T}}_{jk}}_{-\infty}\left(w+\frac{\sqrt{3}w\left(w^{\mathcal{T}}_{jk}-w\right)}{\gamma_k}\right)\frac{1}{\sigma_k\sqrt{2\pi}}\exp\left\{\frac{\sqrt{3}\left(w-w^{\mathcal{T}}_{jk}\right)}{\gamma_k}-\frac{(w-\mu_k)^2}{2\sigma^2_k}\right\}\mathrm{d}w.
     \end{align}
     
     We first calculate term~\eqref{eq:e8} by arranging the terms in the bracket after the integral sign and completing in square:
     \begin{equation*}
         \eqref{eq:e8}=\exp\left\{\frac{3\sigma^2_k+2\sqrt{3}\gamma_k\left(w^{\mathcal{T}}_{jk}-\mu_k\right)}{2\gamma_k^2}\right\}\int^{+\infty}_{w^{\mathcal{T}}_{jk}}\left[E_{11}w^2+E_{10}w\right]\frac{1}{\sigma_k\sqrt{2\pi}}\exp\left\{-\frac{(w-\mu_A)^2}{2\sigma^2_k}\right\}.
    \end{equation*}
     
     By Lemma~\ref{lemma:partial}, we then obtain
     \begin{equation*}
        \eqref{eq:e8}=\exp\left\{\frac{3\sigma^2_k+2\sqrt{3}\gamma_k\left(w^{\mathcal{T}}_{jk}-\mu_k\right)}{2\gamma_k^2}\right\}\left[\mathbf{E}^\top_1\boldsymbol{\Lambda}_{61}\Phi\left(\frac{\mu_A-w^{\mathcal{T}}_{jk}}{\sigma_k}\right)+\mathbf{E}^\top_1\boldsymbol{\Lambda}_{62}\frac{\sigma_k}{\sqrt{2\pi}}\exp\left\{-\frac{(w^{\mathcal{T}}_{jk}-\mu_A)^2}{2\sigma^2_k}\right\}\right],
     \end{equation*}
     where
     \begin{equation*}
     \boldsymbol{\Lambda}_{61}=[\mu_A,\,\mu^2_A+\sigma^2_k]^\top\quad\mathrm{and}\quad\boldsymbol{\Lambda}_{62}=[1,\,\mu_A+w^{\mathcal{T}}_{jk}]^\top.	
     \end{equation*}
     
     Term~\eqref{eq:e9} can be rewritten as follow:
     \begin{align*}
         \eqref{eq:e9}=&\int^{w^{\mathcal{T}}_{jk}}_{-\infty}\left(1+\frac{\sqrt{3}\left(w^{\mathcal{T}}_{jk}-w\right)}{\gamma_k}\right)\frac{w}{\sigma_k\sqrt{2\pi}}\exp\left\{\frac{\sqrt{3}\left(w-w^{\mathcal{T}}_{jk}\right)}{\gamma_k}-\frac{(w-\mu_k)^2}{2\sigma^2_k}\right\}\mathrm{d}w\\
         =&-\int^{+\infty}_{-w^{\mathcal{T}}_{jk}}\left(1+\frac{\sqrt{3}\left(w+w^{\mathcal{T}}_{jk}\right)}{\gamma_k}\right)\frac{w}{\sigma_k\sqrt{2\pi}}\exp\left\{-\frac{\sqrt{3}\left(w+w^{\mathcal{T}}_{jk}\right)}{\gamma_k}-\frac{(w+\mu_k)^2}{2\sigma^2_k}\right\}\mathrm{d}w,
     \end{align*}
     the form of which allows us to obtain the solution of term~\eqref{eq:e9} by simply using that of term~\eqref{eq:e8}. Thus, we have
      \begin{multline*}
        \eqref{eq:e9}=-\exp\left\{\frac{3\sigma^2_k-2\sqrt{3}\gamma_k\left(w^{\mathcal{T}}_{jk}-\mu_k\right)}{2\gamma_k^2}\right\}\\
        \times\left[\mathbf{E}^\top_2\boldsymbol{\Lambda}_{71}\Phi\left(\frac{w^{\mathcal{T}}_{jk}-\mu_B}{\sigma_k}\right)+\mathbf{E}^\top_2\boldsymbol{\Lambda}_{72}\frac{\sigma_k}{\sqrt{2\pi}}\exp\left\{-\frac{(w^{\mathcal{T}}_{jk}-\mu_B)^2}{2\sigma^2_k}\right\}\right],
     \end{multline*}
     where
     \begin{equation*}
     \boldsymbol{\Lambda}_{71}=[-\mu_B,\,\mu^2_B+\sigma^2_k]^\top\quad\mathrm{and}\quad\boldsymbol{\Lambda}_{72}=[1,\,-\mu_B-w^{\mathcal{T}}_{jk}]^\top.	
     \end{equation*} 
     
     Finally, we have
     \begin{align*}
         \psi_{jk}=&\exp\left\{\frac{3\sigma^2_k+2\sqrt{3}\gamma_k\left(w^{\mathcal{T}}_{jk}-\mu_k\right)}{2\gamma_k^2}\right\}\left[\mathbf{E}^\top_1\boldsymbol{\Lambda}_{61}\Phi\left(\frac{\mu_A-w^{\mathcal{T}}_{jk}}{\sigma_k}\right)+\mathbf{E}^\top_1\boldsymbol{\Lambda}_{62}\frac{\sigma_k}{\sqrt{2\pi}}\exp\left\{-\frac{(w^{\mathcal{T}}_{jk}-\mu_A)^2}{2\sigma^2_k}\right\}\right]\\
         &-\exp\left\{\frac{3\sigma^2_k-2\sqrt{3}\gamma_k\left(w^{\mathcal{T}}_{jk}-\mu_k\right)}{2\gamma_k^2}\right\}\left[\mathbf{E}^\top_2\boldsymbol{\Lambda}_{71}\Phi\left(\frac{w^{\mathcal{T}}_{jk}-\mu_B}{\sigma_k}\right)+\mathbf{E}^\top_2\boldsymbol{\Lambda}_{72}\frac{\sigma_k}{\sqrt{2\pi}}\exp\left\{-\frac{(w^{\mathcal{T}}_{jk}-\mu_B)^2}{2\sigma^2_k}\right\}\right].
     \end{align*}

\subsection{Derivation for Mat\'{e}rn-2.5 Case}
\label{sec:matern2.5}
\subsubsection{\texorpdfstring{{Derivation of $\xi_{ik}$}}{Derivation of the first expectation}}

\begin{align}
    \xi_{ik}=&\mathbb{E}\left[c_k(W_k,\,w^{\mathcal{T}}_{ik})\right]\nonumber\\ 
    =&\int\left(1+\frac{\sqrt{5}|w-w^{\mathcal{T}}_{ik}|}{\gamma_k}+\frac{5(w-w^{\mathcal{T}}_{ik})^2}{3\gamma_k^2}\right)\frac{1}{\sigma_k\sqrt{2\pi}}\exp\left\{-\frac{\sqrt{5}|w-w^{\mathcal{T}}_{ik}|}{\gamma_k}-\frac{(w-\mu_k)^2}{2\sigma^2_k}\right\}\mathrm{d}w\nonumber\\
    \label{eq:f1}
    =&\int^{+\infty}_{w^{\mathcal{T}}_{ik}}\left(1+\frac{\sqrt{5}\left(w-w^{\mathcal{T}}_{ik}\right)}{\gamma_k}+\frac{5}{3}\left(\frac{w-w^{\mathcal{T}}_{ik}}{\gamma_k}\right)^2\right)\frac{1}{\sigma_k\sqrt{2\pi}}\exp\left\{-\frac{\sqrt{5}\left(w-w^{\mathcal{T}}_{ik}\right)}{\gamma_k}-\frac{(w-\mu_k)^2}{2\sigma^2_k}\right\}\mathrm{d}w\\
      \label{eq:f2}
     &+\int^{w^{\mathcal{T}}_{ik}}_{-\infty}\left(1+\frac{\sqrt{5}\left(w^{\mathcal{T}}_{ik}-w\right)}{\gamma_k}+\frac{5}{3}\left(\frac{w-w^{\mathcal{T}}_{ik}}{\gamma_k}\right)^2\right)\frac{1}{\sigma_k\sqrt{2\pi}}\exp\left\{\frac{\sqrt{5}\left(w-w^{\mathcal{T}}_{ik}\right)}{\gamma_k}-\frac{(w-\mu_k)^2}{2\sigma^2_k}\right\}\mathrm{d}w.
 \end{align}   
 
 We first calculate term~\eqref{eq:f1} by arranging the terms in the bracket after the integral sign and completing the square:
     \begin{equation*}
         \eqref{eq:f1}=\exp\left\{\frac{5\sigma^2_k+2\sqrt{5}\gamma_k\left(w^{\mathcal{T}}_{ik}-\mu_k\right)}{2\gamma_k^2}\right\}\int^{+\infty}_{w^{\mathcal{T}}_{ik}}\left[E_{12}w^2+E_{11}w+E_{10}\right]\frac{1}{\sigma_k\sqrt{2\pi}}\exp\left\{-\frac{(w-\mu_A)^2}{2\sigma^2_k}\right\},
     \end{equation*}
     where
     \begin{equation*}
     E_{10}=1-\frac{\sqrt{5}w^{\mathcal{T}}_{ik}}{\gamma_k}+\frac{5\left(w^{\mathcal{T}}_{ik}\right)^2}{3\gamma_k^2},\quad
     E_{11}=\frac{\sqrt{5}}{\gamma_k}-\frac{10w^{\mathcal{T}}_{ik}}{3\gamma_k^2},\quad
     E_{12}=\frac{5}{3\gamma_k^2},\quad
     \mu_A=\mu_k-\frac{\sqrt{5}\sigma^2_k}{\gamma_k}.	
     \end{equation*}
          
     By Lemma~\ref{lemma:partial}, we then obtain
     \begin{equation*}
        \eqref{eq:f1}=\exp\left\{\frac{5\sigma^2_k+2\sqrt{5}\gamma_k\left(w^{\mathcal{T}}_{ik}-\mu_k\right)}{2\gamma_k^2}\right\}\left[\mathbf{E}^\top_1\boldsymbol{\Lambda}_{11}\Phi\left(\frac{\mu_A-w^{\mathcal{T}}_{ik}}{\sigma_k}\right)+\mathbf{E}^\top_1\boldsymbol{\Lambda}_{12}\frac{\sigma_k}{\sqrt{2\pi}}\exp\left\{-\frac{(w^{\mathcal{T}}_{ik}-\mu_A)^2}{2\sigma^2_k}\right\}\right],
     \end{equation*}
     where
     \begin{equation*}
     \mathbf{E}_1=[E_{10},\,E_{11},\,E_{12}]^\top,\quad \boldsymbol{\Lambda}_{11}=[1,\,\mu_A,\,\mu^2_A+\sigma^2_k]^\top\quad\mathrm{and}\quad \boldsymbol{\Lambda}_{12}=[0,\,1,\,\mu_A+w^{\mathcal{T}}_{ik}]^\top.
     \end{equation*}
     
     Term~\eqref{eq:f2} can be rewritten as follow:
     \begin{align*}
         \eqref{eq:f2}=&\int^{w^{\mathcal{T}}_{ik}}_{-\infty}\left(1+\frac{\sqrt{5}\left(w^{\mathcal{T}}_{ik}-w\right)}{\gamma_k}+\frac{5}{3}\left(\frac{w-w^{\mathcal{T}}_{ik}}{\gamma_k}\right)^2\right)\frac{1}{\sigma_k\sqrt{2\pi}}\exp\left\{\frac{\sqrt{5}\left(w-w^{\mathcal{T}}_{ik}\right)}{\gamma_k}-\frac{(w-\mu_k)^2}{2\sigma^2_k}\right\}\mathrm{d}w\\
         =&\int^{+\infty}_{-w^{\mathcal{T}}_{ik}}\left(1+\frac{\sqrt{5}\left(w+w^{\mathcal{T}}_{ik}\right)}{\gamma_k}+\frac{5}{3}\left(\frac{w+w^{\mathcal{T}}_{ik}}{\gamma_k}\right)^2\right)\frac{1}{\sigma_k\sqrt{2\pi}}\exp\left\{-\frac{\sqrt{5}\left(w+w^{\mathcal{T}}_{ik}\right)}{\gamma_k}-\frac{(w+\mu_k)^2}{2\sigma^2_k}\right\}\mathrm{d}w,
     \end{align*}
     the form of which allows us to obtain solution of term~\eqref{eq:f2} by simply using that of term~\eqref{eq:f1}. Thus, we have
      \begin{equation*}
        \eqref{eq:f2}=\exp\left\{\frac{5\sigma^2_k-2\sqrt{5}\gamma_k\left(w^{\mathcal{T}}_{ik}-\mu_k\right)}{2\gamma_k^2}\right\}\left[\mathbf{E}^\top_2\boldsymbol{\Lambda}_{21}\Phi\left(\frac{w^{\mathcal{T}}_{ik}-\mu_B}{\sigma_k}\right)+\mathbf{E}^\top_2\boldsymbol{\Lambda}_{22}\frac{\sigma_k}{\sqrt{2\pi}}\exp\left\{-\frac{(w^{\mathcal{T}}_{ik}-\mu_B)^2}{2\sigma^2_k}\right\}\right],
     \end{equation*}
     where
     \begin{equation*}
     \mathbf{E}_2=[E_{20},\,E_{21},\,E_{22}]^\top,\quad \boldsymbol{\Lambda}_{21}=[1,\,-\mu_B,\,\mu^2_B+\sigma^2_k]^\top\quad\mathrm{and}\quad \boldsymbol{\Lambda}_{22}=[0,\,1,\,-\mu_B-w^{\mathcal{T}}_{ik}]^\top
     \end{equation*}
     with
     \begin{equation*}
     E_{20}=1+\frac{\sqrt{5}w^{\mathcal{T}}_{ik}}{\gamma_k}+\frac{5\left(w^{\mathcal{T}}_{ik}\right)^2}{3\gamma_k^2},\quad
     E_{21}=\frac{\sqrt{5}}{\gamma_k}+\frac{10w^{\mathcal{T}}_{ik}}{3\gamma_k^2},\quad
     E_{22}=\frac{5}{3\gamma_k^2},\quad\mathrm{and}\quad       \mu_B=\mu_k+\frac{\sqrt{5}\sigma^2_k}{\gamma_k}.
     \end{equation*}
     Thus, we have
     \begin{align*}
        \xi_{ik}=&\exp\left\{\frac{5\sigma^2_k+2\sqrt{5}\gamma_k\left(w^{\mathcal{T}}_{ik}-\mu_k\right)}{2\gamma_k^2}\right\}\left[\mathbf{E}^\top_1\boldsymbol{\Lambda}_{11}\Phi\left(\frac{\mu_A-w^{\mathcal{T}}_{ik}}{\sigma_k}\right)+\mathbf{E}^\top_1\boldsymbol{\Lambda}_{12}\frac{\sigma_k}{\sqrt{2\pi}}\exp\left\{-\frac{(w^{\mathcal{T}}_{ik}-\mu_A)^2}{2\sigma^2_k}\right\}\right]\\
         &+\exp\left\{\frac{5\sigma^2_k-2\sqrt{5}\gamma_k\left(w^{\mathcal{T}}_{ik}-\mu_k\right)}{2\gamma_k^2}\right\}\left[\mathbf{E}^\top_2\boldsymbol{\Lambda}_{21}\Phi\left(\frac{w^{\mathcal{T}}_{ik}-\mu_B}{\sigma_k}\right)+\mathbf{E}^\top_2\boldsymbol{\Lambda}_{22}\frac{\sigma_k}{\sqrt{2\pi}}\exp\left\{-\frac{(w^{\mathcal{T}}_{ik}-\mu_B)^2}{2\sigma^2_k}\right\}\right].
     \end{align*}

\subsubsection{\texorpdfstring{{Derivation of $\zeta_{ijk}$}}{Derivation of the second expectation}}
\vspace{1ex}
Assume that $w^{\mathcal{T}}_{ik}\leq w^{\mathcal{T}}_{jk}\,$, we have
      \begin{align}
      \label{eq:f4}
      \zeta_{ijk}=&\int^{+\infty}_{w^{\mathcal{T}}_{jk}}\left(1+\frac{\sqrt{5}(w-w^{\mathcal{T}}_{ik})}{\gamma_k}+\frac{5}{3}\left(\frac{w-w^{\mathcal{T}}_{ik}}{\gamma_k}\right)^2\right)\left(1+\frac{\sqrt{5}(w-w^{\mathcal{T}}_{jk})}{\gamma_k}+\frac{5}{3}\left(\frac{w-w^{\mathcal{T}}_{jk}}{\gamma_k}\right)^2\right)\nonumber\\
      &\qquad\times\frac{1}{\sigma_k\sqrt{2\pi}}\exp\left\{-\frac{\sqrt{5}(w-w^{\mathcal{T}}_{ik})+\sqrt{5}(w-w^{\mathcal{T}}_{jk})}{\gamma_k}-\frac{(w-\mu_k)^2}{2\sigma^2_k}\right\}\mathrm{d}w\\
      \label{eq:f5}
      +&\int^{w^{\mathcal{T}}_{jk}}_{w^{\mathcal{T}}_{ik}}\left(1+\frac{\sqrt{5}(w-w^{\mathcal{T}}_{ik})}{\gamma_k}+\frac{5}{3}\left(\frac{w-w^{\mathcal{T}}_{ik}}{\gamma_k}\right)^2\right)\left(1+\frac{\sqrt{5}(w^{\mathcal{T}}_{jk}-w)}{\gamma_k}+\frac{5}{3}\left(\frac{w-w^{\mathcal{T}}_{jk}}{\gamma_k}\right)^2\right)\nonumber\\
      &\qquad\times\frac{1}{\sigma_k\sqrt{2\pi}}\exp\left\{-\frac{\sqrt{5}(w-w^{\mathcal{T}}_{ik})+\sqrt{5}(w^{\mathcal{T}}_{jk}-w)}{\gamma_k}-\frac{(w-\mu_k)^2}{2\sigma^2_k}\right\}\mathrm{d}w\\
      \label{eq:f6}
      +&\int^{w^{\mathcal{T}}_{ik}}_{-\infty}\left(1+\frac{\sqrt{5}(w^{\mathcal{T}}_{ik}-w)}{\gamma_k}+\frac{5}{3}\left(\frac{w-w^{\mathcal{T}}_{ik}}{\gamma_k}\right)^2\right)\left(1+\frac{\sqrt{5}(w^{\mathcal{T}}_{jk}-w)}{\gamma_k}+\frac{5}{3}\left(\frac{w-w^{\mathcal{T}}_{jk}}{\gamma_k}\right)^2\right)\nonumber\\
      &\qquad\times\frac{1}{\sigma_k\sqrt{2\pi}}\exp\left\{-\frac{\sqrt{5}(w^{\mathcal{T}}_{ik}-w)+\sqrt{5}(w^{\mathcal{T}}_{jk}-w)}{\gamma_k}-\frac{(w-\mu_k)^2}{2\sigma^2_k}\right\}\mathrm{d}w.
      \end{align}
      
      We first calculate term~\eqref{eq:f4} by expanding the product of two brackets after the integral sign:    
      \begin{multline*}
       \eqref{eq:f4}=\int^{+\infty}_{w^{\mathcal{T}}_{jk}}(E_{34}w^4+E_{33}w^3+E_{32}w^2+E_{31}w+E_{30})\\
       \times\frac{1}{\sigma_k\sqrt{2\pi}}\exp\left\{-\frac{\sqrt{5}(w-w^{\mathcal{T}}_{ik})+\sqrt{5}(w-w^{\mathcal{T}}_{jk})}{\gamma_k}-\frac{(w-\mu_k)^2}{2\sigma^2_k}\right\}\mathrm{d}w,
      \end{multline*}
      where
      \begin{align*}
        E_{30}=&1+\bigg[25\left(w^{\mathcal{T}}_{ik}\right)^2\left(w^{\mathcal{T}}_{jk}\right)^2-3\sqrt{5}\left(3\gamma_k^3+5\gamma_kw^{\mathcal{T}}_{ik}w^{\mathcal{T}}_{jk}\right)\left(w^{\mathcal{T}}_{ik}+w^{\mathcal{T}}_{jk}\right)\\
        &+15\gamma_k^2\left(\left(w^{\mathcal{T}}_{ik}\right)^2+\left(w^{\mathcal{T}}_{jk}\right)^2+3w^{\mathcal{T}}_{ik}w^{\mathcal{T}}_{jk}\right)\bigg]\bigg/9\gamma_k^4\\
        E_{31}=&\bigg[18\sqrt{5}\gamma_k^3+15\sqrt{5}\gamma_k\left(\left(w^{\mathcal{T}}_{ik}\right)^2+\left(w^{\mathcal{T}}_{jk}\right)^2\right)-75\gamma_k^2\left(w^{\mathcal{T}}_{ik}+w^{\mathcal{T}}_{jk}\right)\\
        &-50w^{\mathcal{T}}_{ik}w^{\mathcal{T}}_{jk}\left(w^{\mathcal{T}}_{ik}+w^{\mathcal{T}}_{jk}\right)+60\sqrt{5}\gamma_kw^{\mathcal{T}}_{ik}w^{\mathcal{T}}_{jk}\bigg]\bigg/9\gamma_k^4\\
        E_{32}=&5\bigg[5\left(w^{\mathcal{T}}_{ik}\right)^2+5\left(w^{\mathcal{T}}_{jk}\right)^2+15\gamma_k^2-9\sqrt{5}\gamma_k\left(w^{\mathcal{T}}_{ik}+w^{\mathcal{T}}_{jk}\right)+20w^{\mathcal{T}}_{ik}w^{\mathcal{T}}_{jk}\bigg]\bigg/9\gamma_k^4\\
        E_{33}=&\dfrac{10\left(3\sqrt{5}\gamma_k-5w^{\mathcal{T}}_{ik}-5w^{\mathcal{T}}_{jk}\right)}{9\gamma_k^4}\quad\mathrm{and}\quad E_{34}=\dfrac{25}{9\gamma_k^4}.
      \end{align*}
      
      Then by completing the square, we have
      \begin{multline*}
        \eqref{eq:f4}=\exp\left\{\frac{10\sigma_k^2+\sqrt{5}\gamma_k\left(w^{\mathcal{T}}_{ik}+w^{\mathcal{T}}_{jk}-2\mu_k\right)}{\gamma_k^2}\right\}\\
        \times\int^{+\infty}_{w^{\mathcal{T}}_{jk}}(E_{34}w^4+E_{33}w^3+E_{32}w^2+E_{31}w+E_{30})\frac{1}{\sigma_k\sqrt{2\pi}}\exp\left\{-\frac{(w-\mu_{C})^2}{2\sigma^2_k}\right\}\mathrm{d}w,
      \end{multline*}
      where
      \begin{equation*}
        \mu_{C}=\mu_k-2\sqrt{5}\frac{\sigma^2_k}{\gamma_k}. 
      \end{equation*}
      
      Using Lemma~\ref{lemma:partial} and arranging terms, we obtain
      \begin{multline*}
        \eqref{eq:f4}=\exp\left\{\frac{10\sigma_k^2+\sqrt{5}\gamma_k\left(w^{\mathcal{T}}_{ik}+w^{\mathcal{T}}_{jk}-2\mu_k\right)}{\gamma_k^2}\right\}\\
        \qquad\qquad\times\left[\mathbf{E}^\top_3\boldsymbol{\Lambda}_{31}\Phi\left(\frac{\mu_C-w^{\mathcal{T}}_{jk}}{\sigma_k}\right)+\mathbf{E}^\top_3\boldsymbol{\Lambda}_{32}\frac{\sigma_k}{\sqrt{2\pi}}\exp\left\{-\frac{(w^{\mathcal{T}}_{jk}-\mu_{C})^2}{2\sigma^2_k}\right\}\right],
      \end{multline*}
      where
      \begin{itemize}
          \item $\mathbf{E}_3=[E_{30},\,E_{31},\,E_{32},\,E_{33},\,E_{34}]^\top\,$;
          \item $\boldsymbol{\Lambda}_{31}=[1,\,\mu_C,\,\mu_C^2+\sigma^2_k,\,\mu_C^3+3\sigma^2_k\mu_C,\,\mu_C^4+6\sigma^2_k\mu_C^2+3\sigma_k^4]^\top\,$;
          \item $\boldsymbol{\Lambda}_{32}=[0,\,1,\,\mu_C+w^{\mathcal{T}}_{jk},\,\mu_C^2+2\sigma^2_k+\left(w^{\mathcal{T}}_{jk}\right)^2+\mu_Cw^{\mathcal{T}}_{jk},\,\mu_C^3+\left(w^{\mathcal{T}}_{jk}\right)^3+w^{\mathcal{T}}_{jk}\mu_C^2+\mu_C\left(w^{\mathcal{T}}_{jk}\right)^2+3\sigma_k^2w^{\mathcal{T}}_{jk}+5\sigma_k^2\mu_C]^\top\,$.
      \end{itemize}
      
      The derivation of term~\eqref{eq:f5} is analogue to that of term~\eqref{eq:f4}. By expanding the product of two brackets after the integral sign, we have
      \begin{multline*}
       \eqref{eq:f5}=\int^{w^{\mathcal{T}}_{jk}}_{w^{\mathcal{T}}_{ik}}(E_{44}w^4+E_{43}w^3+E_{42}w^2+E_{41}w+E_{40})\\
       \times\frac{1}{\sigma_k\sqrt{2\pi}}\exp\left\{-\frac{\sqrt{5}(w-w^{\mathcal{T}}_{ik})+\sqrt{5}(w^{\mathcal{T}}_{jk}-w)}{\gamma_k}-\frac{(w-\mu_k)^2}{2\sigma^2_k}\right\}\mathrm{d}w,
      \end{multline*}
      where
      \begin{align*}
        E_{40}=&1+\bigg[25\left(w^{\mathcal{T}}_{ik}\right)^2\left(w^{\mathcal{T}}_{jk}\right)^2+3\sqrt{5}\left(3\gamma_k^3-5\gamma_kw^{\mathcal{T}}_{ik}w^{\mathcal{T}}_{jk}\right)\left(w^{\mathcal{T}}_{jk}-w^{\mathcal{T}}_{ik}\right)\\
        &+15\gamma_k^2\left(\left(w^{\mathcal{T}}_{ik}\right)^2+\left(w^{\mathcal{T}}_{jk}\right)^2-3w^{\mathcal{T}}_{ik}w^{\mathcal{T}}_{jk}\right)\bigg]\bigg/9\gamma_k^4\\
        E_{41}=&5\bigg[3\sqrt{5}\gamma_k\left(\left(w^{\mathcal{T}}_{jk}\right)^2-\left(w^{\mathcal{T}}_{ik}\right)^2\right)+3\gamma_k^2\left(w^{\mathcal{T}}_{ik}+w^{\mathcal{T}}_{jk}\right)-10w^{\mathcal{T}}_{ik}w^{\mathcal{T}}_{jk}\left(w^{\mathcal{T}}_{ik}+w^{\mathcal{T}}_{jk}\right)\bigg]\bigg/9\gamma_k^4\\
        E_{42}=&5\bigg[5\left(w^{\mathcal{T}}_{ik}\right)^2+5\left(w^{\mathcal{T}}_{jk}\right)^2-3\gamma_k^2-3\sqrt{5}\gamma_k\left(w^{\mathcal{T}}_{jk}-w^{\mathcal{T}}_{ik}\right)+20w^{\mathcal{T}}_{ik}w^{\mathcal{T}}_{jk}\bigg]\bigg/9\gamma_k^4\\
        E_{43}=&-\dfrac{50\left(w^{\mathcal{T}}_{ik}+w^{\mathcal{T}}_{jk}\right)}{9\gamma_k^4}\quad\mathrm{and}\quad E_{44}=\dfrac{25}{9\gamma_k^4}.
      \end{align*}
      
      Then by completing the square, we have
      \begin{multline*}
        \eqref{eq:f5}=\exp\left\{-\frac{\sqrt{5}\left(w^{\mathcal{T}}_{jk}-w^{\mathcal{T}}_{ik}\right)}{\gamma_k}\right\}\\
        \times\int^{w^{\mathcal{T}}_{jk}}_{w^{\mathcal{T}}_{ik}}(E_{44}w^4+E_{43}w^3+E_{42}w^2+E_{41}w+E_{40})\frac{1}{\sigma_k\sqrt{2\pi}}\exp\left\{-\frac{(w-\mu_k)^2}{2\sigma^2_k}\right\}\mathrm{d}w.
      \end{multline*}
      
      Using Lemma~\ref{lemma:partial} and arranging terms, we obtain 
      \begin{multline*}
        \eqref{eq:f5}=\exp\left\{-\frac{\sqrt{5}\left(w^{\mathcal{T}}_{jk}-w^{\mathcal{T}}_{ik}\right)}{\gamma_k}\right\}\Bigg[\mathbf{E}^\top_4\boldsymbol{\Lambda}_{41}\left[\Phi\left(\frac{w^{\mathcal{T}}_{jk}-\mu_k}{\sigma_k}\right)-\Phi\left(\frac{w^{\mathcal{T}}_{ik}-\mu_k}{\sigma_k}\right)\right]\\
        +\mathbf{E}^\top_4\boldsymbol{\Lambda}_{42}\frac{\sigma_k}{\sqrt{2\pi}}\exp\left\{-\frac{(w^{\mathcal{T}}_{ik}-\mu_k)^2}{2\sigma^2_k}\right\}-\mathbf{E}^\top_4\boldsymbol{\Lambda}_{43}\frac{\sigma_k}{\sqrt{2\pi}}\exp\left\{-\frac{(w^{\mathcal{T}}_{jk}-\mu_k)^2}{2\sigma^2_k}\right\}\Bigg],
      \end{multline*}
      where
      \begin{itemize}
          \item $\mathbf{E}_4=[E_{40},\,E_{41},\,E_{42},\,E_{43},\,E_{44}]^\top\,$;
          \item $\boldsymbol{\Lambda}_{41}=[1,\,\mu_k,\,\mu_k^2+\sigma^2_k,\,\mu_k^3+3\sigma^2_k\mu_k,\,\mu_k^4+6\sigma^2_k\mu_k^2+3\sigma_k^4]^\top\,$;
          \item $\boldsymbol{\Lambda}_{42}=[0,\,1,\,\mu_k+w^{\mathcal{T}}_{ik},\,\mu_k^2+2\sigma^2_k+\left(w^{\mathcal{T}}_{ik}\right)^2+\mu_kw^{\mathcal{T}}_{ik},\,\mu_k^3+\left(w^{\mathcal{T}}_{ik}\right)^3+w^{\mathcal{T}}_{ik}\mu_k^2+\mu_k\left(w^{\mathcal{T}}_{ik}\right)^2+3\sigma_k^2w^{\mathcal{T}}_{ik}+5\sigma_k^2\mu_k]^\top\,$;
          \item $\boldsymbol{\Lambda}_{43}=[0,\,1,\,\mu_k+w^{\mathcal{T}}_{jk},\,\mu_k^2+2\sigma^2_k+\left(w^{\mathcal{T}}_{jk}\right)^2+\mu_kw^{\mathcal{T}}_{jk},\,\mu_k^3+\left(w^{\mathcal{T}}_{jk}\right)^3+w^{\mathcal{T}}_{jk}\mu_k^2+\mu_k\left(w^{\mathcal{T}}_{jk}\right)^2+3\sigma_k^2w^{\mathcal{T}}_{jk}+5\sigma_k^2\mu_k]^\top\,$.
      \end{itemize}
      
      Term~\eqref{eq:f6} can be computed in the following way:
      \begin{align*}
          \eqref{eq:f6}=&\int^{w^{\mathcal{T}}_{ik}}_{-\infty}\left(1+\frac{\sqrt{5}(w^{\mathcal{T}}_{ik}-w)}{\gamma_k}+\frac{5}{3}\left(\frac{w-w^{\mathcal{T}}_{ik}}{\gamma_k}\right)^2\right)\left(1+\frac{\sqrt{5}(w^{\mathcal{T}}_{jk}-w)}{\gamma_k}+\frac{5}{3}\left(\frac{w-w^{\mathcal{T}}_{jk}}{\gamma_k}\right)^2\right)\\
      &\qquad\times\frac{1}{\sigma_k\sqrt{2\pi}}\exp\left\{-\frac{\sqrt{5}(w^{\mathcal{T}}_{ik}-w)+\sqrt{5}(w^{\mathcal{T}}_{jk}-w)}{\gamma_k}-\frac{(w-\mu_k)^2}{2\sigma^2_k}\right\}\mathrm{d}w\\
      =&\int^{+\infty}_{-w^{\mathcal{T}}_{ik}}\left(1+\frac{\sqrt{5}(w+w^{\mathcal{T}}_{ik})}{\gamma_k}+\frac{5}{3}\left(\frac{w+w^{\mathcal{T}}_{ik}}{\gamma_k}\right)^2\right)\left(1+\frac{\sqrt{5}(w+w^{\mathcal{T}}_{jk})}{\gamma_k}+\frac{5}{3}\left(\frac{w+w^{\mathcal{T}}_{jk}}{\gamma_k}\right)^2\right)\\
      &\qquad\times\frac{1}{\sigma_k\sqrt{2\pi}}\exp\left\{-\frac{\sqrt{5}(w+w^{\mathcal{T}}_{ik})+\sqrt{5}(w+w^{\mathcal{T}}_{jk})}{\gamma_k}-\frac{(w+\mu_k)^2}{2\sigma^2_k}\right\}\mathrm{d}w,
      \end{align*}
      the form of which allows us to obtain solution of term~\eqref{eq:f6} by simply using that of term~\eqref{eq:f4}. Thus, we have
      \begin{multline*}
        \eqref{eq:f6}=\exp\left\{\frac{10\sigma_k^2-\sqrt{5}\gamma_k\left(w^{\mathcal{T}}_{ik}+w^{\mathcal{T}}_{jk}-2\mu_k\right)}{\gamma_k^2}\right\}\\
        \times\left[\mathbf{E}^\top_5\boldsymbol{\Lambda}_{51}\Phi\left(\frac{w^{\mathcal{T}}_{ik}-\mu_D}{\sigma_k}\right)+\mathbf{E}^\top_5\boldsymbol{\Lambda}_{52}\frac{\sigma_k}{\sqrt{2\pi}}\exp\left\{-\frac{(w^{\mathcal{T}}_{ik}-\mu_{D})^2}{2\sigma^2_k}\right\}\right],
      \end{multline*}
      where
      \begin{itemize}
          \item $\mathbf{E}_5=[E_{50},\,E_{51},\,E_{52},\,E_{53},\,E_{54}]^\top\,$;
          \item $\boldsymbol{\Lambda}_{51}=[1,\,-\mu_D,\,\mu_D^2+\sigma^2_k,\,-\mu_D^3-3\sigma^2_k\mu_D,\,\mu_D^4+6\sigma^2_k\mu_D^2+3\sigma_k^4]^\top\,$;
          \item $\boldsymbol{\Lambda}_{52}=[0,\,1,\,-\mu_D-w^{\mathcal{T}}_{ik},\,\mu_D^2+2\sigma^2_k+\left(w^{\mathcal{T}}_{ik}\right)^2+\mu_Dw^{\mathcal{T}}_{ik},\,-\mu_D^3-\left(w^{\mathcal{T}}_{ik}\right)^3-w^{\mathcal{T}}_{ik}\mu_D^2-\mu_D\left(w^{\mathcal{T}}_{ik}\right)^2-3\sigma_k^2w^{\mathcal{T}}_{ik}-5\sigma_k^2\mu_D]^\top$
      \end{itemize}
      with
      \begin{align*}
        E_{50}=&1+\bigg[25\left(w^{\mathcal{T}}_{ik}\right)^2\left(w^{\mathcal{T}}_{jk}\right)^2+3\sqrt{5}\left(3\gamma_k^3+5\gamma_kw^{\mathcal{T}}_{ik}w^{\mathcal{T}}_{jk}\right)\left(w^{\mathcal{T}}_{ik}+w^{\mathcal{T}}_{jk}\right)\\
        &+15\gamma_k^2\left(\left(w^{\mathcal{T}}_{ik}\right)^2+\left(w^{\mathcal{T}}_{jk}\right)^2+3w^{\mathcal{T}}_{ik}w^{\mathcal{T}}_{jk}\right)\bigg]\bigg/9\gamma_k^4\\
        E_{51}=&\bigg[18\sqrt{5}\gamma_k^3+15\sqrt{5}\gamma_k\left(\left(w^{\mathcal{T}}_{ik}\right)^2+\left(w^{\mathcal{T}}_{jk}\right)^2\right)+75\gamma_k^2\left(w^{\mathcal{T}}_{ik}+w^{\mathcal{T}}_{jk}\right)\\
        &+50w^{\mathcal{T}}_{ik}w^{\mathcal{T}}_{jk}\left(w^{\mathcal{T}}_{ik}+w^{\mathcal{T}}_{jk}\right)+60\sqrt{5}\gamma_kw^{\mathcal{T}}_{ik}w^{\mathcal{T}}_{jk}\bigg]\bigg/9\gamma_k^4\\
        E_{52}=&5\bigg[5\left(w^{\mathcal{T}}_{ik}\right)^2+5\left(w^{\mathcal{T}}_{jk}\right)^2+15\gamma_k^2+9\sqrt{5}\gamma_k\left(w^{\mathcal{T}}_{ik}+w^{\mathcal{T}}_{jk}\right)+20w^{\mathcal{T}}_{ik}w^{\mathcal{T}}_{jk}\bigg]\bigg/9\gamma_k^4\\
        E_{53}=&\dfrac{10\left(3\sqrt{5}\gamma_k+5w^{\mathcal{T}}_{ik}+5w^{\mathcal{T}}_{jk}\right)}{9\gamma_k^4},\quad E_{54}=\dfrac{25}{9\gamma_k^4}\quad\mathrm{and}\quad \mu_D=\mu_k+2\sqrt{5}\frac{\sigma^2_k}{\gamma_k}.	
      \end{align*}
      
      Therefore, the expression for $\zeta_{ijk}$ when $w^{\mathcal{T}}_{ik}\leq w^{\mathcal{T}}_{jk}$ is given by
      \begin{align}
      \label{eq:jd}
      \zeta_{ijk}=&\exp\left\{\frac{10\sigma_k^2+\sqrt{5}\gamma_k\left(w^{\mathcal{T}}_{ik}+w^{\mathcal{T}}_{jk}-2\mu_k\right)}{\gamma_k^2}\right\}\nonumber\\
        &\quad\times\left[\mathbf{E}^\top_3\boldsymbol{\Lambda}_{31}\Phi\left(\frac{\mu_C-w^{\mathcal{T}}_{jk}}{\sigma_k}\right)+\mathbf{E}^\top_3\boldsymbol{\Lambda}_{32}\frac{\sigma_k}{\sqrt{2\pi}}\exp\left\{-\frac{(w^{\mathcal{T}}_{jk}-\mu_{C})^2}{2\sigma^2_k}\right\}\right]\nonumber\\
        &+\exp\left\{-\frac{\sqrt{5}\left(w^{\mathcal{T}}_{jk}-w^{\mathcal{T}}_{ik}\right)}{\gamma_k}\right\}\Bigg[\mathbf{E}^\top_4\boldsymbol{\Lambda}_{41}\left(\Phi\left(\frac{w^{\mathcal{T}}_{jk}-\mu_k}{\sigma_k}\right)-\Phi\left(\frac{w^{\mathcal{T}}_{ik}-\mu_k}{\sigma_k}\right)\right)\nonumber\\
        &\quad+\mathbf{E}^\top_4\boldsymbol{\Lambda}_{42}\frac{\sigma_k}{\sqrt{2\pi}}\exp\left\{-\frac{(w^{\mathcal{T}}_{ik}-\mu_k)^2}{2\sigma^2_k}\right\}-\mathbf{E}^\top_4\boldsymbol{\Lambda}_{43}\frac{\sigma_k}{\sqrt{2\pi}}\exp\left\{-\frac{(w^{\mathcal{T}}_{jk}-\mu_k)^2}{2\sigma^2_k}\right\}\Bigg]\nonumber\\
        &+\exp\left\{\frac{10\sigma_k^2-\sqrt{5}\gamma_k\left(w^{\mathcal{T}}_{ik}+w^{\mathcal{T}}_{jk}-2\mu_k\right)}{\gamma_k^2}\right\}\nonumber\\
        &\quad\times\left[\mathbf{E}^\top_5\boldsymbol{\Lambda}_{51}\Phi\left(\frac{w^{\mathcal{T}}_{ik}-\mu_D}{\sigma_k}\right)+\mathbf{E}^\top_5\boldsymbol{\Lambda}_{52}\frac{\sigma_k}{\sqrt{2\pi}}\exp\left\{-\frac{(w^{\mathcal{T}}_{ik}-\mu_{D})^2}{2\sigma^2_k}\right\}\right],\nonumber
      \end{align}
      and interchanging positions of $w^{\mathcal{T}}_{ik}$ and $w^{\mathcal{T}}_{jk}$ gives the expression for $\zeta_{ijk}$ when $w^{\mathcal{T}}_{ik}> w^{\mathcal{T}}_{jk}\,$.

\subsubsection{\texorpdfstring{{Derivation of $\psi_{jk}$}}{Derivation of the third expectation}}
\begin{align}
     \psi_{jk}=&\int w\left(1+\frac{\sqrt{5}|w-w^{\mathcal{T}}_{jk}|}{\gamma_k}+\frac{5}{3}\left(\frac{w-w^{\mathcal{T}}_{jk}}{\gamma_k}\right)^2\right)\frac{1}{\sigma_k\sqrt{2\pi}}\exp\left\{-\frac{\sqrt{5}|w-w^{\mathcal{T}}_{jk}|}{\gamma_k}-\frac{(w-\mu_k)^2}{2\sigma^2_k}\right\}\mathrm{d}w\nonumber\\
     \label{eq:f8}
     =&\int^{+\infty}_{w^{\mathcal{T}}_{jk}}\left(w+\frac{\sqrt{5}w\left(w-w^{\mathcal{T}}_{jk}\right)}{\gamma_k}+\frac{5w}{3}\left(\frac{w-w^{\mathcal{T}}_{jk}}{\gamma_k}\right)^2\right)\frac{1}{\sigma_k\sqrt{2\pi}}\exp\left\{-\frac{\sqrt{5}\left(w-w^{\mathcal{T}}_{jk}\right)}{\gamma_k}-\frac{(w-\mu_k)^2}{2\sigma^2_k}\right\}\mathrm{d}w\\
     \label{eq:f9}
     &+\int^{w^{\mathcal{T}}_{jk}}_{-\infty}\left(w+\frac{\sqrt{5}w\left(w^{\mathcal{T}}_{jk}-w\right)}{\gamma_k}+\frac{5w}{3}\left(\frac{w-w^{\mathcal{T}}_{jk}}{\gamma_k}\right)^2\right)\frac{1}{\sigma_k\sqrt{2\pi}}\exp\left\{\frac{\sqrt{5}\left(w-w^{\mathcal{T}}_{jk}\right)}{\gamma_k}-\frac{(w-\mu_k)^2}{2\sigma^2_k}\right\}\mathrm{d}w.
     \end{align}
     
     We first calculate term~\eqref{eq:f8} by arranging the terms in the bracket after the integral sign and completing the square:
     \begin{equation*}
         \eqref{eq:f8}=\exp\left\{\frac{5\sigma^2_k+2\sqrt{5}\gamma_k\left(w^{\mathcal{T}}_{jk}-\mu_k\right)}{2\gamma_k^2}\right\}\int^{+\infty}_{w^{\mathcal{T}}_{jk}}\left[E_{12}w^3+E_{11}w^2+E_{10}w\right]\frac{1}{\sigma_k\sqrt{2\pi}}\exp\left\{-\frac{(w-\mu_A)^2}{2\sigma^2_k}\right\}.
     \end{equation*}
     By Lemma~\ref{lemma:partial}, we then obtain
     \begin{equation*}
        \eqref{eq:f8}=\exp\left\{\frac{5\sigma^2_k+2\sqrt{5}\gamma_k\left(w^{\mathcal{T}}_{jk}-\mu_k\right)}{2\gamma_k^2}\right\}\left[\mathbf{E}^\top_1\boldsymbol{\Lambda}_{61}\Phi\left(\frac{\mu_A-w^{\mathcal{T}}_{jk}}{\sigma_k}\right)+\mathbf{E}^\top_1\boldsymbol{\Lambda}_{62}\frac{\sigma_k}{\sqrt{2\pi}}\exp\left\{-\frac{(w^{\mathcal{T}}_{jk}-\mu_A)^2}{2\sigma^2_k}\right\}\right],
     \end{equation*}
     where
     \begin{itemize}
     \item $\boldsymbol{\Lambda}_{61}=\left[\mu_A,\,\mu^2_A+\sigma^2_k,\,\mu^3_A+3\sigma_k^2\mu_A\right]^\top$;
     \item $\boldsymbol{\Lambda}_{62}=\left[1,\,\mu_A+w^{\mathcal{T}}_{jk},\,\mu^2_A+2\sigma^2_k+\left(w^{\mathcal{T}}_{jk}\right)^2+\mu_Aw^{\mathcal{T}}_{jk}\right]^\top$.
     \end{itemize}
Term~\eqref{eq:f9} can be rewritten as follow:
     \begin{align*}
         \eqref{eq:f9}=&\int^{w^{\mathcal{T}}_{jk}}_{-\infty}\left(1+\frac{\sqrt{5}\left(w^{\mathcal{T}}_{jk}-w\right)}{\gamma_k}+\frac{5}{3}\left(\frac{w-w^{\mathcal{T}}_{jk}}{\gamma_k}\right)^2\right)\frac{w}{\sigma_k\sqrt{2\pi}}\exp\left\{\frac{\sqrt{5}\left(w-w^{\mathcal{T}}_{jk}\right)}{\gamma_k}-\frac{(w-\mu_k)^2}{2\sigma^2_k}\right\}\mathrm{d}w\\
         =&-\int^{+\infty}_{-w^{\mathcal{T}}_{jk}}\left(1+\frac{\sqrt{5}\left(w+w^{\mathcal{T}}_{jk}\right)}{\gamma_k}+\frac{5}{3}\left(\frac{w+w^{\mathcal{T}}_{jk}}{\gamma_k}\right)^2\right)\frac{w}{\sigma_k\sqrt{2\pi}}\exp\left\{-\frac{\sqrt{5}\left(w+w^{\mathcal{T}}_{jk}\right)}{\gamma_k}-\frac{(w+\mu_k)^2}{2\sigma^2_k}\right\}\mathrm{d}w,
     \end{align*}
     the form of which allows us to obtain solution of term~\eqref{eq:f9} by using that of term~\eqref{eq:f8}. Thus, we have
      \begin{multline*}
        \eqref{eq:f9}=-\exp\left\{\frac{5\sigma^2_k-2\sqrt{5}\gamma_k\left(w^{\mathcal{T}}_{jk}-\mu_k\right)}{2\gamma_k^2}\right\}\\
        \times\left[\mathbf{E}^\top_2\boldsymbol{\Lambda}_{71}\Phi\left(\frac{w^{\mathcal{T}}_{jk}-\mu_B}{\sigma_k}\right)+\mathbf{E}^\top_2\boldsymbol{\Lambda}_{72}\frac{\sigma_k}{\sqrt{2\pi}}\exp\left\{-\frac{(w^{\mathcal{T}}_{jk}-\mu_B)^2}{2\sigma^2_k}\right\}\right],
     \end{multline*}
     where
     \begin{itemize}
     \item $\boldsymbol{\Lambda}_{71}=\left[-\mu_B,\,\mu^2_B+\sigma^2_k,\,-\mu^3_B-3\sigma_k^2\mu_B\right]^\top$;
     \item $\boldsymbol{\Lambda}_{72}=\left[1,\,-\mu_B-w^{\mathcal{T}}_{jk},\,\mu^2_B+2\sigma^2_k+\left(w^{\mathcal{T}}_{jk}\right)^2+\mu_Bw^{\mathcal{T}}_{jk}\right]^\top$.
     \end{itemize}
Thus, we have
     \begin{align*}
         \psi_{jk}=&\exp\left\{\frac{5\sigma^2_k+2\sqrt{5}\gamma_k\left(w^{\mathcal{T}}_{jk}-\mu_k\right)}{2\gamma_k^2}\right\}\left[\mathbf{E}^\top_1\boldsymbol{\Lambda}_{61}\Phi\left(\frac{\mu_A-w^{\mathcal{T}}_{jk}}{\sigma_k}\right)+\mathbf{E}^\top_1\boldsymbol{\Lambda}_{62}\frac{\sigma_k}{\sqrt{2\pi}}\exp\left\{-\frac{(w^{\mathcal{T}}_{jk}-\mu_A)^2}{2\sigma^2_k}\right\}\right]\\
         &-\exp\left\{\frac{5\sigma^2_k-2\sqrt{5}\gamma_k\left(w^{\mathcal{T}}_{jk}-\mu_k\right)}{2\gamma_k^2}\right\}\left[\mathbf{E}^\top_2\boldsymbol{\Lambda}_{71}\Phi\left(\frac{w^{\mathcal{T}}_{jk}-\mu_B}{\sigma_k}\right)+\mathbf{E}^\top_2\boldsymbol{\Lambda}_{72}\frac{\sigma_k}{\sqrt{2\pi}}\exp\left\{-\frac{(w^{\mathcal{T}}_{jk}-\mu_B)^2}{2\sigma^2_k}\right\}\right].
     \end{align*}

\section[Proof of Proposition~5.1]{Proof of Proposition~\ref{prop:variance}}
\label{sec:proofvariance}
Replace $\mu_g(\mathbf{W},\mathbf{z})$ by equation~\eqref{eq:krigingmean} with Assumption~\ref{ass:1}, we have
\begin{align*}
\mathbb{E}_{W_{k\in\mathbb{S}^\mathsf{c}}}\left[\mu_g(\mathbf{W},\mathbf{z})\right]=&\mathbb{E}_{W_{k\in\mathbb{S}^\mathsf{c}}}\left[\mathbf{W}^\top\right]\widehat{\boldsymbol{\theta}}+\mathbf{h}(\mathbf{z})^\top\widehat{\boldsymbol{\beta}}+\mathbb{E}_{W_{k\in\mathbb{S}^\mathsf{c}}}\left[\mathbf{r}^\top(\mathbf{W},\,\mathbf{z})\right]\mathbf{R}^{-1}\left(\mathbf{y}^{\mathcal{T}}-\mathbf{w}^{\mathcal{T}}\widehat{\boldsymbol{\theta}}-\mathbf{H}(\mathbf{z}^{\mathcal{T}})\widehat{\boldsymbol{\beta}}\right)\nonumber\\
 =&\widetilde{\boldsymbol{\mu}}^\top\widehat{\boldsymbol{\theta}}+\mathbf{h}(\mathbf{z})^\top\widehat{\boldsymbol{\beta}}+\widetilde{\mathbf{I}}^\top\mathbf{A},
\end{align*}
where
\begin{itemize}
\item $\widetilde{\boldsymbol{\mu}}=\mathbb{E}_{W_{k\in\mathbb{S}^\mathsf{c}}}\left[\mathbf{W}^\top\right]\in\mathbb{R}^{d\times 1}$ is a column vector with its $k$-th element:
\begin{equation*}
\widetilde{\mu}_k=\begin{dcases*}
W_k, & $k\in\mathbb{S}$,	\\
\mu_k, & $k\in\mathbb{S}^\mathsf{c}$;
\end{dcases*}
\end{equation*}
\item $\widetilde{\mathbf{I}}=\mathbb{E}_{W_{k\in\mathbb{S}^\mathsf{c}}}\left[\mathbf{r}^\top(\mathbf{W},\,\mathbf{z})\right]\in\mathbb{R}^{m\times 1}$ with its $i$-th element:
    \begin{align*}
        \widetilde{I}_i=&\mathbb{E}_{W_{k\in\mathbb{S}^\mathsf{c}}}\left[c(\mathbf{W},\,\mathbf{w}^{\mathcal{T}}_i)c(\mathbf{z},\,\mathbf{z}^{\mathcal{T}}_i)\right]\\
        =&\mathbb{E}_{W_{k\in\mathbb{S}^\mathsf{c}}}\left[c(\mathbf{W},\,\mathbf{w}^{\mathcal{T}}_i)\right]c(\mathbf{z},\,\mathbf{z}^{\mathcal{T}}_i)\\
        =&\mathbb{E}_{W_{k\in\mathbb{S}^\mathsf{c}}}\left[\prod_{k=1}^dc_k(W_k,\,w^{\mathcal{T}}_{ik})\right]\prod_{k=1}^pc_k(z_k,\,z^{\mathcal{T}}_{ik})\\
        =&\prod_{k\in\mathbb{S}}c_k(W_k,\,w^{\mathcal{T}}_{ik})\prod_{k\in\mathbb{S}^\mathsf{c}}\mathbb{E}_{W_k}\left[c_k(W_k,\,w^{\mathcal{T}}_{ik})\right]\prod_{k=1}^pc_k(z_k,\,z^{\mathcal{T}}_{ik})\\
        =&\prod_{k\in\mathbb{S}}c_k(W_k,\,w^{\mathcal{T}}_{ik})\prod_{k\in\mathbb{S}^\mathsf{c}}\xi_{ik} \prod_{k=1}^pc_k(z_k,\,z^{\mathcal{T}}_{ik}).
    \end{align*}
    \end{itemize}
    
Then, we have
\begin{align}
\label{eq:v1proof}
	V_1(\mathbb{S})=&\mathrm{Var}_{W_{k\in\mathbb{S}}}\left(\widetilde{\boldsymbol{\mu}}^\top\widehat{\boldsymbol{\theta}}+\mathbf{h}(\mathbf{z})^\top\widehat{\boldsymbol{\beta}}+\widetilde{\mathbf{I}}^\top\mathbf{A}\right)\nonumber\\
	=&\mathrm{Var}_{W_{k\in\mathbb{S}}}\left(\widetilde{\boldsymbol{\mu}}^\top\widehat{\boldsymbol{\theta}}+\widetilde{\mathbf{I}}^\top\mathbf{A}\right)\nonumber\\
	=&\underbrace{\mathbb{E}_{W_{k\in\mathbb{S}}}\left[\left(\widetilde{\boldsymbol{\mu}}^\top\widehat{\boldsymbol{\theta}}+\widetilde{\mathbf{I}}^\top\mathbf{A}\right)^2\right]}_{(\mathrm{\ref{eq:v1proof}.1})}-\underbrace{\left(\mathbb{E}_{W_{k\in\mathbb{S}}}\left[\widetilde{\boldsymbol{\mu}}^\top\widehat{\boldsymbol{\theta}}+\widetilde{\mathbf{I}}^\top\mathbf{A}\right]\right)^2}_{(\mathrm{\ref{eq:v1proof}.2})}.
\end{align}	

We first derive~$(\mathrm{\ref{eq:v1proof}.1})$ as follow:
\begin{align}
\label{eq:expectation1exp}
(\mathrm{\ref{eq:v1proof}.1})=&\mathbb{E}_{W_{k\in\mathbb{S}}}\left[\widetilde{\boldsymbol{\mu}}^\top\widehat{\boldsymbol{\theta}}\widehat{\boldsymbol{\theta}}^\top\widetilde{\boldsymbol{\mu}}+\widetilde{\mathbf{I}}^\top\mathbf{A}\mathbf{A}^\top\widetilde{\mathbf{I}}+2\widehat{\boldsymbol{\theta}}^\top\widetilde{\boldsymbol{\mu}}\widetilde{\mathbf{I}}^\top\mathbf{A}\right]	\nonumber\\
=&\mathrm{tr}\left\{\widehat{\boldsymbol{\theta}}\widehat{\boldsymbol{\theta}}^\top\left(\boldsymbol{\mu}\boldsymbol{\mu}^\top+\widetilde{\boldsymbol{\Omega}}\right)\right\}+\mathrm{tr}\left\{\mathbf{A}\mathbf{A}^\top\mathbb{E}_{W_{k\in\mathbb{S}}}\left[\widetilde{\mathbf{I}}\widetilde{\mathbf{I}}^\top\right]\right\}+2\widehat{\boldsymbol{\theta}}^\top\mathbb{E}_{W_{k\in\mathbb{S}}}\left[\widetilde{\boldsymbol{\mu}}\widetilde{\mathbf{I}}^\top\right]\mathbf{A}\nonumber\\
=&\mathrm{tr}\left\{\widehat{\boldsymbol{\theta}}\widehat{\boldsymbol{\theta}}^\top\left(\boldsymbol{\mu}\boldsymbol{\mu}^\top+\widetilde{\boldsymbol{\Omega}}\right)\right\}+\mathrm{tr}\left\{\mathbf{A}\mathbf{A}^\top\widetilde{\mathbf{J}}\right\}+2\widehat{\boldsymbol{\theta}}^\top\widetilde{\mathbf{B}}\mathbf{A},
\end{align}
where the second step uses the derivations analogous to those used for equations~\eqref{eq:expectation1} and~\eqref{eq:expectation2}, and 
\begin{itemize}
\item $\widetilde{\boldsymbol{\Omega}}=\mathrm{Var}_{W_{k\in\mathbb{S}}}\left(\widetilde{\boldsymbol{\mu}}\right)\in\mathbb{R}^{d\times d}$ being a diagonal matrix with its $k$-th diagonal element given by 
\begin{equation*}
\widetilde{\boldsymbol{\Omega}}_k=\sigma_k^2(\mathbf{x}_k)\mathbbm{1}_{\{k\in\mathbb{S}\}};
\end{equation*}
\item $\widetilde{\mathbf{B}}=\mathbb{E}_{W_{k\in\mathbb{S}}}\left[\widetilde{\boldsymbol{\mu}}\widetilde{\mathbf{I}}^\top\right]\in\mathbb{R}^{d\times m}$ with its $lj$-th element:
\begin{align*}
\widetilde{B}_{lj}=&\mathbb{E}_{W_{k\in\mathbb{S}}}\left[\widetilde{\mu}_l\prod_{k\in\mathbb{S}}c_k(W_k,\,w^{\mathcal{T}}_{jk})\prod_{k\in\mathbb{S}^\mathsf{c}}\xi_{jk} \prod_{k=1}^pc_k(z_k,\,z^{\mathcal{T}}_{jk})\right]\\
=&\mathbb{E}_{W_{k\in\mathbb{S}}}\left[\widetilde{\mu}_l\prod_{k\in\mathbb{S}}c_k(W_k,\,w^{\mathcal{T}}_{jk})\right]\prod_{k\in\mathbb{S}^\mathsf{c}}\xi_{jk} \prod_{k=1}^pc_k(z_k,\,z^{\mathcal{T}}_{jk})\\
=&\begin{dcases}
\mathbb{E}_{W_{k\in\mathbb{S}}}\left[W_lc_l(W_l,\,w^{\mathcal{T}}_{jl})\prod_{\substack{k\in\mathbb{S}\\k\neq l}}c_k(W_k,\,w^{\mathcal{T}}_{jk})\right]\prod_{k\in\mathbb{S}^\mathsf{c}}\xi_{jk}\prod_{k=1}^pc_k(z_k,\,z^{\mathcal{T}}_{jk}), & l\in\mathbb{S}\\
\mathbb{E}_{W_{k\in\mathbb{S}}}\left[\mu_l\prod_{k\in\mathbb{S}}c_k(W_k,\,w^{\mathcal{T}}_{jk})\right]\prod_{k\in\mathbb{S}^\mathsf{c}}\xi_{jk}\prod_{k=1}^pc_k(z_k,\,z^{\mathcal{T}}_{jk}), & l\in\mathbb{S}^\mathsf{c}	
\end{dcases}\\
=&\begin{dcases}
\mathbb{E}_{W_l}\left[W_lc_l(W_l,\,w^{\mathcal{T}}_{jl})\right]\prod_{\substack{k\in\mathbb{S}\\k\neq l}}\mathbb{E}_{W_k}\left[c_k(W_k,\,w^{\mathcal{T}}_{jk})\right]\prod_{k\in\mathbb{S}^\mathsf{c}}\xi_{jk}\prod_{k=1}^pc_k(z_k,\,z^{\mathcal{T}}_{jk}), & l\in\mathbb{S}\\
\mu_l\prod_{k\in\mathbb{S}}\mathbb{E}_{W_k}\left[c_k(W_k,\,w^{\mathcal{T}}_{jk})\right]\prod_{k\in\mathbb{S}^\mathsf{c}}\xi_{jk}\prod_{k=1}^pc_k(z_k,\,z^{\mathcal{T}}_{jk}), & l\in\mathbb{S}^\mathsf{c}	
\end{dcases}\\
=&\begin{dcases}
\psi_{jl}\prod_{\substack{k\in\mathbb{S}\\k\neq l}}\xi_{jk}\prod_{k\in\mathbb{S}^\mathsf{c}}\xi_{jk}\prod_{k=1}^pc_k(z_k,\,z^{\mathcal{T}}_{jk}), & l\in\mathbb{S}\\
\mu_l\prod_{k\in\mathbb{S}}\xi_{jk}\prod_{k\in\mathbb{S}^\mathsf{c}}\xi_{jk}\prod_{k=1}^pc_k(z_k,\,z^{\mathcal{T}}_{jk}), & l\in\mathbb{S}^\mathsf{c}	
\end{dcases}\\
=&\begin{dcases}
\psi_{jl}\prod^d_{\substack{k=1\\k\neq l}}\xi_{jk}\prod_{k=1}^pc_k(z_k,\,z^{\mathcal{T}}_{jk}), & l\in\mathbb{S}\\
\mu_l\prod_{k=1}^d\xi_{jk}\prod_{k=1}^pc_k(z_k,\,z^{\mathcal{T}}_{jk}), & l\in\mathbb{S}^\mathsf{c};
\end{dcases}
\end{align*}

\item $\widetilde{\mathbf{J}}=\mathbb{E}_{W_{k\in\mathbb{S}}}\left[\widetilde{\mathbf{I}}\widetilde{\mathbf{I}}^\top\right]\in\mathbb{R}^{m\times m}$ with its $ij$-th element:
\begin{align}
\widetilde{J}_{ij}=&\mathbb{E}_{W_{k\in\mathbb{S}}}\left[\prod_{k\in\mathbb{S}}c_k(W_k,\,w^{\mathcal{T}}_{ik})\prod_{k\in\mathbb{S}^\mathsf{c}}\xi_{ik} \prod_{k=1}^pc_k(z_k,\,z^{\mathcal{T}}_{ik})\times \prod_{k\in\mathbb{S}}c_k(W_k,\,w^{\mathcal{T}}_{jk})\prod_{k\in\mathbb{S}^\mathsf{c}}\xi_{jk} \prod_{k=1}^pc_k(z_k,\,z^{\mathcal{T}}_{jk})\right]\nonumber\\
=&\mathbb{E}_{W_{k\in\mathbb{S}}}\left[\prod_{k\in\mathbb{S}}c_k(W_k,\,w^{\mathcal{T}}_{ik})c_k(W_k,\,w^{\mathcal{T}}_{jk})\prod_{k\in\mathbb{S}^\mathsf{c}}\xi_{ik}\xi_{jk}\prod_{k=1}^pc_k(z_k,\,z^{\mathcal{T}}_{ik})c_k(z_k,\,z^{\mathcal{T}}_{jk})\right]\nonumber\\
=&\mathbb{E}_{W_{k\in\mathbb{S}}}\left[\prod_{k\in\mathbb{S}}c_k(W_k,\,w^{\mathcal{T}}_{ik})c_k(W_k,\,w^{\mathcal{T}}_{jk})\right]\prod_{k\in\mathbb{S}^\mathsf{c}}\xi_{ik}\xi_{jk}\prod_{k=1}^pc_k(z_k,\,z^{\mathcal{T}}_{ik})c_k(z_k,\,z^{\mathcal{T}}_{jk})\nonumber\\
=&\prod_{k\in\mathbb{S}}\mathbb{E}_{W_k}\left[c_k(W_k,\,w^{\mathcal{T}}_{ik})c_k(W_k,\,w^{\mathcal{T}}_{jk})\right]\prod_{k\in\mathbb{S}^\mathsf{c}}\xi_{ik}\xi_{jk}\prod_{k=1}^pc_k(z_k,\,z^{\mathcal{T}}_{ik})c_k(z_k,\,z^{\mathcal{T}}_{jk})\nonumber\\
=&\prod_{k\in\mathbb{S}}\zeta_{ijk}\prod_{k\in\mathbb{S}^\mathsf{c}}\xi_{ik}\xi_{jk}\prod_{k=1}^pc_k(z_k,\,z^{\mathcal{T}}_{ik})c_k(z_k,\,z^{\mathcal{T}}_{jk})\nonumber.
\end{align}
\end{itemize}

We now derive~$(\mathrm{\ref{eq:v1proof}.2})$ as follow: 
\begin{align}
\label{eq:expectation2exp}
(\mathrm{\ref{eq:v1proof}.2})=&\left(\mathbb{E}_{W_{k\in\mathbb{S}}}\left[\widetilde{\boldsymbol{\mu}}^\top\right]\widehat{\boldsymbol{\theta}}+\mathbb{E}_{W_{k\in\mathbb{S}}}\left[\widetilde{\mathbf{I}}^\top\right]\mathbf{A}\right)^2\nonumber\\
=&\left(\boldsymbol{\mu}^\top\widehat{\boldsymbol{\theta}}+\mathbf{I}^\top\mathbf{A}\right)^2\nonumber\\
=&\boldsymbol{\mu}^\top\widehat{\boldsymbol{\theta}}\widehat{\boldsymbol{\theta}}^\top\boldsymbol{\mu}+\mathbf{A}^\top\mathbf{I}\mathbf{I}^\top\mathbf{A}+2\widehat{\boldsymbol{\theta}}^\top\boldsymbol{\mu}\mathbf{I}^\top\mathbf{A}.
\end{align}

Plugging equations~\eqref{eq:expectation1exp} and~\eqref{eq:expectation2exp} back into equation~\eqref{eq:v1proof}, we obtain
\begin{align*}
V_1(\mathbb{S})=&\mathrm{tr}\left\{\widehat{\boldsymbol{\theta}}\widehat{\boldsymbol{\theta}}^\top\left(\boldsymbol{\mu}\boldsymbol{\mu}^\top+\widetilde{\boldsymbol{\Omega}}\right)\right\}+\mathrm{tr}\left\{\mathbf{A}\mathbf{A}^\top\widetilde{\mathbf{J}}\right\}+2\widehat{\boldsymbol{\theta}}^\top\widetilde{\mathbf{B}}\mathbf{A}-\left(\boldsymbol{\mu}^\top\widehat{\boldsymbol{\theta}}\widehat{\boldsymbol{\theta}}^\top\boldsymbol{\mu}+\mathbf{A}^\top\mathbf{I}\mathbf{I}^\top\mathbf{A}+2\widehat{\boldsymbol{\theta}}^\top\boldsymbol{\mu}\mathbf{I}^\top\mathbf{A}\right)\\
=&\mathrm{tr}\left\{\widehat{\boldsymbol{\theta}}\widehat{\boldsymbol{\theta}}^\top\boldsymbol{\mu}\boldsymbol{\mu}\right\}+\mathrm{tr}\left\{\widehat{\boldsymbol{\theta}}\widehat{\boldsymbol{\theta}}^\top\widetilde{\boldsymbol{\Omega}}\right\}+\mathbf{A}^\top\widetilde{\mathbf{J}}\mathbf{A}+2\widehat{\boldsymbol{\theta}}^\top\widetilde{\mathbf{B}}\mathbf{A}-\boldsymbol{\mu}^\top\widehat{\boldsymbol{\theta}}\widehat{\boldsymbol{\theta}}^\top\boldsymbol{\mu}-\mathbf{A}^\top\mathbf{I}\mathbf{I}^\top\mathbf{A}-2\widehat{\boldsymbol{\theta}}^\top\boldsymbol{\mu}\mathbf{I}^\top\mathbf{A}\\
=&\mathrm{tr}\left\{\widehat{\boldsymbol{\theta}}\widehat{\boldsymbol{\theta}}^\top\boldsymbol{\mu}\boldsymbol{\mu}\right\}+\mathrm{tr}\left\{\widehat{\boldsymbol{\theta}}\widehat{\boldsymbol{\theta}}^\top\widetilde{\boldsymbol{\Omega}}\right\}+\mathbf{A}^\top\widetilde{\mathbf{J}}\mathbf{A}+2\widehat{\boldsymbol{\theta}}^\top\widetilde{\mathbf{B}}\mathbf{A}-\mathrm{tr}\left\{\widehat{\boldsymbol{\theta}}\widehat{\boldsymbol{\theta}}^\top\boldsymbol{\mu}\boldsymbol{\mu}\right\}-\mathbf{A}^\top\mathbf{I}\mathbf{I}^\top\mathbf{A}-2\widehat{\boldsymbol{\theta}}^\top\boldsymbol{\mu}\mathbf{I}^\top\mathbf{A}\\
=&\mathrm{tr}\left\{\widehat{\boldsymbol{\theta}}\widehat{\boldsymbol{\theta}}^\top\widetilde{\boldsymbol{\Omega}}\right\}+\mathbf{A}^\top\left(\widetilde{\mathbf{J}}-\mathbf{I}\mathbf{I}^\top\right)\mathbf{A}+2\widehat{\boldsymbol{\theta}}^\top\left(\widetilde{\mathbf{B}}-\boldsymbol{\mu}\mathbf{I}^\top\right)\mathbf{A}.
\end{align*}

In case that the trend is assumed constant, $V_1(\mathbb{S})$ can be simplified to the following expression:
\begin{equation*}
V_1(\mathbb{S})=\mathbf{A}^\top\left(\widetilde{\mathbf{J}}-\mathbf{I}\mathbf{I}^\top\right)\mathbf{A}.	
\end{equation*}

\section[Proof of Theorem~S.2.1]{Proof of Theorem~\ref{thm:fullcov}}
\label{sec:prooffullcov}

\subsection{\texorpdfstring{{Derivation of $\widetilde{\xi}_{i}$}}{Derivation of the first expectation}}
\begin{align*}
\widetilde{\xi}_{i}&=\mathbb{E}\left[c(\mathbf{W},\,\mathbf{w}^{\mathcal{T}}_i)\right]\\
&=\int\exp\left\{-\sum_{k=1}^{d}\frac{\left(w_k-w^{\mathcal{T}}_{ik}\right)^2}{\gamma_k^2}\right\}\frac{1}{\sqrt{(2\pi)^d|\boldsymbol\Sigma|}}\exp\left\{-\frac{1}{2}(\mathbf{w}-\boldsymbol{\mu})^{\top}\boldsymbol{\Sigma}^{-1}(\mathbf{w}-\boldsymbol{\mu})\right\}\mathrm{d}\mathbf{w}\\
 &=\int\exp\left\{-\frac{1}{2}(\mathbf{w}-\boldsymbol{\omega}^{\mathcal{T}}_i)^{\top}\boldsymbol{\Lambda}^{-1}(\mathbf{w}-\boldsymbol{\omega}^{\mathcal{T}}_i)\right\}\frac{1}{\sqrt{(2\pi)^d|\boldsymbol\Sigma|}}\exp\left\{-\frac{1}{2}(\mathbf{w}-\boldsymbol{\mu})^{\top}\boldsymbol{\Sigma}^{-1}(\mathbf{w}-\boldsymbol{\mu})\right\}\mathrm{d}\mathbf{w},
\end{align*}
where $\boldsymbol{\Lambda}=\mathrm{diag}(\frac{\gamma_1^2}{2},\dots,\frac{\gamma_d^2}{2})\in\mathbb{R}^{d\times d}$ is a diagonal matrix. 

By completing in squares, we then have
\begin{align*}
\widetilde{\xi}_{i}&=\frac{1}{\sqrt{(2\pi)^d|\mathbf{M}^{-1}|}}\frac{1}{\sqrt{|\boldsymbol\Sigma\mathbf{M}|}}\\  
&\quad\times\int\exp\left\{-\frac{1}{2}(\mathbf{w}-\mathbf{M}^{-1}\mathbf{V})^\top\mathbf{M}(\mathbf{w}-\mathbf{M}^{-1}\mathbf{V})+\frac{1}{2}(\mathbf{V}^\top\mathbf{M}^{-1}\mathbf{V}-R)\right\}\mathrm{d}\mathbf{w},
\end{align*}
where $\mathbf{M}=\boldsymbol{\Sigma}^{-1}+\boldsymbol{\Lambda}^{-1}$, $\mathbf{V}=\boldsymbol{\Sigma}^{-1}\boldsymbol{\mu}+\boldsymbol{\Lambda}^{-1}\boldsymbol{\omega}^{\mathcal{T}}_i$ and $R=\boldsymbol{\mu}^\top\boldsymbol{\Sigma}^{-1}\boldsymbol{\mu}+(\boldsymbol{\omega}^{\mathcal{T}}_i)^\top\boldsymbol{\Lambda}^{-1}\boldsymbol{\omega}^{\mathcal{T}}_i$.

By integrating out the probability density function of a multivariate normal distribution with mean $\mathbf{M}^{-1}\mathbf{V}$ and covariance matrix $\mathbf{M}^{-1}$, we have
\begin{equation*}
\widetilde{\xi}_{i}=\frac{1}{\sqrt{|\boldsymbol\Sigma\mathbf{M}|}}\exp\left\{\frac{1}{2}(\mathbf{V}^\top\mathbf{M}^{-1}\mathbf{V}-R)\right\}
\end{equation*}

Using the Woodbury identity~\cite{Petersen2012}, we have
\begin{align*}
   \mathbf{M}^{-1}&=\boldsymbol{\Sigma}-\boldsymbol{\Sigma}(\boldsymbol{\Sigma}+\boldsymbol{\Lambda})^{-1}\boldsymbol{\Sigma}\\
   \mathbf{M}^{-1}&=\boldsymbol{\Lambda}-\boldsymbol{\Lambda}(\boldsymbol{\Sigma}+\boldsymbol{\Lambda})^{-1}\boldsymbol{\Lambda}.
\end{align*}
Thus, we obtain
\begin{equation*}
\widetilde{\xi}_{i}=\frac{1}{\sqrt{|(\boldsymbol\Lambda+\boldsymbol\Sigma)\boldsymbol\Lambda^{-1}|}}\exp\left\{-\frac{1}{2}(\boldsymbol{\omega}^{\mathcal{T}}_i-\boldsymbol{\mu})^\top(\boldsymbol{\Lambda}+\boldsymbol{\Sigma})^{-1}(\boldsymbol{\omega}^{\mathcal{T}}_i-\boldsymbol{\mu})\right\},
\end{equation*}

\subsection{\texorpdfstring{{Derivation of $\widetilde{\zeta}_{ij}$}}{Derivation of the second expectation}}

\begin{align*}
  \widetilde{\zeta}_{ij}&=\mathbb{E}\left[c(\mathbf{W},\,\mathbf{w}^{\mathcal{T}}_i)c(\mathbf{W},\,\mathbf{w}^{\mathcal{T}}_j)\right]\\
  &=\int\exp\left\{-\sum_{k=1}^{d}\frac{\left(w_k-w^{\mathcal{T}}_{ik}\right)^2}{\gamma_k^2}-\sum_{k=1}^{d}\frac{\left(w_k-w^{\mathcal{T}}_{jk}\right)^2}{\gamma_k^2}\right\}\\
  &\quad\times\frac{1}{\sqrt{(2\pi)^d|\boldsymbol\Sigma|}}\exp\left\{-\frac{1}{2}(\mathbf{w}-\boldsymbol{\mu})^{\top}\boldsymbol{\Sigma}^{-1}(\mathbf{w}-\boldsymbol{\mu})\right\}\mathrm{d}\mathbf{w}\\
  &=\int\exp\left\{-\sum_{k=1}^{d}\frac{2(w_k-w^{\mathcal{T}}_{ik})(w_k-w^{\mathcal{T}}_{jk})}{\gamma_k^2}-\sum_{k=1}^{d}\frac{\left(w^{\mathcal{T}}_{ik}-w^{\mathcal{T}}_{jk}\right)^2}{\gamma_k^2}\right\}\\
  &\quad\times\frac{1}{\sqrt{(2\pi)^d|\boldsymbol\Sigma|}}\exp\left\{-\frac{1}{2}(\mathbf{w}-\boldsymbol{\mu})^{\top}\boldsymbol{\Sigma}^{-1}(\mathbf{w}-\boldsymbol{\mu})\right\}\mathrm{d}\mathbf{w}\\
  &=\int\exp\left\{-\frac{1}{2}(\mathbf{w}-\boldsymbol{\omega}^{\mathcal{T}}_i)^{\top}\boldsymbol{\Gamma}^{-1}(\mathbf{w}-\boldsymbol{\omega}^{\mathcal{T}}_j)-\frac{1}{4}(\boldsymbol{\omega}^{\mathcal{T}}_i-\boldsymbol{\omega}^{\mathcal{T}}_j)^\top\boldsymbol{\Gamma}^{-1}(\boldsymbol{\omega}^{\mathcal{T}}_i-\boldsymbol{\omega}^{\mathcal{T}}_j)\right\}\\
  &\quad\times\frac{1}{\sqrt{(2\pi)^d|\boldsymbol\Sigma|}}\exp\left\{-\frac{1}{2}(\mathbf{w}-\boldsymbol{\mu})^{\top}\boldsymbol{\Sigma}^{-1}(\mathbf{w}-\boldsymbol{\mu})\right\}\mathrm{d}\mathbf{w}\\
  &=\exp\left\{-\frac{1}{4}(\boldsymbol{\omega}^{\mathcal{T}}_i-\boldsymbol{\omega}^{\mathcal{T}}_j)^\top\boldsymbol{\Gamma}^{-1}(\boldsymbol{\omega}^{\mathcal{T}}_i-\boldsymbol{\omega}^{\mathcal{T}}_j)\right\}\frac{1}{\sqrt{(2\pi)^d|\boldsymbol\Sigma|}}\\
  &\quad\times\int\exp\left\{-\frac{1}{2}\left[(\mathbf{w}-\boldsymbol{\omega}^{\mathcal{T}}_i)^{\top}\boldsymbol{\Gamma}^{-1}(\mathbf{w}-\boldsymbol{\omega}^{\mathcal{T}}_j)+(\mathbf{w}-\boldsymbol{\mu})^{\top}\boldsymbol{\Sigma}^{-1}(\mathbf{w}-\boldsymbol{\mu})\right]\right\}\mathrm{d}\mathbf{w},
\end{align*}
where $\boldsymbol{\Gamma}=\mathrm{diag}(\frac{\gamma_1^2}{4},\dots,\frac{\gamma_d^2}{4})\in\mathbb{R}^{d\times d}$ is a diagonal matrix. By completing in squares, we then have
\begin{align*}
\widetilde{\zeta}_{ij}&=\exp\left\{-\frac{1}{4}(\boldsymbol{\omega}^{\mathcal{T}}_i-\boldsymbol{\omega}^{\mathcal{T}}_j)^\top\boldsymbol{\Gamma}^{-1}(\boldsymbol{\omega}^{\mathcal{T}}_i-\boldsymbol{\omega}^{\mathcal{T}}_j)\right\}\frac{1}{\sqrt{(2\pi)^d|\mathbf{M}^{-1}|}}\frac{1}{\sqrt{|\boldsymbol\Sigma\mathbf{M}|}}\\  
&\quad\times\int\exp\left\{-\frac{1}{2}(\mathbf{w}-\mathbf{M}^{-1}\mathbf{V})^\top\mathbf{M}(\mathbf{w}-\mathbf{M}^{-1}\mathbf{V})+\frac{1}{2}(\mathbf{V}^\top\mathbf{M}^{-1}\mathbf{V}-R)\right\}\mathrm{d}\mathbf{w},
\end{align*}
where $\mathbf{M}=\boldsymbol{\Sigma}^{-1}+\boldsymbol{\Gamma}^{-1}$; $\mathbf{V}=\boldsymbol{\Sigma}^{-1}\boldsymbol{\mu}+\boldsymbol{\Gamma}^{-1}\boldsymbol{\omega}$ with $\boldsymbol{\omega}=\frac{1}{2}(\boldsymbol{\omega}^{\mathcal{T}}_i+\boldsymbol{\omega}^{\mathcal{T}}_j)$; and $R=\boldsymbol{\mu}^\top\boldsymbol{\Sigma}^{-1}\boldsymbol{\mu}+(\boldsymbol{\omega}^{\mathcal{T}}_i)^\top\boldsymbol{\Gamma}^{-1}\boldsymbol{\omega}^{\mathcal{T}}_j$.

By integrating out the probability density function of a multivariate normal distribution with mean $\mathbf{M}^{-1}\mathbf{V}$ and covariance matrix $\mathbf{M}^{-1}$, we have
\begin{equation*}
\widetilde{\zeta}_{ij}=\exp\left\{-\frac{1}{4}(\boldsymbol{\omega}^{\mathcal{T}}_i-\boldsymbol{\omega}^{\mathcal{T}}_j)^\top\boldsymbol{\Gamma}^{-1}(\boldsymbol{\omega}^{\mathcal{T}}_i-\boldsymbol{\omega}^{\mathcal{T}}_j)\right\}\frac{1}{\sqrt{|\boldsymbol\Sigma\mathbf{M}|}}\exp\left\{\frac{1}{2}(\mathbf{V}^\top\mathbf{M}^{-1}\mathbf{V}-R)\right\}.
\end{equation*}
Using the Woodbury identity~\cite{Petersen2012}, we have
\begin{align*}
   \mathbf{M}^{-1}&=\boldsymbol{\Sigma}-\boldsymbol{\Sigma}(\boldsymbol{\Sigma}+\boldsymbol{\Gamma})^{-1}\boldsymbol{\Sigma}\\
   \mathbf{M}^{-1}&=\boldsymbol{\Gamma}-\boldsymbol{\Gamma}(\boldsymbol{\Sigma}+\boldsymbol{\Gamma})^{-1}\boldsymbol{\Gamma}.
\end{align*}
Thus, we obtain
\begin{equation*}
\widetilde{\zeta}_{ij}=\exp\left\{-\frac{1}{8}(\boldsymbol{\omega}^{\mathcal{T}}_i-\boldsymbol{\omega}^{\mathcal{T}}_j)^\top\boldsymbol{\Gamma}^{-1}(\boldsymbol{\omega}^{\mathcal{T}}_i-\boldsymbol{\omega}^{\mathcal{T}}_j)\right\}\frac{1}{\sqrt{|(\boldsymbol\Gamma+\boldsymbol\Sigma)\boldsymbol\Gamma^{-1}|}}\exp\left\{-\frac{1}{2}(\boldsymbol{\omega}-\boldsymbol{\mu})^\top(\boldsymbol{\Gamma}+\boldsymbol{\Sigma})^{-1}(\boldsymbol{\omega}-\boldsymbol{\mu})\right\}.
\end{equation*}

\subsection{\texorpdfstring{{Derivation of $\widetilde{\psi}_{jl}$}}{Derivation of the third expectation}}
\begin{align*}
\widetilde{\psi}_{jl}&=\mathbb{E}\left[W_l c(\mathbf{W},\,\mathbf{w}^{\mathcal{T}}_j)\right]\\
&=\int w_l\exp\left\{-\sum_{k=1}^{d}\frac{\left(w_k-w^{\mathcal{T}}_{jk}\right)^2}{\gamma_k^2}\right\}\frac{1}{\sqrt{(2\pi)^d|\boldsymbol\Sigma|}}\exp\left\{-\frac{1}{2}(\mathbf{w}-\boldsymbol{\mu})^{\top}\boldsymbol{\Sigma}^{-1}(\mathbf{w}-\boldsymbol{\mu})\right\}\mathrm{d}\mathbf{w}\\
 &=\int w_l\exp\left\{-\frac{1}{2}(\mathbf{w}-\boldsymbol{\omega}^{\mathcal{T}}_i)^{\top}\boldsymbol{\Lambda}^{-1}(\mathbf{w}-\boldsymbol{\omega}^{\mathcal{T}}_i)\right\}\frac{1}{\sqrt{(2\pi)^d|\boldsymbol\Sigma|}}\exp\left\{-\frac{1}{2}(\mathbf{w}-\boldsymbol{\mu})^{\top}\boldsymbol{\Sigma}^{-1}(\mathbf{w}-\boldsymbol{\mu})\right\}\mathrm{d}\mathbf{w},
\end{align*}
where $\boldsymbol{\Lambda}=\mathrm{diag}(\frac{\gamma_1^2}{2},\dots,\frac{\gamma_d^2}{2})\in\mathbb{R}^{d\times d}$ is a diagonal matrix. 

By completing in squares, we then have
\begin{align*}
\widetilde{\psi}_{jl}&=\frac{1}{\sqrt{(2\pi)^d|\mathbf{M}^{-1}|}}\frac{1}{\sqrt{|\boldsymbol\Sigma\mathbf{M}|}}\\  
&\times\int w_l\exp\left\{-\frac{1}{2}(\mathbf{w}-\mathbf{M}^{-1}\mathbf{V})^\top\mathbf{M}(\mathbf{w}-\mathbf{M}^{-1}\mathbf{V})+\frac{1}{2}(\mathbf{V}^\top\mathbf{M}^{-1}\mathbf{V}-R)\right\}\mathrm{d}\mathbf{w},
\end{align*}
where $\mathbf{M}=\boldsymbol{\Sigma}^{-1}+\boldsymbol{\Lambda}^{-1}$, $\mathbf{V}=\boldsymbol{\Sigma}^{-1}\boldsymbol{\mu}+\boldsymbol{\Lambda}^{-1}\boldsymbol{\omega}^{\mathcal{T}}_j$ and $R=\boldsymbol{\mu}^\top\boldsymbol{\Sigma}^{-1}\boldsymbol{\mu}+(\boldsymbol{\omega}^{\mathcal{T}}_j)^\top\boldsymbol{\Lambda}^{-1}\boldsymbol{\omega}^{\mathcal{T}}_j$.

By integrating out $w_l$ with respect to the probability density function of a multivariate normal distribution with mean $\mathbf{M}^{-1}\mathbf{V}$ and covariance matrix $\mathbf{M}^{-1}$, we have
\begin{equation*}
\widetilde{\psi}_{jl}=\frac{\mathbf{e}_l\mathbf{M}^{-1}\mathbf{V}}{\sqrt{|\boldsymbol\Sigma\mathbf{M}|}}\exp\left\{\frac{1}{2}(\mathbf{V}^\top\mathbf{M}^{-1}\mathbf{V}-R)\right\},
\end{equation*}
where $\mathbf{e}_l$ is a unit row vector with $l$-th element being one. 

Using the Woodbury identity~\citep{Petersen2012}, we have
\begin{align*}
   \mathbf{M}^{-1}&=\boldsymbol{\Sigma}-\boldsymbol{\Sigma}(\boldsymbol{\Sigma}+\boldsymbol{\Lambda})^{-1}\boldsymbol{\Sigma}\\
   \mathbf{M}^{-1}&=\boldsymbol{\Lambda}-\boldsymbol{\Lambda}(\boldsymbol{\Sigma}+\boldsymbol{\Lambda})^{-1}\boldsymbol{\Lambda}.
\end{align*}
Thus, we obtain
\begin{equation*}
\widetilde{\psi}_{jl}=\frac{\mathbf{e}_l[\boldsymbol{\Lambda}(\boldsymbol{\Lambda}+\boldsymbol{\Sigma})^{-1}\boldsymbol{\mu}+\boldsymbol{\Sigma}(\boldsymbol{\Lambda}+\boldsymbol{\Sigma})^{-1}\boldsymbol{\omega}^{\mathcal{T}}_j]}{\sqrt{|(\boldsymbol\Lambda+\boldsymbol\Sigma)\boldsymbol\Lambda^{-1}|}}\exp\left\{-\frac{1}{2}(\boldsymbol{\omega}^{\mathcal{T}}_j-\boldsymbol{\mu})^\top(\boldsymbol{\Lambda}+\boldsymbol{\Sigma})^{-1}(\boldsymbol{\omega}^{\mathcal{T}}_j-\boldsymbol{\mu})\right\},
\end{equation*}
which is
\begin{equation*}
\widetilde{\psi}_{jl}=\mathbf{e}_l[\boldsymbol{\Lambda}(\boldsymbol{\Lambda}+\boldsymbol{\Sigma})^{-1}\boldsymbol{\mu}+\boldsymbol{\Sigma}(\boldsymbol{\Lambda}+\boldsymbol{\Sigma})^{-1}\boldsymbol{\omega}^{\mathcal{T}}_j]\,\widetilde{\xi}_{j}.  
\end{equation*}

\end{document}